\documentclass[11pt,reqno]{article}
\pdfoutput=1

\usepackage[ascii]{inputenc}
\usepackage[bookmarksnumbered,hypertexnames=false,colorlinks=true,linkcolor=blue,urlcolor=blue,citecolor=blue,anchorcolor=green,breaklinks=true,pdfusetitle]{hyperref}
\usepackage{bbm,braket,microtype,mathrsfs,amsmath,amssymb,xcolor,amsthm,mathtools,fullpage,graphicx,enumitem,caption,booktabs,bm,xspace}
\usepackage[longend,ruled,linesnumbered,nokwfunc]{algorithm2e}

\usepackage{footnote}
\makesavenoteenv{tabular}
\makesavenoteenv{table}
\usepackage{multirow}
\usepackage[T1]{fontenc}
\usepackage[english]{babel}
\usepackage[capitalize]{cleveref}
\usepackage{thm-restate}
\newtheorem{thm}{Theorem}[section]\crefname{thm}{Theorem}{Theorems}
\newtheorem{lem}[thm]{Lemma}\crefname{lem}{Lemma}{Lemmas}
\crefname{claim}{Claim}{Claims}
\newtheorem{definition}[thm]{Definition}\crefname{defn}{Definition}{Definitions}
\newtheorem{prp}[thm]{Proposition}\crefname{prop}{Proposition}{Propositions}
\newtheorem{prb}[thm]{Problem}\crefname{prb}{Problem}{Problems}
\newtheorem{rem}[thm]{Remark}\crefname{rem}{Remark}{Remarks}
\newtheorem{cor}[thm]{Corollary}\crefname{cor}{Corollary}{Corollaries}
\numberwithin{equation}{section}

\DeclareMathOperator{\poly}{poly}
\DeclareMathOperator{\polylog}{polylog}
\DeclareMathOperator{\diag}{diag}

\DeclareMathOperator{\conv}{conv}

\renewcommand{\vec}[1]{\bm{#1}}
\newcommand{\mat}[1]{\bm{#1}}
\newcommand{\A}[2]{\mat A(\vec{#1}, \vec{#2})}

\newcommand{\RR}{\mathbb R}
\newcommand{\N}{\mathbb N}
\newcommand{\R}{\mathbb R}

\newcommand{\Exp}{\mathbb E}

\newcommand{\indic}{\mathbbm{1}}
\newcommand{\grad}{\nabla}
\newcommand{\eps}{\varepsilon}
\renewcommand{\epsilon}{\varepsilon}
\newcommand{\bigO}[1]{O\left(#1\right)}
\renewcommand{\i}{\mathbf{i}\,}

\usepackage{xparse}
\NewDocumentCommand{\bigOt}{o m}{%
  \IfNoValueTF{#1}
    {\ensuremath{\widetilde{O}\left(#2\right)}}
    {\ensuremath{\widetilde{O}_{#1}\left(#2\right)}}%
}

\newcommand{\approxA}{\gamma}
\renewcommand{\vec}[1]{\boldsymbol{\mathbf{#1}}}
\DeclarePairedDelimiter{\abs}{\lvert}{\rvert}
\DeclarePairedDelimiter{\norm}{\lVert}{\rVert}
\DeclarePairedDelimiterX{\ip}[2]{\langle}{\rangle}{#1,#2}
\allowdisplaybreaks[4]

\newcommand{\Step}{\Delta}
\newcommand{\prevrowmarginal}{\ensuremath{\phi_t}}
\newcommand{\prevcolmarginal}{\ensuremath{\psi_t}}

\crefalias{AlgoLine}{line}
\crefname{line}{line}{lines}

\SetAlgoProcName{Subroutine}{Subroutine}

\SetKwInput{Input}{Input}
\SetKwInput{Output}{Output}
\SetKwInput{Guarantee}{Guarantee}
\SetKw{True}{True}
\SetKw{False}{False}
\newcommand{\ApproxScalingFactor}{\hyperref[subroutine:ApproxScalingFactor spec]{\texttt{\ref*{subroutine:ApproxScalingFactor spec}}}\xspace}
\newcommand{\TestScaling}{\hyperref[subroutine:TestScaling spec]{\texttt{\ref*{subroutine:TestScaling spec}}}\xspace}
\newcommand{\QuantumApproximateSum}{\hyperref[subroutine:QAC]{\texttt{QuantumApproximateSum}}\xspace}
\newcommand{\GreaterOrEqual}{\hyperref[alg:GreaterOrEqual]{\texttt{GreaterOrEqual}}\xspace}
\newcommand{\RelativeEntryAdditiveApprox}{\hyperref[alg:RelativeEntryAdditiveApprox]{\texttt{RelativeEntryAdditiveApprox}}\xspace}

\makeatletter

\newcommand*\wthelper[2]{%
        \hbox{\dimen@\accentfontxheight#1%
                \accentfontxheight#11.3\dimen@
                $\m@th#1\widetilde{#2}$%
                \accentfontxheight#1\dimen@
        }%
}

\newcommand*\accentfontxheight[1]{%
        \fontdimen5\ifx#1\displaystyle
                \textfont
        \else\ifx#1\textstyle
                \textfont
        \else\ifx#1\scriptstyle
                \scriptfont
        \else
                \scriptscriptfont
        \fi\fi\fi3
}
\makeatother

\newcommand{\ApproxScalingFactorIO}{\Input{Oracle access to rational~$\vec a \in [0,1]^n$, rational $r \in (0, 1]$, oracle access to~$\vec y \in \R^n$ encoded in~$(b_1, b_2)$ fixed-point format, desired precision~$\delta \in (0, 1]$, desired failure probability $\eta \in [0,1]$, lower bound~$\mu > 0$ on non-zero entries of $\vec a$}\Output{A number~$x$ encoded in~$(b_1, b_2)$ fixed-point format}}

\newcommand{\email}[1]{\href{mailto:#1}{\texttt{#1}}}

\def\lemgeneralizedpinskercontent{
  Let $\vec a, \vec b \in \vec \R_+^n$ and assume $\norm{\vec a}_1 = 1$ and $b_\ell > 0$ for all $\ell \in [n]$.
  Define the function $w\colon (-1, \infty) \to \R$ by $w(\beta) = \beta - \ln(1 + \beta)$.
  Then
  \begin{align*}
    D(\vec a \Vert \vec b) \geq w(\norm{\vec a - \vec b}_1).
  \end{align*}
  For $\beta \in [0, 1]$, we have the estimate $w(\beta) \geq \beta^2 / 4$, while for $\beta \geq 1$ we have $w(\beta) \geq (1 - \ln 2) \beta$.
  In particular, if $\norm{\vec a - \vec b}_1 \leq 1$, then $D(\vec a \Vert \vec b) \geq \norm{\vec a - \vec b}_1^2 / 4$.
}

\title{Quantum algorithms for matrix scaling and matrix balancing}
\date{}
\author{Joran van Apeldoorn\thanks{Institute for Information Law and QuSoft, University of Amsterdam, the Netherlands. Supported by the Dutch Research Council (NWO/OCW), as part of the Quantum Software Consortium programme (project number 024.003.037). \email{work@bitofbytes.com}} \and
Sander Gribling\thanks{IRIF, Universit\'e de Paris, CNRS, Paris, France. Partially supported by SIRTEQ-grant QuIPP. \email{gribling@irif.fr}} \and
Yinan Li\thanks{Graduate School of Mathematics, Nagoya University, Japan. Partially supported by MEXT Quantum Leap Flagship Program (MEXT Q-LEAP) Grant Number JPMXS0120319794.
\email{Yinan.Li@math.nagoya-u.ac.jp}} \and
Harold Nieuwboer\thanks{Korteweg--de Vries Institute for Mathematics and QuSoft, University of Amsterdam. \email{h.a.nieuwboer@uva.nl}} \and
Michael Walter\thanks{Korteweg--de Vries Institute for Mathematics, Institute for Theoretical Physics, Institute for Language, Logic, and Computation, and QuSoft, University of Amsterdam. Partially supported by the Dutch Research Council (NWO) through Veni grant no.~680-47-459. \email{m.walter@uva.nl}} \and
Ronald de Wolf\thanks{QuSoft, CWI and University of Amsterdam, the Netherlands. Partially supported by the Dutch Research Council (NWO/OCW), as part of the Quantum Software Consortium programme (project number 024.003.037), and through QuantERA ERA-NET Cofund project QuantAlgo (680-91-034). \email{rdewolf@cwi.nl}}
}

\begin{document}

\maketitle

\begin{abstract}
    Matrix scaling and matrix balancing are two basic linear-algebraic problems with a wide variety of applications, such as approximating the permanent, and pre-conditioning linear systems to make them more numerically stable.
    We study the power and limitations of quantum algorithms for these problems.

    We provide quantum implementations of two classical (in both senses of the word) methods: Sinkhorn's algorithm for matrix scaling and Osborne's algorithm for matrix balancing.
    Using amplitude estimation as our main tool, our quantum implementations
    both run in time $\widetilde O(\sqrt{mn}/\eps^4)$ for scaling or balancing an $n \times n$ matrix (given by an oracle) with $m$ non-zero entries to within $\ell_1$-error~$\eps$.
    Their classical analogs use time~$\widetilde O(m/\eps^2)$, and every classical algorithm for scaling or balancing with small constant $\eps$ requires $\Omega(m)$ queries to the entries of the input matrix.
    We thus achieve a polynomial speed-up in terms of~$n$, at the expense of a worse polynomial dependence on the obtained $\ell_1$-error~$\eps$.
    We emphasize that even for constant $\eps$ these problems are already non-trivial (and relevant in applications).

    Along the way, we extend the classical analysis of Sinkhorn's and Osborne's algorithm to allow for errors in the computation of marginals.
    We also adapt an improved analysis of Sinkhorn's algorithm for entrywise-positive matrices to the $\ell_1$-setting, leading to an $\widetilde O(n^{1.5}/\eps^3)$-time quantum algorithm for $\eps$-$\ell_1$-scaling in this case.

    We also prove a lower bound, showing that our quantum algorithm for matrix scaling is essentially optimal for constant $\eps$: every quantum algorithm for matrix scaling that achieves a constant $\ell_1$-error with respect to uniform marginals needs to make at least $\Omega(\sqrt{mn})$ queries.
\end{abstract}

\clearpage

\setcounter{tocdepth}{2}
\tableofcontents

\clearpage

\section{Introduction}\label{sec:intro}

\subsection{Matrix scaling and matrix balancing}
Matrix scaling is a basic linear-algebraic problem with many applications. A \emph{scaling} of an~$n\times n$ matrix~$\mat A$ with non-negative entries is a matrix  ${\mat B}={\mat{XAY}}$ where $\mat X$ and $\mat Y$ are positive diagonal matrices.\footnote{We assume $\mat A$ is square for simplicity, but everything can straightforwardly be extended to non-square matrices.} In other words, we multiply the $i$-th row with $X_{ii}$ and the $j$-th column with $Y_{jj}$. We say~$\mat A$ is \emph{exactly scalable} to marginals $\vec r\in\R_+^n$ and $\vec c\in\R_+^n$ if there exist $\mat X$ and $\mat Y$ such that the vector ${\vec r}(\mat B)=(\sum_{j=1}^n B_{ij})_{i\in[n]}$ of row sums of the scaled matrix ${\mat B}$ equals $\vec r$, and its vector ${\vec c}(\mat B)$ of column sums equals~$\vec c$.
One typical example would be if $\mat r$ and $\mat c$ are the all-$1$ vectors, which means we want $\mat B$ to be doubly stochastic: the rows and columns of $\mat B$ would be probability distributions.

In many cases it suffices to find \emph{approximate} scalings. Different applications use different notions of approximation. We could for instance require ${\vec r}({\mat B})$ to be $\eps$-close to ${\vec r}$ in $\ell_1$- or $\ell_2$-norm, or in relative entropy (Kullback-Leibler divergence), for some parameter $\eps$ of our choice, and similarly require ${\vec c}({\mat B})$ to be close to ${\vec c}$.

A related problem is \emph{matrix balancing}. Here we do not prescribe desired marginals, but the goal is to find a diagonal $\mat X$ such that the row and column marginals of ${\mat B}={\mat{XAX^{-1}}}$ are close \emph{to each other}. Again, different notions of closeness ${\vec r}(\mat B)\approx{\vec c}(\mat B)$ are possible.

An important application, used in theory as well as in practical linear-algebra software (e.g.~LAPACK~\cite{lapack} and MATLAB~\cite{matlab-balance}), is in improving the numerical stability of linear-system solving. Suppose we are given matrix $\mat A$ and vector $\vec b$, and we want to find a solution to the linear system ${\mat A}{\vec v}={\vec b}$. Note that $\vec v$ is a solution iff ${\mat B}{\vec v'}={\vec b'}$ for ${\vec v'}={\mat X}{\vec v}$ and ${\vec b'}={\mat X}{\vec b}$. An appropriately balanced matrix $\mat B$ will typically be more numerically stable than the original $\mat A$, so solving the linear system ${\mat B}{\vec v'}={\vec b'}$ and then computing ${\vec v}={\mat X}^{-1} {\vec v'}$, is often a better way to solve the linear system ${\mat A}{\vec v}={\vec b}$ than directly computing ${\mat A^{-1}}{\vec b}$.

Matrix scaling and balancing have surprisingly many and wide-ranging applications. Matrix scaling was introduced by Kruithof for Dutch telephone traffic computation~\cite{kruithof:telefoon}, and has also been used in other areas of economics~\cite{stone:socialaccounting}.
In theoretical computer science it has been used for instance to approximate the permanent of a given matrix~\cite{lsw00}, as a tool to get lower bounds on unbounded-error communication complexity~\cite{forster:probcc}, and for approximating optimal transport distances~\cite{altschuler2017nearlinear}.
In mathematics, it has been used as a common tool in practical linear algebra computations~\cite{livne&golub:scaling,bradley:phd,pock&chambolle,COPB:splitting}, but also in statistics~\cite{sinkhorn:scaling}, optimization~\cite{rothblum&schneider:scaling}, and for strengthening the Sylvester-Gallai theorem~\cite{BDYW11}.
Matrix balancing has a similarly wide variety of applications, including pre-conditioning to make practical matrix computations more stable (as mentioned above), and approximating the min-mean-cycle in a weighted graph~\cite{altschuler2020minmeancycle}.
Many more applications of matrix scaling and balancing are mentioned in~\cite{lsw00,idel2016review,gargoliveira:recent}. Related scaling problems have applications to algorithmic non-commutative algebra~\cite{garg2019operator,burgisser2019towards}, functional analysis~\cite{garg2017algorithmic}, and quantum information~\cite{gurvits2004classical,burgisser2018alternating,burgisser2018efficient}.

\subsection{Known (classical) algorithms}

Given the importance of good matrix scalings and balancings, how efficiently can we actually find them? For concreteness, let us first focus on scaling. Note that left-multiplying $\mat A$ with a diagonal matrix $\mat X$ corresponds to multiplying the $i$-th row of $\mat A$ with $X_{ii}$. Hence it is very easy to get the desired row sums:
just compute all row sums $r_i({\mat A})$ of $\mat A$ and define $\mat X$ by $X_{ii}=r_i/r_i({\mat A})$, then~$\mat{XA}$ has exactly the right row sums. Subsequently, it is easy to get the desired column sums: just right-multiply the current matrix $\mat{XA}$ with diagonal matrix $\mat Y$ where $Y_{jj}=c_j/c_j({\mat{XA}})$, then~$\mat{XAY}$ will have the right column sums. The problem with this approach, of course, is that the second step is likely to undo the good work of the first step, changing the row sums away from the desired values; it is not at all obvious how to \emph{simultaneously} get the row sums and column sums right. Nevertheless, the approach of alternating row-normalizations with column-normalizations turns out to work. This alternating algorithm is known as \emph{Sinkhorn's algorithm}~\cite{sinkhorn:scaling}, and has actually been (re)discovered independently in several different contexts.

For matrix balancing there is a similar method known as \emph{Osborne's algorithm}~\cite{10.1145/321043.321048,parlett&reinsch:balancing}. In each iteration this chooses a row index $i$ and defines $X_{ii}$ such that the $i$-th row sum and the $i$-th column sum become equal. Again, because each iteration can undo the good work of earlier iterations it is not at all obvious that this converges to a balancing of $\mat A$.
Remarkably, even though Osborne's algorithm was proposed more than six decades ago and is widely used in linear algebra software, an explicit bound on its convergence rate has only been proven very recently~\cite{schulman&sinclair:precond,doi:10.1137/1.9781611974782.11}!

At the same time there have been other, more sophisticated algorithmic approaches for scaling and balancing.
Just to mention one: we can parametrize ${\mat X}=\diag(e^{\vec x})$ and ${\mat Y}=\diag(e^{\vec y})$ by vectors ${\vec x},{\vec y}\in\R^n$ and consider the following convex potential function:
\[
f({\vec x},{\vec y})=\sum_{i,j=1}^n A_{ij}e^{x_i+y_j} - \sum_{i=1}^n r_i x_i - \sum_{j=1}^n c_j y_j.
\]
Note that the partial derivative of this~$f$ w.r.t.\ the variable $x_i$ is
$\sum_{j=1}^n A_{ij}e^{x_i+y_j}-r_i=r_i(\mat{XAY})-r_i$,
and the partial derivative w.r.t.\ $y_j$ is $c_j(\mat{XAY}) - c_j$.
A minimizer ${\vec x},{\vec y}$ of~$f$ will have the property that all these $2n$ partial derivatives are equal to~0, which means $\mat{XAY}$ \emph{is exactly scaled!}
Accordingly, (approximate) scalings can be obtained by finding (approximate) minimizers using methods from convex optimization. In fact, Sinkhorn's original algorithm can be interpreted as coordinate descent on this $f$, and Osborne's algorithm can similarly be derived by slightly modifying~$f$. More advanced methods from convex optimization have also been applied, such as ellipsoid methods~\cite{KALANTARI199687,doi:10.1137/S0895479895289765,NEMIROVSKI1999435}, box-constrained Newton methods~\cite{azlow17,cmtv17} and interior-point methods~\cite{cmtv17,burgisser2020interior}.

Historically, research on matrix scaling and matrix balancing (and generalizations such as operator scaling) has focused on finding $\eps$-$\ell_2$-scalings. More recently also algorithms for finding $\eps$-$\ell_1$-scalings have been extensively studied, due to their close connection with permanents and finding perfect matchings in bipartite graphs~\cite{lsw00,chakrabarty2020better}, and because the $\ell_1$-distance is an important error measure for statistical problems such as computing the optimal transport distance between distributions~\cite{NIPS2013_af21d0c9,altschuler2017nearlinear}, even already for constant~$\eps$. By the Cauchy-Schwarz inequality, an $(\eps/\sqrt{n})$-$\ell_2$-scaling for $\mat A$ is also an $\eps$-$\ell_1$-scaling, but often more direct methods work better for finding an $\eps$-$\ell_1$-scaling.

Below in \cref{tab: best time complexity for uniform r c} we tabulate the best known algorithms for finding $\eps$-scalings in $\ell_1$-norm for entrywise-positive matrices and general non-negative matrices.\footnote{For entrywise-positive matrices, the second-order methods (i.e., those that use the Hessian, not just the gradient) theoretically outperform the \emph{classical} first-order methods in any parameter regime. However, they depend on highly non-trivial results for graph sparsification and Laplacian system solving which are not easy to efficiently  implement in practice, in sharp contrast with the eminently practical Sinkhorn and Osborne algorithms.}
For the well-definedness of the algorithms, we will always assume the $n\times n$ input matrix $\mat A$ has at least one non-zero entry in every row and column, and every entry of the target marginals $\vec r, \vec c$ is non-zero. In addition, we assume $\mat A$ is \emph{asymptotically scalable}: for every $\eps>0$, there exist $\mat X$ and $\mat Y$ such that
\[
\norm{{\vec r}({\mat B})-{\vec r}}_1+ \norm{{\vec c}({\mat B})-{\vec c}}_1\leq \eps,
\]
where $\mat B=\mat{XAY}$.
A sufficient condition for this is that the matrix is entrywise-positive.
To state the complexity results, let $m$ be the number of non-zero entries in $\mat A$, assume $\sum_{i,j=1}^n A_{ij} =1$, assume that its non-zero entries lie in $[\mu, 1]$, and $\norm{\vec r}_1=\norm{\vec c}_1=1$ (so a uniform marginal would be~$\vec 1/n$).
We will assume $\eps\in(0,1)$. The input numbers to the algorithm are all assumed to be rational, with bit size bounded by $\polylog(n)$, unless specified otherwise.
\begin{table}[ht!]
 {   \centering
\begin{tabular}{|c|l|l|l|}
\hline
& \multicolumn{1}{c|}{$(\vec 1/n, \vec 1/n)$} & \multicolumn{1}{c|}{$(\vec r,\vec c)$} & \multicolumn{1}{c|}{References and remarks} \\
\hline
\multirow{5}{*}{General} & $\widetilde O(m/\eps^2)$ & $\widetilde O(m/\eps^2)$ & Sinkhorn, via KL \cite{chakrabarty2020better}\footnote{Their proofs work only for input matrices that are exactly scalable. However, with our potential gap bound we can generalize their analysis to work for arbitrary asymptotically-scalable matrices.} \\
& $\widetilde O(mn^{2/3}/\eps^{2/3})$&$\widetilde O(mn/(h^{1/3}\eps^{2/3}))$& first-order, via $\ell_2$ \cite{azlow17} \\
& $\widetilde O(m\log\kappa)$&$\widetilde O(m\log\kappa)$ & box-constrained method, via $\ell_2$ \cite{cmtv17} \\
& $\widetilde O(m^{1.5})$ &$\widetilde O(m^{1.5})$& interior-point method, via $\ell_2$ \cite{cmtv17} \\
& \bm{$\widetilde O(\sqrt{mn}/\eps^4)$}& \bm{$\widetilde O(\sqrt{mn}/\eps^4)$}& \textbf{Sinkhorn, quantum, \cref{cor:full sinkhorn performance l1}} \\
\hline
&$\widetilde O(n^2/\eps)$ & $\widetilde O(n^3/\eps)$ & Sinkhorn, via $\ell_2$~\cite{KALANTARI1993237,klrs08}, $h\eps\leq\sqrt{2n}$ \\
Entrywise & $\widetilde O(n^2/\eps^2)$ & $\widetilde O(n^2/\eps^2)$ & Sinkhorn, via KL \cite{altschuler2017nearlinear,chakrabarty2020better} \\
positive & $\widetilde O(n^2)$&$\widetilde O(n^2)$& box-constrained, via $\ell_2$~\cite{azlow17,cmtv17} \\
& $\bm{\widetilde O(n^{1.5}/ \eps^3)}$& $\bm{\widetilde O(n^{1.5}/ \eps^3)}$& \textbf{Sinkhorn, quantum, \cref{thm: faster for positive informal}} \\
\hline
\end{tabular}
}
    \caption{State-of-the-art time complexity of first- and second-order methods for finding an $\eps$-$\ell_1$-scaling, both to uniform marginals and to arbitrary marginals. The boldface lines are from this paper, and the only quantum algorithms for scaling that we are aware of. Here $h$ is the smallest integer such that $h\vec r$ and $h\vec c$ are integer vectors; $m$ is an upper bound on the number of non-zero entries of $\mat A$; $\kappa$ represents the ratio between the largest and the smallest entries of the optimal scalings $\mat X$ and $\mat Y$, which can be exponential in $n$. Many referenced results originally use a different error model (e.g., $\ell_2$ or Kullback-Leibler divergence), which we convert to guarantees in the $\ell_1$-norm for comparison.  Here the $\widetilde O$-notation hides polylogarithmic factors in $n$, $1/\eps$ and $1/\mu$.}
    \label{tab: best time complexity for uniform r c}
\end{table}

For matrix balancing, Osborne's algorithm has very recently been shown to produce an $\eps$-$\ell_1$-balancing in time $\widetilde O(m/\eps^2)$ when in each iteration the update is chosen randomly~\cite{altschuler2020random}. Algorithms based on box-constrained Newton methods and interior-point methods can find $\eps$-balancings in time $\widetilde O(m\log \kappa)$ and $\widetilde O(m^{1.5})$, respectively, where $\kappa$ denotes the ratio between the largest and the smallest entries of the optimal balancing.

\subsection{Our first contribution: quantum algorithms for \texorpdfstring{$\ell_1$}{l\_1}-scaling and balancing}
Because a classical scaling algorithm has to look at each non-zero matrix entry (at least with large probability), it is clear that $\Omega(m)$ is a classical lower bound. This would be $\Omega(n^2)$ in the case of a dense or even entrywise-positive matrix~$\mat A$. As can be seen from~\cref{tab: best time complexity for uniform r c}, the best classical algorithms also achieve this $m$ lower bound up to logarithmic factors, with various dependencies on $\eps$. The same is true for matrix balancing: $\Omega(m)$ queries are necessary, and this is achievable in different ways, with different dependencies on $\eps$ and/or other parameters.

Our main contribution in this paper is to give quantum algorithms for scaling and balancing that beat the best-possible classical algorithms, at least for relatively large~$\eps\in(0,1)$:
\begin{quote}
{\bf Scaling:} We give a quantum algorithm that (with probability $\geq 2/3$) finds an $\eps$-$\ell_1$-scaling for an asymptotically-scalable $n\times n$ matrix $\mat A$ with $m$ non-zero entries (given by an oracle) to desired positive marginals $\vec r$ and $\vec c$ in time $\widetilde O(\sqrt{mn}/\eps^4)$.
When $\mat A$ is entrywise positive (and hence $m=n^2$), the upper bound can be improved to $\widetilde O(n^{1.5}/\eps^3)$.\\[0.3em]
{\bf Balancing:} We give a quantum algorithm that (with probability $\geq 2/3$) finds an $\eps$-$\ell_1$-balancing for an asymptotically-balanceable $n\times n$ matrix $\mat A$ with $m$ non-zero entries (given by an oracle) in time $\widetilde O(\sqrt{mn}/\eps^4)$.
\end{quote}
Our scaling algorithms in fact achieve closeness measured in terms of the relative entropy, and then use Pinsker's inequality ($\norm{\vec p - \vec q}_1^2=O(D(\vec p||\vec q)$)) to convert this to an upper bound on the $\ell_1$-error.

Note that compared to the classical algorithms we have polynomially better dependence on $n$ and $m$, at the expense of a worse dependence on $\eps$.
There have recently been a number of new quantum algorithms with a similar tradeoff: they are better than classical in terms of the main size parameter but worse in terms of the precision parameter. Examples are the quantum algorithms for solving linear and semidefinite programs~\cite{brandao2016QSDPSpeedup,apeldoorn2017QSDPSolvers,brandao2017QSDPSpeedupsLearning,apeldoorn2018ImprovedQSDPSolving} and for boosting of weak learning algorithms~\cite{arunachalam&maity:qboosting,hamoudiea:qhedge,izdebski&wolf:qboosting}.

Conceptually our algorithms are quite simple: we implement the Sinkhorn and Osborne algorithms but replace the exact computation of each row and column sum by \emph{quantum amplitude estimation};
this allows us to approximate the sum of $n$ numbers up to some small multiplicative error $\delta$ (with high probability) at the expense of roughly $\sqrt{n} / \delta$ queries to those numbers, and a similar number of other operations.

Our analysis is based on a potential argument (for Sinkhorn we use the above-mentioned potential~$f$). The error~$\delta$ causes us to make less progress in each iteration compared to an ``exact'' version of Sinkhorn or Osborne. If $\delta$ is too large we may even make backwards progress, while if $\delta$ is very small there is no quantum speed-up! We show that there is a choice of $\delta$ for which the negative contribution due to the approximation errors is of the same order as the progress in the ``exact'' version, and that choice also results in a speed-up.
We should caution, however, that it is quite complicated to actually implement this idea precisely and to keep track of and control the various approximation errors and error probabilities induced by the quantum estimation algorithms, as well as by the fact that we cannot represent the numbers involved with infinite precision.\footnote{This issue of precision is sometimes swept under the rug in classical research on scaling algorithms.}

\subsection{Our second contribution: quantum lower bound for scaling}
A natural question would be whether our upper bounds can be improved further.
Since the output has length roughly $n$, there is an obvious lower bound of $n$ even for quantum algorithms. An $\widetilde O(n)$ quantum algorithm would, however, still be an improvement over our algorithm, and it would be a quadratic speed-up over the best classical algorithm.
We dash this hope here by showing that our algorithm is essentially optimal for constant~$\eps$, even for the special case of $\mat A$ that is exactly scalable to uniform marginals:
\begin{quote}
There exists a constant $\eps>0$ such that every quantum algorithm that (with probability $\geq 2/3$) finds an $\eps$-$\ell_1$-scaling for given $n\times n$ matrix $\mat A$ that is exactly scalable to uniform marginals and has $m$ potentially non-zero entries, has to make $\Omega(\sqrt{mn})$ queries to~$\mat A$.
\end{quote}
Our proof constructs a set of instances $\mat A$ that hide permutations, it shows how approximate scalings of such an $\mat A$ give us some information about the hidden permutation, and then uses the adversary method~\cite{ambainis:lowerboundsj} to lower bound the number of quantum queries to the matrix needed to find that information. In particular, we show that for a permutation $\sigma\in S_n$, learning for each $i \in [n]$ the value $\sigma(i) \bmod 2$, takes $\Omega(n\sqrt{n})$ queries to the entries of the associated permutation matrix.

\paragraph{Organization.}
In \cref{sec:preliminaries} we discuss notation, problem definitions, our input model and our computational model.
\Cref{sec: full testing} contains the analysis of Sinkhorn's algorithm where we allow for errors in the computation of the marginals and assume access to two subroutines. The classical and quantum implementations of those two subroutines are presented in \cref{sec:arithmetic}.
In \cref{sec:overview random} we study a randomized version of Sinkhorn's algorithm that is known to have good performance in practice; the asymptotic complexity will be the same as that of the algorithm described in \cref{sec: full testing}, but its analysis is more involved.
The analysis of the randomized version of Sinkhorn's algorithm does extend naturally to Osborne's algorithm for matrix balancing which we discuss in \cref{sec:balancing}; it seems that randomization is necessary to obtain a quantum speed-up here.
In \cref{sec:lower bound} we discuss our (matching) lower bound on matrix scaling.
In \cref{sec:potential bound} and \cref{sec:generalized pinsker} we present two technical results: a bound on the potential function we use for matrix scaling and a version of Pinsker's inequality for non-normalized vectors.
Finally in \cref{sec:positive} we give an improved analysis of Sinkhorn's algorithm for the special case of entrywise-positive matrices.

\section{Preliminaries}\label{sec:preliminaries}

\subsection{Notation and conventions}
We abbreviate $\R_+ = [0,\infty)$ and write $[n]=\{1,\dots,n\}$. We use $\log_2$ to denote the logarithm with base $2$ and $\ln$ to denote the natural logarithm with base $e$. We use $\indic_S$ to denote the indicator function of a probabilistic event $S$. When necessary, we use the $(b_1,b_2)$-fixed-point format to denote numbers: we use $b_1$ leading bits, $b_2$ trailing bits, and $1$ bit to denote the sign. That is, a number $a$ written in $(b_1,b_2)$-fixed-point format is a number of the form $a = \pm \sum_{i=-b_2}^{b_1-1} a_i 2^i$ where $a_i \in \{0,1\}$ for all $i$. For $a\in\R$ and a precision $\delta$, we say $\hat{a}$ is a $\delta$-additive approximation of $a$ if $\hat{a}\in[a+\delta,a-\delta]$; for positive $a$ we say $\tilde{a}$ is a $(1 \pm \delta)$-multiplicative approximation of $a$ if $\tilde{a}\in[(1-\delta)a,(1+\delta)a]$.

We write vectors $\vec x$ and matrices $\mat A$ in boldface, but their entries $x_i$ and $A_{ij}$ are written in regular face. Let $\R^n$ be the (Euclidean) space of all $n$-dimensional real vectors and $\R^{n\times n}$ be the space of all $n\times n$ real matrices. By convention we use $\vec 1$ for the all-$1$ vector and $\vec 0$ for the all-$0$ vector. For a matrix $\mat A\in\R^{n\times n}$, let $\mat A_{\ell\bullet}\in\R^n$ be the vector corresponding to the $\ell$-th row of $\mat A$ and $\mat A_{\bullet\ell}\in\R^n$ be the vector corresponding to the $\ell$-th column of $\mat A$.

We use the standard definition of big-O as a set. We write $\polylog(n) = \bigcup_{i=0}^{\infty} \bigO{\ln^i(n)}$. We use the big-$\widetilde O$ notation to hide polylogarithmic factors in the variables appearing within the parentheses.

For $p \in [1, \infty)$, we use the $\ell_p$-norm of $\vec x \in \R^n$ which is defined as $\|\vec x\|_p = \big(\sum_{i=1}^n |x_i|^p\big)^{1/p}$. Similarly, for $p = \infty$ we let the $\ell_\infty$-norm be $\|\vec x\|_\infty = \max_{i \in [n]} |x_i|$. We use $x_{\max}$ and $x_{\min}$ to denote the largest and smallest entry of $\vec x$, respectively. We define the $\ell_p$-norm of a matrix $\mat A\in\R^{n\times n}$ by viewing~$\mat A$ as a vector in $\R^{n^2}$ and computing its $\ell_p$-norm, e.g., $\norm{\mat A}_1=\sum_{i,j=1}^n |A_{ij}|$.

We apply functions to vectors entry-wise, for example $\sqrt{\vec x} = (\sqrt{x_i})_{i \in [n]}\in\R^n_+$ for $\vec x \in \R^n_+$ and $e^{\vec x} = (e^{x_i})_{i \in [n]}\in\R^n$ for $\vec x \in \R^n$.

We say $\vec x$ and $\mat A$ are \emph{non-negative} if all of their entries are greater or equal to $0$; and we say $\vec x$ and $\mat A$ are \emph{entrywise-positive} if all of their entries are strictly greater than $0$. Let $\R_+^n$ be the cone of all $n$-dimensional non-negative vectors and $\R_+^{n\times n}$ be the cone of all $n\times n$ non-negative matrices.

\subsection{Distance measures}
For pairs of non-negative vectors we are interested in their \emph{relative entropy} (or \emph{Kullback-Leibler divergence}) and \emph{Hellinger distance}. To define this we need the function $\rho\colon \R_+ \times \R_+ \to [0,\infty]$ which is defined as $\rho(a \Vert b) = b - a + a \ln \frac{a}{b}$ (with the usual conventions, in particular $0\ln 0=0$). The \emph{relative entropy} or \emph{Kullback-Leibler divergence} $D\colon \R_+^n \times \R_+^n \to [0,\infty]$ is then defined as $D(\vec a \Vert \vec b) = \sum_{i=1}^n \rho(a_i \Vert b_i)$.
If $\vec a,\vec b$ are probability distributions then $D(\vec a \Vert \vec b) = \sum_{i=1}^n a_i\ln \frac{a_i}{b_i}$, but we will also consider unnormalized $\vec a,\vec b$.  The (unnormalized) Hellinger distance between $\vec a, \vec b \in \R_+^n$ is defined as $\|\sqrt{\vec a}-\sqrt{\vec b}\|_2$.
A particularly useful inequality for the relative entropy is Pinsker's inequality. We will use a generalization of Pinsker's inequality that allows for unnormalized $\vec b$:

\begin{restatable}[Generalized Pinsker]{lem}{lemGenPinsker}
  \label{lem:generalized pinsker}
\lemgeneralizedpinskercontent
\end{restatable}

We also have a lower bound on the Hellinger distance between two unnormalized vectors in terms of their $\ell_1$-distance, which is similar to Pinsker's inequality for the relative entropy.

\begin{restatable}[Lower bound on Hellinger distance]{lem}{lemPinskerHell}
  \label{lem:unnormalized hellinger}
  Let $\vec a, \vec b \in \R_+^n$ with at least one of the two vectors being non-zero.
  Then
  \[
    \norm{\sqrt{\vec a} - \sqrt{\vec b}}_2^2 \geq \frac{\norm{\vec a - \vec b}_1^2}{2 (\norm{\vec a}_1 + \norm{\vec b}_1)}.
  \]
\end{restatable}

The proofs of both lemmas can be found in~\cref{sec:generalized pinsker}.

\subsection{Matrix scaling and balancing}
Throughout we use $\vec r, \vec c \in \R^n$ as the desired row and column marginals. Unambiguously, we also use $\vec r\colon \R^{n \times n} \to \R^n$ (resp.\ $\vec c\colon \R^{n \times n} \to \R^n$) as the function that sends an $n \times n$-matrix to its row (resp.\ column) marginal: $\vec r(\mat A)$ (resp.\ $\vec c(\mat A)$) is the vector whose $i$-th entry equals $r_i(\mat A) = \sum_{j=1}^n A_{ij}$ (resp.\ $c_i(\mat A) = \sum_{j=1}^n A_{ji}$). We use $\A x y = (A_{ij} e^{x_i+y_j})_{i,j \in [n]}$ to denote the rescaled matrix $\mat A$ with scalings given by $e^{\vec x}$ and~$e^{\vec y}$. In the sections on matrix balancing we abbreviate $\mat A(\vec x)=\mat A(\vec x, -\vec x)$.

We say a non-negative matrix $\mat A\in\R_+^{n\times n}$ is \emph{exactly scalable} to $(\vec r,\vec c)\in\R_+^n\times \R_+^n$, if there exist $\vec x,\vec y\in\R^n$ such that
\[
  \vec r(\A x y)=\vec r~\text{and}~\vec c(\A x y)=\vec c.
\]
For an $\eps>0$, we say $\mat A\in\R_+^{n\times n}$ is \emph{$\eps$-$\ell_1$-scalable} to $(\vec r,\vec c)\in\R_+^n\times \R_+^n$, if there exist $\vec x,\vec y\in\R^n$ such that
\[
  \norm{\vec r(\A x y) - \vec r}_1 \leq \eps \text{ and } \norm{\vec c(\A x y) - \vec c}_1 \leq \eps.
\]
We say $\mat A\in\R_+^{n\times n}$ is \emph{asymptotically scalable} to $(\vec r,\vec c)\in\R_+^n\times \R_+^n$ if it is $\eps$-$\ell_1$-scalable to $(\vec r,\vec c)$ for every $\eps>0$.

In the matrix-balancing setting we require $\vec y = -\vec x$, and the marginals are compared to each other. We say a non-negative matrix $\mat A\in\R_+^{n\times n}$ is \emph{exactly balanceable}, if there exists a vector $\vec x \in\R^n$ such that
$\vec r(\mat A (\vec x))=\vec c(\mat A (\vec x))$.
For an $\eps>0$, we say $\mat A\in\R_+^{n\times n}$ is \emph{$\eps$-$\ell_1$-balanceable}, if there exists an $\vec x\in\R^n$ such that
\[
  \frac{\norm{\vec r(\mat A (\vec x)) -  \vec c(\mat A (\vec x))}_1}{\norm{\mat A}_1} \leq \eps.
\]
We say $\mat A\in\R_+^{n\times n}$ is \emph{asymptotically balanceable} if it is $\eps$-$\ell_1$-balanceable for every $\eps>0$.

We formally state the optimization problems associated with scaling and balancing below.
\begin{prb}[$\eps$-$\ell_1$-scaling problem]
\label{prob:ell1 scaling}
Given an $n \times n$ matrix $\mat A \in \R_+^{n\times n}$ and desired marginals $\vec r, \vec c \in \R_+^n$ with $\norm{\vec r}_1=\norm{\vec c}_1 = 1$, find $\vec x, \vec y \in \R^n$ such that
\[
\norm{\vec r(\A x y) - \vec r}_1 \leq \eps \text{ and } \norm{\vec c(\A x y) - \vec c}_1 \leq \eps.
\]
\end{prb}
One can use (a generalized version of) Pinsker's inequality to upper bound $\ell_1$-distance by relative entropy, and it turns out that our algorithm is most naturally analyzed with the error measured by the relative entropy. We therefore also consider the following problem:
\begin{prb}[$\eps$-relative-entropy-scaling problem]
  \label{prob:relentropy scaling}
  Given an $n \times n$ matrix $\mat A \in \R_+^{n\times n}$ and desired marginals $\vec r, \vec c \in \R_+^n$ with $\norm{\vec r}_1=\norm{\vec c}_1 = 1$, find $\vec x, \vec y \in \R^n$ such that
  \[
    D(\vec r \Vert \vec r(\mat A(\vec x, \vec y))) \leq \eps \text{ and } D(\vec c \Vert \vec c(\mat A(\vec x, \vec y))) \leq \eps.
  \]
\end{prb}
Finally, the matrix balancing problem is defined in a similar fashion:
\begin{prb}[$\eps$-$\ell_1$-balancing problem]\label{prob: balancing}
  Given an $n \times n$ matrix $\mat A \in \R_+^{n \times n}$, find $\vec x \in \R^n$ such that
  \[
    \frac{\|\vec r(\mat A(\vec x)) - \vec c(\mat A(\vec x))\|_1}{\|\mat A(\vec x)\|_1} \leq \eps.
  \]
\end{prb}

\paragraph{Assumptions on the instances.}
For convenience, we will always assume that the matrix $\mat A$ has at least one non-zero entry in every row and column (for both the scaling and balancing problems), and in the matrix scaling problem we assume the target marginals $\vec r, \vec c$ to be entrywise-positive.

In the matrix balancing problem we additionally assume that the diagonal entries of $\mat A$ are zero; note that this is without loss of generality, as the requirement for being $\eps$-balanced becomes stricter as the diagonal entries of $\mat A$ become smaller. Finally, we assume throughout that $\mat A$ is asymptotically scalable to $\vec r, \vec c$ (or asymptotically balanceable in \cref{sec:balancing}).

To consider the complexity of algorithms for both the matrix scaling problem and the matrix balancing problem, we assume the following about the encoding of the instance: each entry of $\mat A$ (and $\vec r, \vec c$ for the scaling problem) is a rational number whose denominator and numerator can be represented using $\polylog(n)$ bits. We furthermore assume $\mat A \in [0,1]^{n \times n}$ with $\|\mat A\|_1 \leq 1$, and $\mat A$ has $m$ possibly non-zero entries which are assumed to be at least $\mu >0$.\footnote{If such a bound $\mu > 0$ is unknown, it can be found (with high probability) with $\bigO{\sqrt{m}}$ queries and similar time complexity using quantum minimum finding~\cite{durr&hoyer:minimum}. Similarly, one can enforce $\norm{\mat A}_1 \leq 1$ with high probability by first using quantum maximum finding on $\mat A$ and then dividing by $n^2$ times the largest entry.} As mentioned before we assume that $\|\vec r\|_1 = \|\vec c\|_1=1$ in the matrix scaling problem.

\subsection{Data structure and computational model}

\paragraph{Input model.}
Here we describe how our quantum algorithms can access the entries of the input matrix $\mat A$. For classical algorithms the access model is the same, though of course classical algorithms cannot make queries in superposition.

We assume \emph{sparse black-box access} to the elements of $\mat A$ via lists of the potentially non-zero entries for each row and each column, as follows.
We assume in the $i$-th row there are $s^r_i$ potentially non-zero entries (the algorithm doesn't know their locations nor their values in advance), and in the $j$-th column there are $s^c_j$ potentially non-zero entries, where $\sum_{i=1}^n s^r_i=\sum_{j=1}^n s^c_j=m$. If our lists only contain the non-zero entries of the matrix then $m$ would be the total number of non-zero entries in $\mat A$, but our current set-up is a bit more flexible, allowing these lists to also contain some 0-entries. We will for simplicity assume these positive integers $s^r_1,\ldots,s^r_n,s^c_1,\ldots,s^c_n$ are known to the algorithm (either given explicitly as part of the input, or via their own oracle), though they could also be computed efficiently by binary search on the other oracles as explained below.

The oracles (unitaries) $O_I^{\mathrm{row}}$ and $O_I^{\mathrm{col}}$ allow us to find the indices of potentially non-zero elements of rows and columns of $\mat A$. Specifically:
\begin{align*}
  O_I^{\mathrm{row}} \ket{i}\ket{k}\ket{b} &= \ket{i}\ket{k}\ket{b+j(i,k)}  \qquad \text{ for } i,k,b \in [n],\\
  O_I^{\mathrm{col}} \ket{k}\ket{j}\ket{b} &= \ket{k}\ket{j}\ket{b+i(j,k)} \qquad \text{ for } j,k,b \in [n],
\end{align*}
where $j(i,k)\in[n]$ is the position of the $k$-th potentially non-zero element of row $i$ and similarly $i(j,k)\in[n]$ is the position of the $k$-th potentially non-zero element of column $j$. The addition in the last register is modulo $n$.
If $k>s^r_i$ then we define $j(i,k)=0$, so if $O_I^{\mathrm{row}}$ maps $\ket{i}\ket{k}\ket{b}$ to itself, then we learn that $k>s^r_i$ (this is what allows us to learn $s^r_i$ ourselves efficiently via binary search). We do the same for the columns.

We furthermore assume access to binary representations of the numerators and denominators of entries of $\mat A$ and $\vec r, \vec c$ in the usual way: we have oracles $O_{\mat A},O_{\vec r},O_{\vec c}$ that return the numerators and denominators of the entries of $\mat A, \vec r, \vec c$. For $i,j \in [n]$, and $b$ a string of the same number of bits as used for denominator and numerator, we have\footnote{Our algorithms actually do not require $O_{\vec r}$ and $O_{\vec c}$ to be quantum oracles, a classical input suffices.}
\begin{align*}
    O_{\mat A} \ket{i}\ket{j}\ket{b} &= \ket{i}\ket{j}\ket{b \oplus (A_{ij}^\mathrm{den},A_{ij}^\mathrm{num})} \\
    O_{\vec r} \ket{i}\ket{b} &= \ket{i}\ket{b \oplus (r_i^\mathrm{den},r_i^\mathrm{num})} \\
    O_{\vec c} \ket{j} \ket{b} &= \ket{j}\ket{b \oplus (c_j^\mathrm{den},c_j^\mathrm{num})}
\end{align*}
The reason we allow for rational inputs rather than fixed-point format for $\mat A$, $\vec r$ and $\vec c$ is because this allows uniform marginals to be represented exactly, converting from fixed-point inputs to rational inputs is a trivial task, and for consistency with the classical literature.

\paragraph{Computational model.}
Our computational model is of a classical computer (say, a Random Access Machine for concreteness) that, in addition to its classical computations, can invoke a quantum computer. The classical computer can write to a classical-write quantum-read memory (``QCRAM''), and send a description of a quantum circuit that consists of one- and two-qubit gates from some fixed discrete universal gate set\footnote{For concreteness assume our gate set contains the Hadamard gate, $T$-gate, Controlled-NOT, and 2-qubit controlled rotations over angles $2\pi/2^s$ for positive  integers~$s$ (these controlled rotations are used in the circuit for the quantum Fourier transform (QFT), which we invoke later in the paper).}, queries to the input oracles, and queries to the QCRAM to the quantum computer. The quantum computer runs the circuit, measures the full final state in the computational basis, and returns the measurement outcome to the classical computer.
We will use the QCRAM to store the scaling vectors $\vec x, \vec y$ at any point in time in our iterative algorithms and hence will need enough QCRAM to store these $2n$ numbers up  to sufficient precision (the required precision is analyzed later, in the body of the paper).
The QCRAM can be queried by the quantum computer in the same way as the above oracles $O_{\vec r}$ and $O_{\vec c}$.

The cost of the quantum subroutines will be measured by the total number of queries to $\mat A$ and the QCRAM, plus the number of one- and two-qubit gates. The cost of the classical computer will be measured by its number of elementary steps. This includes the cost of writing down the descriptions of the quantum circuits that the classical machine subcontracts to the quantum machine; in our algorithms these will be relatively simple, like versions of Grover's algorithm and amplitude estimation, and hence can be written down with at most a logarithmic overhead over their number of gates.
The total cost (or ``time complexity'') of our algorithms is the sum of their classical and quantum costs.

\section{Sinkhorn algorithm with testing for matrix scaling}\label{sec: full testing}

In this section we state~\cref{alg:FSFP testing}, a variant of the well-known Sinkhorn algorithm, and provide its analysis. The objective of Sinkhorn's algorithm is to find scaling vectors~$\vec x, \vec y \in \R^n$ such that the matrix $\A x y = (A_{ij}e^{x_i + y_j})_{i,j \in [n]}$ has row and column marginals $\vec r$ and $\vec c$, respectively. It does so in an iterative way. Starting from the rational matrix~$\mat A \in [0, 1]^{n \times n}$, it finds a vector $\vec x$ such that the row marginals of $(A_{ij}e^{x_i})$ are $\vec r$, and then it finds a $\vec y$ such that the column marginals of $\A x y$ are~$\vec c$. The second step may have changed the row marginals, so we repeat the procedure. We can view this as updating the coordinates of $\vec x$ and $\vec y$ one at a time, starting from the all-$0$ vectors.
To update the row scaling vectors, we wish to find $\vec x' = \vec x + \vec\Step$ such that
\[
    \vec r(\A {x'} y) = \vec r.
\]
Expanding the above equation yields
\[
    e^{\Step_\ell} \cdot r_\ell(\A x y) = r_\ell,
\]
for $\ell\in[n]$.
Since we assume every row and column contains at least one non-zero entry, the above equation has a unique solution, resulting in the following formula:
\begin{equation}\label{eq:sinkhorn row update}
  x'_\ell = x_\ell + \Step_\ell = x_\ell + \ln \left(\frac{r_\ell}{r_\ell(\A x y)} \right) = \ln \left(\frac{r_\ell}{\sum_{j=1}^n A_{\ell j} e^{y_j}} \right).
\end{equation}
Analogously, we can achieve that $\vec c(\A x {y'}) = \vec c$ if we instead update $\vec y' = \vec y + \vec\Step$, where
\begin{equation}\label{eq:sinkhorn col update}
  y'_\ell = y_\ell + \Step_\ell = y_\ell + \ln \left(\frac{c_\ell}{c_\ell(\A x y)} \right) = \ln \left(\frac{c_\ell}{\sum_{i=1}^n A_{i \ell} e^{x_i}} \right).
\end{equation}
We use the term ``one Sinkhorn iteration'' to refer to the process of updating all $n$ row scaling vectors, or updating all $n$ column scaling vectors.

We study a version of Sinkhorn's algorithm where, instead of computing row and column marginals in each iteration exactly, we use a 
\emph{multiplicative} approximation of the marginals to compute $\delta$-additive approximations of \cref{eq:sinkhorn row  update,eq:sinkhorn col update}.
In the classical literature, the approximation errors can be chosen to be very small, since the cost per iteration scales as $\polylog(1/\delta)$, and hence that error is essentially a minor technical detail.
In the quantum setting, we can obtain a better dependence in terms of $n$ at the cost of allowing for a $\poly(1/\delta)$-dependence.
Therefore, in the analysis below we pay particular attention to the required precision $\delta$.
We state the Sinkhorn algorithm in terms of two subroutines. For both subroutines we provide both classical and quantum implementations in~\cref{sec:arithmetic}. For the analysis of~\cref{alg:FSFP testing}, we only use the guarantees of the subroutines as stated, and do not refer to their actual implementation.

The first subroutine we use is \ApproxScalingFactor, which is used to update the scaling vectors.
In odd iterations we update the row scaling vector~$\vec x$ according to \cref{eq:sinkhorn row update}, while in even iterations we update the column scaling factors~$\vec y$ according to \cref{eq:sinkhorn col update} -- in both cases with additive precision~$\delta$ assuming the subroutine does not fail.
The second subroutine is \TestScaling, which tests whether scaling vectors $(\vec x, \vec y)$ yield a relative-entropy-scaling of the desired precision.
Both of these subroutines have a precision parameter and an upper bound on their failure probability.
Note that allowing for the possibility of failure is essential since the quantum implementation of the subroutines is inherently probabilistic.

The Sinkhorn algorithm thus has a number of tunable parameters. We provide an upper bound~$T$ on the number of Sinkhorn iterations to be performed, and a choice of fixed-point format $(b_1, b_2)$, which is used for storing each entry of the scaling vectors $(\vec x, \vec y)$. Apart from that, we use two precision parameters $\delta, \delta' \in (0, 1)$, one for each subroutine used in the algorithm, and a failure probability $\eta \in [0, 1]$ for each individual subroutine call. In \cref{thm: Full testing} we show how to choose these parameters for~\cref{alg:FSFP testing} such that the output $(\vec x, \vec y) \in \R^n \times \R^n$ forms an $\eps$-relative-entropy-scaling of $\mat A$ to $(\vec r, \vec c)$ with probability at least $2/3$. The resulting time complexity for running the algorithm with these parameters gives us our main result of this section, \cref{cor:full sinkhorn performance}, where we use results from \cref{sec:arithmetic} for the cost of implementing \ApproxScalingFactor and \TestScaling.
Note that the error as measured in relative entropy can be converted to $\ell_1$-error using (a generalization of) Pinsker's inequality (cf.~\cref{lem:generalized pinsker}).

\begin{algorithm}[t]
  \caption{Full Sinkhorn with finite precision and failure probability}\label{alg:FSFP testing}
  \Input{Oracle access to~$\mat A \in [0,1]^{n \times n}$ with~$\norm{\mat A}_1 \leq 1$
  and non-zero entries at least $\mu > 0$, target marginals~$\vec r, \vec c \in (0,1]^n$ with $\norm{\vec r}_1 = \norm{\vec c}_1 = 1$, iteration count~$T\in\N$, bit counts~$b_1, b_2\in\N$, estimation precision~$0<\delta<1$, test precision~$0<\delta'<1$ and subroutine failure probability~$ \eta\in[0, 1]$}
  \Output{Vectors $\vec x, \vec y \in \R^n$ with entries encoded in~$(b_1, b_2)$ fixed-point format}
  \Guarantee{For $\eps \in (0, 1]$, with parameters chosen as in~\cref{thm: Full testing}, $(\vec x, \vec y)$ form an $\eps$-relative-entropy-scaling of $\mat A$ to $(\vec r, \vec c)$ with probability $\geq 2/3$}

  $\vec x^{(0)}, \vec y^{(0)} \leftarrow \vec 0$\tcp*{entries in~$(b_1, b_2)$ fixed-point format}

  \smallskip

  \For{$t\leftarrow 1,2,\dotsc,T$}{
    \uIf{$t$ is odd}{
      \For{$\ell \leftarrow 1,2,\ldots,n$}{
        $x_\ell^{(t)} \leftarrow \ApproxScalingFactor(\mat A_{\ell \bullet}, r_\ell, \vec y^{(t-1)}, \delta, b_1, b_2, \eta, \mu)$\label{algline:sinkhorn approx row}\;
      }
      $\vec y^{(t)} \leftarrow \vec y^{(t-1)}$\;
    }
    \ElseIf{$t$ is even}{
      \For{$\ell \leftarrow 1,2,\ldots,n$}{
        $y_\ell^{(t)} \leftarrow \ApproxScalingFactor(\mat A_{\bullet \ell}, c_\ell, \vec x^{(t-1)}, \delta, b_1, b_2, \eta, \mu)$\label{algline:sinkhorn approx col}\;
      }
      $\vec x^{(t)} \leftarrow \vec x^{(t-1)}$\;
    }
    \If{\textup{\TestScaling}$(\mat A, \vec r, \vec c, \vec x^{(t)}, \vec y^{(t)}, \delta', b_1, b_2, \eta, \mu)$}{
    \Return $(\vec x^{(t)}, \vec y^{(t)})$\;}
    }
\Return{$(\vec{x}^{(T)}, \vec{y}^{(T)})$}\;
\end{algorithm}

\begin{restatable}{thm}{thmFullTestingCost}
  \label{cor:full sinkhorn performance}
  Let $\mat A \in [0, 1]^{n \times n}$ be a rational matrix with $\norm{\mat A}_1 \leq 1$ and $m$ non-zero entries, each at least $\mu > 0$, let $\vec r, \vec c \in (0,1]^n$ with $\norm{\vec r}_1 = \norm{\vec c}_1 = 1$, and let $\eps \in (0, 1]$.
  Assume $\mat A$ is asymptotically scalable to $(\vec r, \vec c)$.
  Then there exists a quantum algorithm that, given sparse oracle access to $\mat A$, with probability $\geq 2/3$, computes $(\vec x, \vec y) \in \R^n \times \R^n$ such that $\mat A(\vec x, \vec y)$ is $\eps$-relative-entropy-scaled to $(\vec r, \vec c)$, for a total time complexity of $\widetilde O(\sqrt{mn} / \eps^2)$.
\end{restatable}

\begin{procedure}[t]
   \caption{ApproxScalingFactor($\vec a, r, \vec y, \delta, b_1, b_2, \eta, \mu$)}
   \label{subroutine:ApproxScalingFactor spec}
   \ApproxScalingFactorIO
   \Guarantee{If $b_1 \geq \lceil \log_2(\abs{\ln(r / \sum_{j=1}^n a_j e^{y_j})}) \rceil$ and $b_2 \geq \lceil \log_2(1/\delta) \rceil$, then with probability at least~$1 - \eta$, $x$ is a $\delta$-additive approximation of $\ln(r / \sum_{j=1}^n a_j e^{y_j})$}
\end{procedure}

\begin{procedure}[t]
   \caption{TestScaling($\mat A, \vec r, \vec c, \vec x, \vec y, \delta, b_1, b_2, \eta, \mu$)}
   \label{subroutine:TestScaling spec}
   \Input{Oracle access to rational~$\mat A \in [0,1]^{n \times n}$  with $\|\mat A\|_1 \leq 1$, rational $\vec r, \vec c \in (0,1]^n$, oracle access to~$\vec x, \vec y \in \R^n$ encoded in~$(b_1, b_2)$ fixed-point format, test precision~$\delta \in (0, 1)$, desired failure probability $\eta \in [0,1]$, lower bound~$\mu > 0$ on non-zero entries of $\mat A$}

   \Output{A bit indicating whether $\vec x, \vec y$ forms a $\delta$-relative-entropy-scaling of $\mat A$ to target marginals $\vec r, \vec c$.}
   \Guarantee{If $b_1 \geq \log_2(\abs{\ln(r_\ell / \sum_{j=1}^n A_{\ell j} e^{y_j})})$ for all $\ell\in [n]$, and similarly for the columns, and furthermore $b_2 \geq \lceil \log_2(1/\delta) \rceil$, then with probability at least~$1 - \eta$: outputs \False if $D(\vec r \Vert \vec r(\A x y)) \geq 2 \delta$ or $D(\vec c \Vert \vec c(\A x y)) \geq 2 \delta$, outputs \True if both are at most $\delta$}
\end{procedure}

\subsection{Potential argument}
The analysis will be based on a potential argument, using the following convex function (already mentioned in the introduction) as potential:
\begin{align*}
  f(\vec x, \vec y) = \sum_{i,j=1}^n A_{ij} e^{x_i + y_j} - \sum_{i=1}^n r_i x_i - \sum_{j=1}^n c_j y_j.
\end{align*}
This potential function is often used in the context of matrix scaling, as its gradient is precisely the difference between the current and the desired marginals (as we mentioned in~\cref{sec:intro}).
Many of the more sophisticated algorithms for matrix scaling also try to minimize this function directly.
For our purposes, we first state a bound on the potential gap $f(\vec 0, \vec 0)-\inf_{\vec x, \vec y \in \R^n} f(\vec x, \vec y)$, whose proof is delayed until~\cref{sec:potential bound}.

\begin{restatable}[Potential gap]{lem}{potentialgap} \label{lem:potential gap}
  Assume $\mat A \in [0, 1]^{n \times n}$ with $\norm{\mat A}_1 \leq 1$ and non-zero entries at least $\mu > 0$.
  If $\mat A$ is asymptotically $(\vec r, \vec c)$-scalable, then we have
  \[
    f(\vec 0, \vec 0) - \inf_{\vec x, \vec y \in \RR^n} f(\vec x, \vec y) \leq \ln\left( \frac1\mu \right).
  \]
\end{restatable}

\noindent
For matrices $\mat A$ that are \emph{exactly} $(\vec r, \vec c)$-scalable, this bound is well-known (see~e.g.~\cite{klrs08,chakrabarty2020better}), but to the best of our knowledge, it has not yet appeared in the literature when $\mat A$ is only assumed to be asymptotically scalable to $(\vec r, \vec c)$.

One can show that, for a Sinkhorn iteration in which we update the rows exactly, i.e., $\hat x_\ell =  \ln(r_\ell/\sum_{j=1}^n A_{\ell j} e^{y_j})$ for $\ell\in[n]$, the potential decreases by exactly the relative entropy:
\begin{align}\label{eq:exact row update}
f(\vec x,\vec y) - f(\vec{\hat x},\vec y) = D(\vec r\Vert \vec r(\A x y)),
\end{align}
and similarly for exact column updates.
The next lemma generalizes this to allow for error in the update; it shows that we can lower bound the decrease of the potential function in every iteration in terms of the relative entropy between the target marginal and the current marginal, under the assumption that every call to the subroutine \ApproxScalingFactor succeeds.

\begin{lem}\label{lem:progress no failure full}
  Let $\mat A \in \R_+^{n \times n}$, let $\vec x, \vec y \in \R^n$, let $\delta \in [0, 1]$, and let $\vec {\hat x} \in \R^n$ be a vector such that for every $\ell \in [n]$, we have $\abs{\hat x_\ell - \ln(r_\ell / \sum_{j=1}^n A_{\ell j} e^{y_j})} \leq \delta$.
  Then
  \[
    f(\vec x, \vec y) - f(\vec {\hat x}, \vec y) \geq D\bigl(\vec r \big\Vert \vec r(\A x y)\bigr) - 2 \delta.
  \]
  A similar statement holds for an update of $\vec y$ (using $\vec c$ instead of $\vec r$ in the relative entropy).
\end{lem}
\begin{proof}
  We first note that we have the equalities
  \begin{align*}
    f(\vec x, \vec y) - f(\vec {\hat x}, \vec y) & = \sum_{\ell, j = 1}^n A_{\ell j} e^{x_\ell + y_j} - \sum_{\ell, j = 1}^n A_{\ell j} e^{\hat x_\ell + y_j} - \sum_{i=1}^n r_i\cdot (x_i-\hat x_i) \\
    & = \sum_{\ell = 1}^n \left(r_\ell(\A x y) - r_\ell(\A {\hat x} y) - r_\ell\cdot (x_\ell-\hat x_\ell)\right)
  \end{align*}
  Denote $z_\ell = \hat x_\ell - \ln(r_\ell/\sum_{j=1}^n A_{\ell j} e^{y_j})$, so that $\abs{z_\ell} \leq \delta$.
  Note that
  \[
    r_\ell(\A {\hat x} y) = e^{\hat x_\ell} \sum_{j=1}^n A_{\ell j} e^{y_j} = r_\ell e^{z_\ell}.
  \]
  Furthermore, we also have
  \[
    x_\ell - \hat x_\ell = \ln(\tfrac{1}{r_\ell} \sum_{j=1}^n A_{\ell, j} e^{x_\ell + y_j}) - z_\ell = -\ln(\tfrac{r_\ell}{r_\ell(\A x y)}) - z_\ell.
  \]
  Therefore we can rewrite
  \[
    r_\ell(\A x y) - r_\ell(\A {\hat x} y) - r_\ell\cdot (x_\ell - \hat x_\ell) = r_\ell(\A x y) - r_\ell (e^{z_\ell}-z_{\ell}) + r_\ell \ln (\tfrac{r_\ell}{r_\ell(\A x y)}).
  \]
  For $z \in [-1, 1]$ one can easily show that $e^z - z \leq 1 + 2 \abs z$, and so
  \begin{align*}
      r_\ell(\A x y) - r_\ell (e^{z_\ell}-z_{\ell}) + r_\ell \ln (\tfrac{r_\ell}{r_\ell(\A x y)}) & \geq r_\ell(\A x y) - r_\ell - 2 r_\ell \abs {z_\ell} + r_\ell \ln(\tfrac{r_\ell}{r_\ell(\A x y)})
      \\
      & = \rho(r_\ell \Vert r_\ell(\A x y)) - 2 r_\ell \abs {z_\ell},
  \end{align*}
  so that we may conclude
  \[
    f(\vec x, \vec y) - f(\vec {\hat x}, \vec y) \geq D(\vec r \Vert \vec r(\A x y)) - 2 \sum_{\ell=1}^n r_\ell \abs{z_\ell} \geq D(\vec r \Vert \vec r(\A x y)) - 2 \delta
  \]
  since $\abs{z_\ell} \leq \delta$ for every $\ell \in [n]$, and $\norm{\vec r}_1 = 1$.
\end{proof}
\noindent
The previous lemma showed that updating the scaling vectors with additive precision~$\delta$ suffices to make progress in minimizing the potential function~$f$, as long as we are still far away from the desired marginals (in relative entropy distance).
As we wish to store the entries of $\vec x$ and $\vec y$ with additive precision $\delta > 0$ using a $(b_1, b_2)$ fixed-point format, we need $b_2 \geq \lceil \log_2(1/\delta) \rceil$.
The guarantees of \ApproxScalingFactor and \TestScaling assert that this choice of $b_2$ is also sufficient.
\Cref{lem:scaling norm bound} shows how large we need to take $b_1$ to ensure that the requirements of \ApproxScalingFactor and \TestScaling are satisfied in any particular iteration.
\begin{rem}
Note that the algorithm returns as soon as {\TestScaling} returns \True, or after $T$ iterations.
However, for the sake of simplifying the analysis, we always assume that $\vec x^{(t)}$ and $\vec y^{(t)}$ are defined for $t = 0, \dotsc, T$.
\end{rem}

\begin{lem}[Bounding the scalings]\label{lem:scaling norm bound}
Let $\mat A \in [0, 1]^{n\times n}$ with $\norm{\mat A}_1 \leq 1$ and non-zero entries at least $\mu > 0$.
Let $T \geq 1$ and $\delta \in [0, 1]$.
Denote $\sigma = \max(\abs{\ln r_{\min}}, \abs{\ln c_{\min}})$. Let $b_2 = \lceil \log_2(1/\delta) \rceil$ and choose $b_1=\lceil\log_2(T) + \log_2 (\ln(\frac1{\mu}) + 1 + \sigma)\rceil$.
If for all $t\in[T]$ the subroutine {\ApproxScalingFactor} succeeds, then for all $t\in[T]$ and $\ell \in [n]$ we have
\begin{align*}
  \left|\ln\left(\frac{r_\ell}{\sum_{j=1}^n A_{\ell j} e^{y_j^{(t)}}}\right)\right| \leq 2^{b_1}, \quad
  \left|\ln\left(\frac{c_\ell}{\sum_{i=1}^n A_{i \ell} e^{x_i^{(t)}}}\right)\right| \leq 2^{b_1}
\end{align*}
and
\begin{align*}
  \norm{(\vec x^{(t)}, \vec y^{(t)})}_\infty \leq t \left( \ln \left(\frac1{\mu}\right) + \delta + \sigma \right) \leq t \left( \ln \left(\frac1{\mu}\right) + 1 + \sigma \right) .
\end{align*}
\end{lem}
\begin{proof}
    We prove the norm bound by induction and the other claim as we go along.
    The norm bound clearly holds at time $t = 0$.
    Now assume it holds at time $t$. Assume that $t+1$ is odd, so that we update the rows in iteration $t+1$ (the case when $t+1$ is even follows similarly). For each $\ell \in [n]$, we bound $x_\ell^{(t+1)}$.
    Note that, by assumption, we have
    \begin{equation}
      \label{eq:simp x_l}
      \delta \geq \abs*{x_\ell^{(t+1)} - \ln\left(\frac{r_\ell}{\sum_{j=1}^n A_{\ell j} e^{y_j^{(t)}}}\right)} = \abs*{x_\ell^{(t+1)} - \ln(r_\ell) + \ln\left(\sum_{j=1}^n A_{\ell j} e^{y_j^{(t)}}\right)}.
    \end{equation}
    Observe that
    \begin{equation*}
        - \ln \left(\sum_{j=1}^n A_{\ell j} e^{y_j^{(t)}}\right)
        \geq - \ln\left(\sum_{j=1}^n A_{\ell j} e^{\norm{\vec y^{(t)}}_\infty}\right)
        = - \norm{\vec y^{(t)}}_\infty -\ln\left(\sum_{j=1}^n A_{\ell j}\right)
        \geq - \norm{\vec y^{(t)}}_\infty,
    \end{equation*}
    where the last inequality uses $\sum_{j=1}^n A_{\ell j} \leq \norm{\mat A}_1 \leq 1$.
    Similarly, for the upper bound, we have
    \begin{equation*}
      - \ln\left(\sum_{j=1}^n A_{\ell j} e^{y_j^{(t)}}\right)
      \leq - \ln\left(\sum_{j=1}^n A_{\ell j} e^{-\norm{\vec y^{(t)}}_\infty}\right)
      \leq \norm{\vec y^{(t)}}_\infty - \ln\left(\sum_{j=1}^n A_{\ell j}\right)
      \leq \norm{\vec y^{(t)}}_\infty - \ln(\mu)
    \end{equation*}
    where we used that all non-zero entries of $\mat A$ are at least~$\mu$,
    and every row contains at least one non-zero entry. Note that these bounds together with the choice of $b_1$ and the inductive assumption on $\norm{\vec x^{(t)}}_\infty$ and $\norm{\vec y^{(t)}}_\infty$ imply the first claim.

    If we use these estimates in \cref{eq:simp x_l}, we obtain
    \begin{align*}
      \ln(r_\ell) - \norm{\vec y^{(t)}}_\infty - \delta \;\leq\; x_\ell^{(t+1)} \;\leq\; \norm{\vec y^{(t)}}_\infty + \ln(r_\ell) + \ln\left(\frac{1}{\mu}\right) + \delta.
    \end{align*}
    Since $\mu \leq 1$ and $r_\ell \leq 1$, this implies
    \begin{align*}
      \abs{x_\ell^{(t+1)}} \leq \norm{\vec y^{(t)}}_\infty + \ln\left(\frac 1 \mu\right) + \abs{\ln (r_\ell)} + \delta.
    \end{align*}
    This shows that $\norm{{\vec x}^{(t+1)}}_\infty \leq \norm{\vec y^{(t)}}_\infty + \ln(\frac{1}{\mu}) + \abs{\ln (r_{\min})} + \delta$.
    The case that we updated the columns in the $(t+1)$-st iteration is treated completely similarly.
    Thus, we conclude that
    \[
        \norm{(\vec{x}^{(t+1)}, \vec{y}^{(t+1)})}_\infty \leq \norm{(\vec x^{(t)}, \vec y^{(t)})}_\infty + \ln \left(\frac{1}{\mu}\right) + \sigma + \delta.
    \]
    By the induction hypothesis and $\delta \leq 1$, the desired upper bound holds at time $t+1$, and it suffices to use $b_1$ bits in any iteration to meet the requirements of \ApproxScalingFactor.
\end{proof}

To formally analyze the expected progress it will be convenient to define the following events.
\begin{definition}[Important events] \label{def: events St}
For $t = 1,\ldots, T$, we define the following events:
\begin{itemize}
    \item Let $S_t$ denote the event that all $n$ calls to {\ApproxScalingFactor} succeed in the $t$-th iteration.
    \item Define $S$ to be the intersection of the events $S_t$, i.e., $S=\cap_{t=1}^{T} S_t$.
\end{itemize}
\end{definition}
To give some intuition, we show below that the event $S$ is the `good' event where a row-update makes the relative entropy between $\vec r$ and the updated row-marginals at most $\delta$ (and similarly for the columns). We only use \cref{lem:other marginals full} in \cref{sec:positive}.
\begin{lem}
\label{lem:other marginals full}
If $S$ holds and $\delta \leq 1$, then the following holds for all $t\in[T]$:
\begin{itemize}
\item If $t$ is odd, then $D(\vec r \Vert \vec r(\mat A(\vec x^{(t)},\vec y^{(t)}))) \leq \delta$.
\item If $t$ is even, then $D(\vec c \Vert \vec c(\mat A(\vec x^{(t)},\vec y^{(t)}))) \leq \delta$.
\end{itemize}
\end{lem}
\begin{proof}
If $S_t$ holds and $t$ is odd, then
\begin{align*}
  \abs*{ x^{(t)}_\ell - \ln\left( \frac{r_\ell}{\sum_{j=1}^n A_{\ell j} e^{y^{(t-1)}_j}} \right) } \leq \delta
\end{align*}
for all $\ell\in[n]$, while $\vec y^{(t)} = \vec y^{(t-1)}$. Accordingly,
\[
  r_\ell(\mat A(\vec x^{(t)},\vec y^{(t)}))) = e^{x_\ell^{(t)}} \sum_{j=1}^n A_{\ell j} e^{y_j^{(t)}} \in [e^{-\delta}, e^\delta] \cdot \frac{r_\ell}{\sum_{j=1}^n A_{\ell j} e^{y_j^{(t-1)}}} \cdot \left(\sum_{j=1}^n A_{\ell j} e^{y_j^{(t-1)}}\right) = [e^{- \delta} r_\ell, e^{\delta} r_\ell].
\]
Since
\begin{align*}
  \rho(r_\ell \Vert r_\ell e^z)
    = r_\ell e^z - r_\ell + r_\ell \ln\left(\frac{r_\ell}{r_\ell e^z}\right)
    = r_\ell \left( e^z - 1 - z \right)
    \leq r_\ell \abs z
\end{align*}
for any $\abs z \leq 1$, we obtain
\begin{align*}
  D(\vec r \Vert \vec r(\mat A(\vec x^{(t)},\vec y^{(t)})))
    = \sum_{\ell = 1}^n \rho(r_\ell \Vert r_\ell(\mat A(\vec x^{(t)},\vec y^{(t)})))
    \leq \sum_{\ell=1}^n r_\ell\delta=\delta,
\end{align*}
again using $\norm{\vec r}_1 = 1$.
A similar computation yields the result for $t$ even.
\end{proof}

We can combine \cref{lem:potential gap,lem:progress no failure full} to show  \cref{alg:FSFP testing} returns, with high probability, an $\eps$-relative-entropy-scaling to $(\vec r, \vec c)$ by choosing $\delta = O(\eps)$.

\begin{prp}\label{thm: Full testing}
Let $\mat A \in [0, 1]^{n\times n}$ with $\norm{\mat A}_1 \leq 1$ and non-zero entries at least $\mu > 0$ and let $\vec r, \vec c \in (0,1]^n$ with $\norm{\vec r}_1 = \norm{\vec c}_1 = 1$.
Assume $\mat A$ is asymptotically scalable to $(\vec r, \vec c)$. For $\eps \in (0, 1]$,
choose
\[
T = \left\lceil \frac{8}{\eps} \ln\left(\frac1\mu\right)\right\rceil + 1,
\]
$\delta = \frac{\eps}{16}$,
$\delta' = \frac{\eps}{2}$,
$\eta = \frac 1{3 (n+1) T}$,
$b_2 = \lceil \log_2(\frac{1}{\delta}) \rceil$,
and $b_1 = \lceil \log_2(T) + \log_2(\ln( \frac1{\mu}) + \sigma+1) \rceil$, where $\sigma = \max(\abs{\ln r_{\min}}, \abs{\ln c_{\min}})$.
Then, \cref{alg:FSFP testing} with these parameters returns a pair $(\vec x, \vec y)$ such that
$D(\vec r \Vert \vec r(\mat A(\vec x, \vec y))) \leq \eps$ and $D(\vec c \Vert \vec c(\mat A(\vec x, \vec y))) \leq \eps$
with probability $\geq 2/3$.
\end{prp}
\begin{proof}
  The choice of $\eta$ is such that with probability at least $1 - (n+1)T \eta = 2/3$, the event $S$ holds (i.e., all calls to \ApproxScalingFactor succeed) and all calls to \TestScaling succeed.
  Assume this is the case. If there exists an iteration $t \in [T]$ for which \TestScaling outputs \True, then we have obtained a $2 \delta'$-relative-entropy-scaling of $\mat A$ to $(\vec r, \vec c)$, which is an $\eps$-relative-entropy-scaling by the choice of $\delta'$.
  We now bound the number of iterations in which \TestScaling can output \False, i.e., the number of iterations in which $(\vec x, \vec y)$ is not a $\delta'$-relative-entropy-scaling of $\mat A$.
  Suppose that $\tau \in [T]$ is the last iteration such that \TestScaling outputs \False.
  By \cref{lem:potential gap}, we have
  \[
  f_0 - f_\tau \leq \ln\left(\frac 1 \mu\right),
  \]
  which is positive since $\mu \leq 1$.
  We now lower bound the left-hand side by a telescoping sum: using \cref{lem:progress no failure full} (and implicitly \cref{lem:scaling norm bound}) we obtain
  \[
    f_0 - f_\tau = \sum_{t=1}^\tau (f_{t-1} - f_t) \geq 2\tau\delta =\frac{\eps \tau}{8}
  \]
  It follows that $\tau \leq 8\ln(\tfrac 1\mu)/\eps < T$, so \TestScaling must output \True in the $T$-th iteration at the latest.
  This concludes the proof.
\end{proof}

\subsection{Time complexity}
With the performance guarantees provided by~\cref{thm:quantum ApproxScalingFactor guarantee} and~\cref{lem:TestScaling guarantee} for quantum implementations of \ApproxScalingFactor and \TestScaling, we can state the time complexity of computing an $\eps$-relative-entropy-scaling of $\mat A$ to marginals $(\vec r, \vec c)$. We now prove the main theorem of this section, already stated earlier and repeated here for convenience.

\thmFullTestingCost*
\begin{proof}
  We show that \cref{alg:FSFP testing} with the parameters chosen as in \cref{thm: Full testing} has the stated time complexity. Note that the cost of computing these parameters from the input will be dominated by the runtime of the algorithm. \cref{thm: Full testing} shows that \cref{alg:FSFP testing} runs for at most $O(\ln(1/\mu)/\eps)$ iterations. Next we show the time complexity per iteration is $\widetilde O(\sqrt{mn}/\eps)$, which implies the claimed total time complexity of $\widetilde O(\sqrt{mn} / \eps^2)$.

  \Cref{thm:quantum ApproxScalingFactor guarantee} shows that invoking \ApproxScalingFactor with precision $\delta = \Theta(\eps)$ on a row containing $s$ potentially non-zero entries incurs a cost of order $\widetilde O(\sqrt{s} / \eps)$, where we suppress a polylogarithmic dependence on $n$.
  Since in one iteration of~\cref{alg:FSFP testing} we apply \ApproxScalingFactor once to each row or once to each column, using Cauchy--Schwarz the total cost of the calls to \ApproxScalingFactor in one iteration is
  \[
    \widetilde O\left(\sum_{i=1}^n \sqrt{s^r_i}/\eps + \sum_{j=1}^n \sqrt{s^c_j}/\eps\right) \subseteq \widetilde O(\sqrt{mn}/\eps),
  \]
  where we recall that $s^r_i$ and $s^c_j$ are the numbers of potentially non-zero entries in the $i$-th row and $j$-th column of $\mat A$, respectively, and $m$ is the total number of potentially non-zero entries in $\mat A$ (i.e., $\sum_{i=1}^n s^r_i = m = \sum_{j=1}^n s^c_j$).
  Similarly, \cref{lem:TestScaling guarantee} shows invoking \TestScaling with precision $\delta' = \Theta(\eps)$ incurs a cost of order $\widetilde O(\sqrt{mn} / \eps)$. Finally we observe that compiling the quantum circuits (and preparing their inputs) for the calls to \ApproxScalingFactor and \TestScaling can be done with at most a polylogarithmic overhead.
\end{proof}
\noindent
Note that the dependency on $\ln(1/\mu)$ is suppressed by the $\widetilde O$, since we assume that the numerator and denominator of any rational number in the input is bounded above by a polynomial in $n$.
Using a generalization of Pinsker's inequality (cf.~\cref{lem:generalized pinsker}), an $\eps$-relative-entropy-scaling is a $O(\sqrt{\eps})$-$\ell_1$-scaling, which implies the following:
\begin{cor}
  \label{cor:full sinkhorn performance l1}
  Let $\mat A \in [0, 1] \in \R^{n \times n}$ be a rational matrix with $\norm{\mat A}_1 \leq 1$ and $m$ non-zero entries, each at least $\mu > 0$, let $\vec r, \vec c \in (0,1]^n$ with $\norm{\vec r}_1 = \norm{\vec c}_1 = 1$, and let $\eps \in (0, 1]$.
  Assume $\mat A$ is asymptotically scalable to $(\vec r, \vec c)$.
  Then there exists a quantum algorithm that, given sparse oracle access to $\mat A$, with probability $\geq 2/3$,  computes $(\vec x, \vec y) \in \R^n \times \R^n$ such that $\mat A(\vec x, \vec y)$ is $\eps$-$\ell_1$-scaled to $(\vec r, \vec c)$, at a total time complexity of $\widetilde O(\sqrt{mn} / \eps^4)$.
\end{cor}
\noindent
In~\cref{sec:lower bound} we show that the complexity in terms of $n$ and $m$ is tight (up to logarithmic factors): $\Omega(\sqrt{mn})$ queries to the input are needed to solve the scaling problem for constant $\ell_1$-error (by Pinsker's inequality, the same lower bound is then implied for relative-entropy scaling as well).

In \cref{sec:positive},~\cref{thm: faster for positive}, we show that if the matrix $\mat A$ is entrywise-positive, then the number of iterations to obtain an $\eps$-relative-entropy-scaling can be reduced to roughly $1/\sqrt{\eps}$ rather than roughly $1/\eps$, leading to a quantum algorithm with time complexity $\widetilde O(n^{1.5} / \eps^{1.5})$. This also implies that one can find an $\eps$-$\ell_1$-scaling in time $\widetilde O(n^{1.5} / \eps^3)$.
\begin{thm}[Informal $\ell_1$-version of~\cref{thm: faster for positive}] \label{thm: faster for positive informal}
  Let $\mat A \in [\mu,1] \in \R^{n \times n}$ be a rational matrix with $\norm{\mat A}_1 \leq 1$, let $\vec r, \vec c \in (0,1]^n$ with $\norm{\vec r}_1 = \norm{\vec c}_1 = 1$, and let $\eps \in (0, 1]$.
  Assume $\mat A$ is asymptotically scalable to $(\vec r, \vec c)$.
  Then there exists a quantum algorithm that, given sparse oracle access to $\mat A$, with probability $\geq 2/3$,  computes $(\vec x, \vec y) \in \R^n \times \R^n$ such that $\mat A(\vec x, \vec y)$ is $\eps$-$\ell_1$-scaled to $(\vec r, \vec c)$, at a total time complexity of $\widetilde O(n^{1.5} / \eps^3)$.
\end{thm}

\section{Classical and quantum implementations of subroutines}\label{sec:arithmetic}

In this section we describe how to implement the \ApproxScalingFactor and \TestScaling subroutines described and invoked in \cref{sec: full testing}.
We provide both a classical and a quantum implementation.
Whereas the former uses classical methods to find maxima and compute sums, the latter uses quantum maximum finding and quantum approximate counting algorithms to achieve a better dependence on $n$, at the expense of a worse dependence on the desired precision $\delta$.

One obstacle for efficiently implementing the oracles is the fact that one cannot simply compute all numbers appearing in the expressions to sufficient precision. For \ApproxScalingFactor for instance, we aim to compute a number of the form $\ln(r / \sum_{j=1}^n a_j e^{y_j})$, but the $y_j$ can (and typically do) grow linearly with $n$, so we cannot compute $e^{y_j}$ with sufficiently high precision in time sublinear in $n$.
The idea is to instead compute additive approximations of relative quantities such as $e^{y_i - y_j} \leq 1$ for $i, j \in [n]$ with $y_j \geq y_i$, and use properties of the logarithm to relate this to the original desired quantity.
This approach is widely used in practice, and discussed for instance in \cite{altschuler2020random}.
Note that these are issues that concern both the classical and quantum setting, but are particularly important for the latter, since we aim for a better dependence on $m$ and $n$ for the time complexity.
Furthermore, we implement everything in a way such that the fixed-point format $(b_1, b_2)$ for both the input and output of the oracles is the same, as this avoids the need for the Sinkhorn and Osborne algorithms to change the encoding format in every iteration.

Throughout this section, we will take for granted the fact that there exist arithmetic circuits for computing ratios, exponentials, logarithms and trigonometric functions, with the following property:
given two fixed-point formats $(b_1, b_2)$ and $(b_3, b_4)$, for all inputs encoded in $(b_1, b_2)$ fixed-point format for which the number to be computed can be encoded in $(b_3, b_4)$ fixed-point format with additive error at most $2^{-b_4}$, the output of the circuit is such an additive approximation.
Furthermore, these circuits have size at most polynomial in $b_1$, $b_2$, $b_3$ and $b_4$.
For exponentials and logarithms, one may achieve this (for instance) by using repeated squaring and Taylor series approximations; see e.g.~\cite{brentzimmermann2010}.
This assumption also implies that given a rational number, encoded by its numerator and denominator, we can efficiently compute a fixed-point representation of the rational number to a desired additive precision (by treating the numerator and denominator as fixed-point numbers and computing their ratio), as long as the number of leading bits for the output format is chosen to be sufficiently large.

The following subroutine \RelativeEntryAdditiveApprox is common to both the classical and quantum settings, it efficiently computes the ratio of two numbers of the form $a e^y$ up to a certain precision.
\begin{algorithm}[ht!]
  \caption{RelativeEntryAdditiveApprox$(a_1, a_2, y_1, y_2, b, b_1, b_2, c, d)$}\label{alg:RelativeEntryAdditiveApprox}
  \Input{Non-negative numbers $a_1, a_2 \in [0, 1 - 2^{-b}]$ encoded in $(0, b)$ fixed-point format, numbers $y_1, y_2 \in \R$ encoded in $(b_1, b_2)$ fixed-point format, natural numbers $c, d$}
  \Output{$2^{-b}$-additive approximation $\gamma$ of $\min\{e^{y_1 - y_2} a_1 / a_2, 2^d - 2^{-c}\}$ encoded in $(d, c)$ fixed-point format, with the convention that $a_1 / a_2 = \infty$ whenever $a_2 = 0$}

  \If{$a_2 = 0$}{
    \Return{$2^d - 2^{-c}$}\;
  }
  \If{$a_1 = 0$}{
    \Return{$0$}\;
  }
  compute $\Delta \leftarrow y_1 - y_2$\;
  \uIf{$\Delta > b + d$}{
    \Return{$2^d - 2^{-c}$}\;
  }
  \uElseIf{$\Delta < -b - c$}{
    \Return{$0$}\;
  }
  \Else{
    compute estimate $\alpha \geq 0$ of $e^\Delta$ encoded in $(2 (b+d), b+c+3)$ fixed-point format\;
    compute estimate $\beta \geq 0$ of $a_1 / a_2$ encoded in $(b, 2(b+d)+c+3)$ fixed-point format\;
    compute $\gamma' = \alpha \cdot \beta$ exactly encoded in $(3b+2d, 3b+2d+2c+6)$ fixed-point format\;

    let $\gamma$ be the result of rounding $\min\{\gamma', 2^d - 2^{-c}\}$ to the nearest integer multiple of $2^{-c}$\;
    \Return{$\gamma$}\;
  }

\end{algorithm}
\begin{lem}
  \label{lem:RelativeEntryAdditiveApprox guarantee}
  Let $y_1, y_2 \in \R$ be two numbers encoded in $(b_1, b_2)$ fixed-point format, and let $a_1, a_2 \in [0, 1 - 2^{-b}]$ be non-negative numbers encoded in $(0, b)$ fixed-point format.
  Furthermore, let $c, d \geq 1$ be natural numbers.
  Then \cref{alg:RelativeEntryAdditiveApprox} with these inputs returns a non-negative number $\gamma \in [0, 2^d - 2^{-c}]$ encoded in $(d, c)$ fixed-point format, such that
  \[
    \abs{\gamma - \min\{\tfrac{a_1}{a_2} e^{y_1 - y_2}, 2^d - 2^{-c}\}} \leq 2^{-c}.
  \]
  Here we use the convention that $a_1 / a_2 = \infty$ whenever $a_2 = 0$.
  The algorithm terminates in time polynomial in $b, b_1, b_2, c, d$.
\end{lem}
\begin{proof}
  We use the same notation as in the algorithm.
  The cases where $a_2 = 0$ or $a_1 = 0$ are clear.
  Otherwise, if $\Delta = y_1 - y_2 > b + d$ then
  \[
    \frac{a_1}{a_2} e^\Delta > \frac{a_1}{a_2} e^{b + d} \geq 2^{-b} e^{b + d} > 2^d
  \]
  so $\min\{e^\Delta a_1 / a_2, 2^d - 2^{-c}\} = 2^d - 2^{-c}$, and returning this value is a correct output.
  Similarly, if $\Delta < -b-c$, then
  \begin{align*}
    \frac{a_1}{a_2} e^\Delta < \frac{a_1}{a_2} e^{-b-c} <  \frac{a_1}{a_2} 2^{-b-c} \leq 2^{-c}
  \end{align*}
  so $0$ is a $2^{-c}$-additive approximation of $e^\Delta a_1 / a_2$.
  For the final and most interesting case, note that $e^\Delta < 2^{2(b+d)}$ and $\tfrac{a_1}{a_2} \leq 2^b$, so it suffices to use $2(b+d)$ and $b$ leading bits to ensure that we include the first possibly non-trivial binary digit of $e^\Delta$ and $a_1 / a_2$, respectively.
  Due to our choice of the number of trailing bits, $\alpha$ and $\beta$ satisfy
  \[
    \abs{\alpha - e^\Delta} \leq 2^{-b-c-3}, \quad \abs{\beta - \tfrac{a_1}{a_2}} \leq 2^{-2(b+d)-c-3},
  \]
  and as a consequence we see that
  \begin{align*}
    \gamma' = \alpha \cdot \beta & \leq \frac{a_1}{a_2} e^{\Delta} + 2^{-b-c-3} \frac{a_1}{a_2} + 2^{-2(b+d)-c-3} e^{\Delta} + 2^{-(3b+2d+2c+6)} \\
                        & \leq \frac{a_1}{a_2} e^\Delta + 2^{-c-3} + 2^{-c-3} + 2^{-c-3} \\
                        & \leq \frac{a_1}{a_2} e^\Delta + 2^{-c-1}
  \end{align*}
  where we used $\tfrac{a_1}{a_2} \leq 2^b$ and $e^\Delta \leq 2^{2(b+d)}$ in the second inequality.
  Similarly, we obtain the lower bound
  \begin{align*}
    \gamma' = \alpha \cdot \beta & \geq \frac{a_1}{a_2} e^{\Delta} - 2^{-b-c-3} \frac{a_1}{a_2} - 2^{-2(b+d)-c-3} e^{\Delta} - 2^{-(3b+2d+2c+6)} \\
                        & \geq \frac{a_1}{a_2} e^{\Delta} - 2^{-c-3} - 2^{-c-3} - 2^{-4b-2c-6} \\
                        & \geq \frac{a_1}{a_2} e^{\Delta} - 2^{-c-1}
  \end{align*}
  again using $\tfrac{a_1}{a_2} \leq 2^b$ and $e^\Delta \leq 2^{2(b+d)}$ in the second inequality.
  The quantity $\gamma' = \alpha \cdot \beta$ can be computed exactly using $3(b+d)$ leading bits and $3b + 2d + 2c + 6$ trailing bits, and is guaranteed to be a $2^{-c-1}$-additive approximation of $e^\Delta a_1 / a_2$.
  Therefore, $\min\{\gamma', 2^d - 2^{-c}\}$ is a $2^{-c-1}$-additive approximation of $\min\{e^\Delta a_1 / a_2, 2^d - 2^{-c}\}$, and rounding to the nearest integer multiple of $2^{-c}$ incurs an additional additive error of at most $2^{-c-1}$, so $\abs{\gamma - \min\{e^\Delta a_1 / a_2, 2^d - 2^{-c}\}} \leq 2^{-c}$.
\end{proof}
With \cref{alg:RelativeEntryAdditiveApprox} in hand, we can construct \cref{alg:GreaterOrEqual}, which allows us to compare two numbers of the form $a e^y$. Note that the comparisons cannot be exact, since we cannot compute these numbers explicitly; if two numbers of this form are approximately equal, then we may return an ``incorrect'' result, but for the purposes of our algorithms, approximate comparisons suffice.
\begin{algorithm}[ht!]
  \caption{GreaterOrEqual$(a_1, a_2, y_1, y_2, c)$}\label{alg:GreaterOrEqual}
  \Input{Access to non-negative numbers $a_1, a_2 \in [0, 1 - 2^{-b}]$ with entries encoded in $(0, b)$ fixed-point format, numbers $y_1, y_2 \in \R$ encoded in $(b_1, b_2)$ fixed-point format, a natural number $c \geq 1$}
  \Output{A boolean value indicating whether $e^{y_1 - y_2} a_1 / a_2 \geq 1$ or not}
  \Guarantee{\True whenever $e^{y_1 - y_2} a_1 / a_2 \geq 1 + 2^{-c}$ or $(a_1, y_1) = (a_2, y_2)$, and \False whenever $e^{y_1 - y_2} a_1 / a_2 \leq 1 - 2^{-c}$, with $a_1 / a_2 = \infty$ when $a_2 = 0$}

  \lIf{$(a_1, y_1) = (a_2, y_2)$}{\Return{\True}}
  $\gamma \leftarrow $ \RelativeEntryAdditiveApprox$(a_1, a_2, y_1, y_2, c, 1)$ encoded in $(1, c)$ fixed-point format\;
  \Return{\True if $\gamma \geq 1$; and \False if $\gamma<1$}\;
\end{algorithm}
\begin{cor}
  \cref{alg:GreaterOrEqual} returns \True whenever $e^{y_1 - y_2} a_1 / a_2 \geq 1 + 2^{-c}$ or $(a_1, y_1) = (a_2, y_2)$, and $\False$ if $e^{y_1 - y_2} a_1 / a_2 \leq 1 - 2^{-c}$, and runs in time polynomial in $b, b_1, b_2, c$.
\end{cor}

\subsection{Classical implementation of ApproxScalingFactor}
In this subsection we show how to implement \ApproxScalingFactor classically.
\begin{algorithm}[ht!]
   \caption{Classical implementation of ApproxScalingFactor$(\vec a, r, \vec y, \delta, b_1, b_2, \eta, \mu)$}\label{alg:ApproxMarginalClassicalImpl}
   \ApproxScalingFactorIO
   replace the oracle for $\vec a$ by an oracle which maps $j$ to an encoding of $a_j$ in $(0, \lceil \log_2(1 / (\delta \mu)) + 2 \rceil)$ fixed-point format\;
   replace $r$ by an encoding of $r$ in $(0, \lceil \log_2(1/(r \delta)) + 2 \rceil)$ fixed-point format\;
   set $\delta' = \delta / 2$\;
   set $c = \lceil \log_2(n/\delta') + 4 \rceil$\;
   find $j^*$ be a $j$ such that $a_j e^{y_j}$ is maximal, using \GreaterOrEqual$(a_i, a_j, y_i, y_j, 1)$ for comparisons\;
   compute $S \leftarrow \sum_{j =1}^n$ \RelativeEntryAdditiveApprox$(a_j, a_{j^*}, y_j, y_{j^*}, c, 1)$ in $(\lceil \log_2(n) + 1 \rceil, c)$ fixed-point format\; \label{algline:classical REAA}
   compute estimate $\alpha \geq 0$ of $\ln(S)$ encoded in $(\lceil \log_2 (\log_2(n) + 1) \rceil, \lceil \log_2(1/\delta') + 3 \rceil)$ fixed-point format\;
   compute estimate $\beta$ of $\ln(r)$ encoded in $(2 \lceil \abs{\log_2(r)} \rceil, \lceil \log_2(1/\delta') + 3 \rceil)$ fixed-point format\;
   compute estimate $\gamma$ of $\ln(a_{j^*})$ encoded in $(b, \lceil \log_2(1/\delta') + 3 \rceil)$ fixed-point format\;
   \Return{$\beta - (y_{j^*} + \gamma + \alpha)$} encoded in $(b_1, b_2)$ fixed-point format\;
\end{algorithm}

\begin{thm}
  \label{thm:classical ApproxScalingFactor guarantee}
  For rational~$\vec a \in [0, 1]^n$ (whose non-zero entries are $\geq \mu$) and~$r \in (0, 1]$, a vector $\vec y \in \R^n$ with entries encoded in $(b_1, b_2)$ fixed-point format and $\delta \in (0, 1]$, such that $b_1 \geq \lceil \log_2(\abs{\ln(r / \sum_{j=1}^n a_j e^{y_j})}) \rceil$ and $b_2 \geq \lceil \log_2(1/\delta) \rceil$, \cref{alg:ApproxMarginalClassicalImpl} returns a $\delta$-additive approximation of
  \[
    \ln\left(r / \sum_{j=1}^n a_j e^{y_j}\right),
  \]
  encoded in $(b_1, b_2)$ fixed-point format, within time $\widetilde O(n \cdot \poly(b_1, b_2))$, where $\widetilde O$ hides polylogarithmic factors in $n$, $1 / \delta$, $1/\mu$, and $1/r$.
\end{thm}
\begin{proof}
  We first show that the conversion of $\vec a$ and $r$ from rational to fixed-point format is harmless, as long as we use the specified precision.
  For every $j \in [n]$, the fixed-point encoding $\hat a_j$ of $a_j$ is guaranteed to be a $\frac{\delta \mu}{8}$-additive approximation of $a_j$, which is zero if and only if $a_j = 0$.
  Since the non-zero entries of $\vec a$ are lower bounded by $\mu$, we know that $\hat a_j$ is a $(1 \pm \frac{\delta}{8})$-multiplicative approximation of $a_j$ for every $j \in [n]$.
  Similarly, the fixed-point encoding $\hat r$ of $r$ is guaranteed to be a $(1 \pm \frac{\delta}{8})$-multiplicative approximation of $r$. Note that $\hat r$ is encoded in $(0, \lceil \log_2(1/(r \delta)) + 2 \rceil)$ fixed-point format, which depends on $r$, but one may also use a lower bound for $r$ for determining the number of bits.
  Since $\hat a_j$ and $\hat r$ are multiplicative approximations of $a_j$ and $r$, the bound $\abs{\ln(1 + z)} \leq 2 \abs{z}$ for $z \in [-1/2, 1/2]$ gives
  \[
    \abs*{\ln\left(\hat r / \sum_{j=1}^n \hat a_j e^{y_j}\right) - \ln\left(r / \sum_{j=1}^n a_j e^{y_j}\right)} \leq \frac{\delta}{2}
  \]
  and it suffices to approximate $\ln(\hat r / \sum_{j=1}^n \hat a_j e^{y_j})$ to additive precision $\delta' = \frac{\delta}{2}$.
  In the rest of the proof we simply write $a_j$ and $r$ for their fixed-point versions $\hat a_j$ and $\hat r$.

  From the guarantees of \GreaterOrEqual, it follows that classical maximum finding (in time $\widetilde O(n \cdot \poly(b_1, b_2))$) returns a $j^* \in [n]$ such that $e^{y_j - y_{j^*}} a_j / a_{j^*} \leq 3/2 \leq 2 - 2^{-c}$ for any $j \in [n]$.
  Define $\xi_j$ to be the result of the $j$-th call to \RelativeEntryAdditiveApprox on \cref{algline:classical REAA} in \cref{alg:ApproxMarginalClassicalImpl}.
  By the guarantees of \RelativeEntryAdditiveApprox, $\xi_j$ is a $2^{-c}$-additive approximation of $\min\{e^{y_j - y_{j^*}} a_j / a_{j^*}, 2 - 2^{-c}\} = e^{y_j - y_{j^*}} a_j / a_{j^*}$, where the equality follows from the guarantee on $j^*$.
  It follows that $S$ is a $n2^{-c}$-additive approximation of $\sum_{j=1}^n e^{y_j - y_{j^*}} a_j / a_{j^*}$. Therefore, we see that
  \begin{align*}
    \ln(S) - \ln\left(\sum_{j=1}^n e^{y_j - y_{j^*}} a_j / a_{j^*}\right) \in  \left[\ln\left(1 - \frac{n 2^{-c}}{\sum_{j=1}^n e^{y_j - y_{j^*}} a_j / a_{j^*}}\right), \ln\left(1 + \frac{n 2^{-c}}{\sum_{j=1}^n e^{y_j - y_{j^*}} a_j / a_{j^*}}\right)\right].
  \end{align*}
  Since $\sum_{j=1}^n e^{y_j - y_{j^*}} a_j / a_{j^*} \geq e^{y_{j*} - y_{j^*}} a_{j^*} / a_{j^*} = 1$,
  \begin{align*}
    \ln\left(1 + \frac{n 2^{-c}}{\sum_{j=1}^n e^{y_j - y_{j^*}} a_j / a_{j^*}}\right) \leq \ln(1 + n 2^{-c}) \leq \ln\left(1 + \frac{\delta'}{16}\right) \leq \frac{\delta'}{8}
  \end{align*}
  using the bound $\abs{\ln(1 + z)} \leq 2 \abs z$ for $z \in [-1/4, 1/4]$.
  Similarly, we can obtain the lower bound
  \[
    \ln\left(1 - \frac{n 2^{-c}}{\sum_{j=1}^n e^{y_j - y_{j^*}} a_j / a_{j^*}}\right) \geq - \frac{\delta'}{8}.
  \]
  Therefore, computing $\ln(S)$ using $\lceil \log_2(1/\delta') + 3 \rceil$ trailing bits of precision yields a number $\alpha$ such that
  \[
    \abs*{\alpha - \ln\left(\sum_{j=1}^n e^{y_j - y_{j^*}} a_j / a_{j^*}\right)} \leq
    \abs*{\alpha - \ln(S)}+\abs*{\ln(S)-\left(\sum_{j=1}^n e^{y_j - y_{j^*}} a_j / a_{j^*}\right)} \leq
    \frac{\delta'}{4}.
  \]
  Since $\abs{\beta - \ln(r)} \leq \delta'/8$ and $\abs{\gamma - \ln(a_{j^*})} \leq \delta'/8$, we have
  \[
    \abs*{\beta - (y_{j^*} + \gamma + \alpha) - \ln\left(r / \sum_{j=1}^n a_j e^{y_j}\right)} \leq \frac{\delta'}{2}.
  \]
  as desired.
  By the assumption on $b_1$ and $b_2$, we may encode this number in a $(b_1, b_2)$ fixed-point format (which is the desired output format) with an additional error of at most $\frac{\delta'}{2}$, yielding a $\delta'$-additive approximation of $\ln(r / \sum_{j=1}^n a_j e^{y_j})$ as desired.
\end{proof}

\subsection{Quantum implementation of ApproxScalingFactor}

We use the subroutine \QuantumApproximateSum, which allows us to efficiently find a multiplicative approximation of the sum of the entries of a positive vector.
\begin{algorithm}[ht!]
  \caption{QuantumApproximateSum($U_{\vec v}$, $b$, $\delta$, $\eta$)}
  \label{subroutine:QAC}
  \Input{Oracle $U_{\vec v}$ for access to the entries of the vector $\vec v \in [0, 3/4]^n$ given in $(0, b)$ fixed-point format, such that at least one entry of $\vec v$ is $\geq 1/4$, precision $\delta \in (0, 1/2]$, failure probability $\eta > 0$}
  \Output{A number $\tilde s$ encoded in $(\lceil \log_2(n) + 1 \rceil, b)$ fixed-point format}
  \Guarantee{If $b \geq \lceil \log_2(1/\delta) \rceil + 6$, then with probability at least $1 - \eta$, $\tilde s$ is a $(1\pm\delta)$-multiplicative approximation of~$\sum_{j=1}^n v_j$}
  create an oracle $U_{\vec \zeta}$, using $U_{\vec v}$, which maps $j\mapsto \zeta_j$ where $\zeta_j$ is a $(1, \lceil \log_2(n/\delta) \rceil + 4)$-fixed-point representation of $\arcsin(\sqrt{v_j})$\;
  let $V$ be a unitary that prepares (with some auxiliary qubits that start and end in $\ket{0}$)
  \[
    \ket{\psi} = \frac{1}{\sqrt{n}}\sum_{j=1}^n\ket{j}\left(\sqrt{\tilde{v}_j}\ket{0}+\sqrt{1-\tilde{v}_j}\ket{1}\right)
  \]
  from $\ket{0^{\lceil \log_2 n\rceil}1}$, where $\tilde{v}_j = \sin(\zeta_j)^2$\;
  let $\mathcal G  = V(2\ket{0^{\lceil \log_2 n \rceil}1}\!\bra{0^{\lceil \log_2 n\rceil}1} - I)V^{\dagger} \left(I\otimes (2\ket{1}\!\bra{1} - I)\right)$ be the unitary that first reflects through the part of the state with a $\ket{1}$ in the last qubit, and then reflects through $\ket{\psi}$.
   $\mathcal G$ has eigenvalues $e^{ \pm \i \phi}$ on the 2-dimensional space~$\mathrm{span}\{(I \otimes \ket{z}\!\bra{z}) \ket{\psi} \colon z \in \{0,1\}\}$ for $\phi = 2\arcsin(\sqrt{\tfrac1n \sum_{j=1}^n \tilde v_j})$\;
  perform phase estimation on $\mathcal G$ with starting state $\ket{\psi}$ to obtain a $(0, \lceil \log_2(\sqrt{n}/\delta) \rceil + 8)$-fixed-point representation $\beta$ of $\pm\phi /(2 \pi)$ with probability $\geq 2/3$\;\label{eq:phase1}
  compute estimate $\gamma$ of $\sin(\pi \beta )^2$ in $(1, \lceil \log_2(n/\delta) \rceil + 8)$ fixed-point format\;\label{eq:phase2}
  compute $\tilde{s} = n \cdot \gamma$ in $(\lceil \log_2 n+1\rceil,b)$ fixed-point format\;\label{eq:phase3}
  repeat \cref{eq:phase1,eq:phase2,eq:phase3} $\bigO{\ln(1/\eta)}$ times and \Return the median of the $\tilde{s}$-values of the different runs\;
\end{algorithm}

We first describe how to implement \QuantumApproximateSum efficiently via standard amplitude-estimation techniques~\cite{bhmt:countingj}.
For completeness we give a full proof, paying special attention to the required bit-complexity of each operation.
\begin{thm}\label{thm:QAC}
  \QuantumApproximateSum can be implemented in time $\widetilde O(\sqrt{n} \cdot  \frac{1}{\delta} \cdot \ln(\tfrac 1 \eta))$, where $\widetilde O$ here hides polynomial factors in $b$, and polylogarithmic factors in $1/\delta$ and $n$.
\end{thm}
\begin{proof}
We may assume without loss of generality that $n = 2^k$ for some integer $k \geq 1$ (by padding $\vec v$ with additional $0$s to make its dimension a power of~2).
With two queries to the oracle $U_{\vec v}$ and $\poly(b, \log_2 n, \log_2(1/\delta))$ other gates, we can prepare the state
\[
\ket{\psi} = \frac{1}{\sqrt{n}}\sum_{j=1}^n\ket{j}\left(\sqrt{\tilde v_j}\ket{0}+\sqrt{1- \tilde v_j}\ket{1}\right),
\]
where $\tilde v_j = \sin(\zeta_j)^2$ is the squared sine of $\zeta_j$, and $\zeta_j \in [0, \pi/2]$ is a $\frac{\delta}{16n}$-additive approximation of $\arcsin(\sqrt{v_j}) \in [0, \pi/2]$. We let $a = \frac1n \sum_{j=1}^n \tilde{v}_j$ be the squared $\ell_2$-norm of the part of $\ket{\psi}$ ending in $\ket{0}$; this is the quantity we will estimate below. 
We write $V$ for a unitary that prepares the state $\ket{\psi}$ starting from $\ket{0^{\log_2 n}1}$. 
Such a unitary $V$ can be implemented by creating a uniform superposition in the first register, applying $U_{\vec \zeta}$, performing $\lceil \log_2(n/\delta) \rceil + 4$ controlled single-qubit rotations to create the amplitudes $\sqrt{\tilde{v}_j}$, and applying $U_{\vec \zeta}^{\dagger}$.
For $V$'s gate complexity, note that creating a uniform superposition $\tfrac1{\sqrt{n}} \sum_{j=1}^n \ket{j}$ can be done with an exact circuit consisting of $\log_2 n$ Hadamard gates, where we use the assumption that $n$ is a power of $2$.

Let $\ket{\psi_{good}} = \frac{1}{\sqrt{n}}\sum_{j=1}^n\ket{j}\sqrt{\tilde{v}_j}\ket{0}$ and $\ket{\psi_{bad}} =  \frac{1}{\sqrt{n}}\sum_{j=1}^n\ket{j}\sqrt{1-\tilde{v}_j}\ket{1}$ be the projections of $\ket{\psi}$ onto the subspaces with a 0 or 1 in the last qubit, respectively. 
Let $R_{bad} = \left(I\otimes (2\ket{1}\!\bra{1} - I)\right)$ be the unitary that reflects through states with a $1$ in the last register. Similarly, let $R_{\psi} = V(2\ket{0^{\log_2 n}1}\!\bra{0^{\log_2 n}1} - I)V^{\dagger}$ be the unitary that reflects through $\ket{\psi} = V\ket{0^{\log_2 n}1}$. Let $\mathcal G = R_{\psi} R_{bad}$.
In the literature $\mathcal G$ is known as the amplitude amplification operator (or Grover operator), and it is well known that $\mathcal G$ acts as a rotation by an angle $\phi = 2 \arcsin(\sqrt{a})$ on the $2$-dimensional subspace spanned by $\ket{\psi_{good}}$ and $\ket{\psi_{bad}}$. In particular, when viewed as an operator on this $2$-dimensional subspace $\mathcal G$ has eigenvalues $e^{\i \phi}$ and $e^{-\i \phi}$. Hence, performing phase estimation on $\mathcal G$ from the starting state $\ket{\psi}$ (which is a superposition of $\ket{\psi_{good}}$ and $\ket{\psi_{bad}}$) gives with probability $\geq 2/3$ an estimate of either $\phi$ or $-\phi$. 
To be more precise, phase estimation on $\mathcal G$ from the starting state $\ket{\psi}$ can be used to compute (with probability $\geq 2/3$) a $(0, \lceil \log_2(\sqrt{n}/\delta) \rceil + 8)$-fixed-point representation $\beta$ of either $+\phi/(2\pi)$ or $-\phi/(2\pi)$. Since we only use $\beta$ to compute the even function $\sin(\pi\beta)^2$ (see \cref{eq:phase2}), we will assume for ease of notation that $\beta$ is a representation of $\phi/(2\pi)$. 
Phase estimation can be done with the desired precision using $2^{\lceil \log_2(\sqrt{n}/\delta) \rceil + 8} \in O(\sqrt{n}/\delta)$ controlled applications of $\mathcal G$, and $\bigOt{(\lceil \log_2(\sqrt{n}/\delta) \rceil + 8)^2}$ other gates. Note that here we also use the assumption that $n$ is a power of~$2$, since the ``other gates'' include the circuit for the quantum Fourier transform.
We now show that phase estimation on $\mathcal G$ with this precision suffices to compute a $(1 \pm \delta)$-multiplicative approximation of $\sum_{j=1}^n v_j$.

First note that the squared sine is $1$-Lipschitz, so $\zeta_j$ being a $\tfrac{\delta}{16n}$-additive approximation of $\arcsin(\sqrt{v_j})$ and $\tilde v_j = \sin(\zeta_j)^2$ imply that $\abs{\tilde v_j - v_j} \leq \tfrac{\delta}{16n}$ and $\sum_{j=1}^n \tilde v_j$ is a $\tfrac\delta{16}$-additive approximation of $\sum_{j=1}^n v_j$. 
Since $\sum_{j=1}^n v_j \geq \tfrac{1}{4}$, it follows that $\sum_{j=1}^n \tilde v_j$ is also a $(1\pm\frac{\delta}{4})$-multiplicative approximation of $\sum_{j=1}^n v_j$.
Therefore, we only need to show that the returned quantity~$\tilde s$ is a sufficiently precise multiplicative approximation of $n \cdot a = \sum_{j=1}^n \tilde v_j$.

Since $v_j \in [0, 3/4]$ for every $j$, $\tilde v_j \in [0, 7/8]$ for every $j$, and because at least one $v_j \geq 1/4$, at least one $\tilde v_j \geq 1/8$.
As a consequence, the squared norm $a = \frac{1}{n} \sum_{j=1}^n \tilde v_j$ of the part of $\ket{\psi}$ ending in $\ket{0}$ satisfies $a \in [\frac{1}{8n}, \frac{7}{8}]$. 

Recall that $\phi = 2 \arcsin(\sqrt{a})$ and $\beta$ is a $(0, \lceil \log_2(\sqrt{n} / \delta) \rceil + 8)$-fixed-point representation of $\phi / (2 \pi)$ as obtained (with probability $\geq 2/3$) via phase estimation on the operator $\mathcal G$.
In other words, $\beta$ is a $\frac{\delta}{256 \sqrt{n}}$-additive approximation of $\phi/(2\pi)$, so $\pi \beta$ is a $\frac{\pi \delta}{256 \sqrt{n}}$-additive approximation of $\phi/2$. 
In~\cite[Lemma~7]{bhmt:countingj}, the following is shown: for $\theta_1, \theta_2 \in [0, 2\pi]$ with $\abs{\theta_1 - \theta_2} \leq \Delta$, we have
\begin{equation}\label{eq: ineq in BHMT}
    \abs{\sin(\theta_1)^2 - \sin(\theta_2)^2} \leq 2 \Delta \sqrt{\sin(\theta_2)^2 (1 - \sin(\theta_2)^2)} + \Delta^2.
\end{equation}
Applying this bound with $\theta_1 = \pi \beta$ and $\theta_2 = \phi/2 = \arcsin(\sqrt{a})$, and using that
\[
    \abs*{\pi \beta - \arcsin(\sqrt{a})} = \abs*{\pi \beta - \frac{\phi}{2}} \leq \frac{\pi \delta}{256 \sqrt{n}},
\]
we see that
\begin{align*}
    \abs*{\sin(\pi \beta)^2 - a} & \leq \frac{2 \pi \delta \sqrt{a (1 - a)}}{ 256 \sqrt{n}} + \frac{\pi^2 \delta^2}{256^2 \cdot n} \\
    & \leq \frac{8 \pi \delta}{256} \cdot a \cdot \sqrt{1 - a} + \frac{16\pi^2 \delta^2}{256^2} \cdot a \\
    & \leq \frac{\delta}{8} \cdot a,
\end{align*}
where the second inequality follows from the bound $a \geq 1/(8n) \geq 1/(16n)$, and the last inequality follows from $\sqrt{1 - a} \leq 1$, $\delta \leq 1$ and $8 \pi/256+16\pi^2/256^2\leq 1/8$.
In particular, this implies that $\sin(\pi \beta)^2$ is a $(1 \pm \tfrac{\delta}{8})$-multiplicative approximation of $a$.

Since $\sin(\pi \beta)^2$ is a $(1 \pm \tfrac{\delta}{8})$-multiplicative approximation of $a \geq 1/(8n)$, we see that $\sin(\pi \beta)^2 \geq 1/(16n)$.
Thus a $\tfrac{\delta}{256 n}$-additive approximation $\gamma$ of $\sin(\pi \beta)^2$ is a $(1 \pm \tfrac{\delta}{16})$-multiplicative approximation of $\sin(\pi \beta)^2$.
Since $\sin(\pi \beta)^2$ is a $(1 \pm \tfrac{\delta}{8})$-multiplicative approximation of $a$, $\gamma$ is a $(1 \pm \tfrac{\delta}{4})$-multiplicative approximation of $a$.
As a consequence, $n \cdot \gamma$ is a $(1 \pm \tfrac{\delta}{4})$-multiplicative approximation of $n \cdot a$.
Let $\tilde s$ be the $(\lceil\log_2(n) + 1\rceil, b)$-fixed-point encoding of $n \cdot \gamma$; then $\abs{\tilde s - n \cdot \gamma} \leq 2^{-b}$. 
Since $b \geq \lceil \log_2(1/\delta) \rceil + 6$, one has $2^{-b} \leq 2^{-6} \cdot \delta = \tfrac{\delta}{64}$. As we have the lower bound $n \cdot \gamma \geq n \cdot a \cdot \tfrac{3}{4} \geq \tfrac{1}{8} \cdot \tfrac{3}{4} \geq \tfrac{3}{32}$, computing $n \cdot \gamma$ with $2^{-b}$-additive error amounts to computing it with at most $(1 \pm \tfrac{\delta}{6})$-multiplicative error.
In other words, $\tilde s$ is a $(1 \pm \tfrac{\delta}{6})$-multiplicative approximation of $n \cdot \gamma$, which is a $(1 \pm \tfrac{\delta}{4})$-multiplicative approximation of $n \cdot a$,  which is a $(1 \pm \tfrac{\delta}{4})$-multiplicative approximation of $\sum_{j=1}^n v_j$.
As $\delta \leq 1$, one may check that this implies that $\tilde s$ is a $(1 \pm \delta)$-multiplicative approximation of $\sum_{j=1}^n v_j$.

We can further reduce the error probability to a small $\eta>0$ of our choice, by running the amplitude-estimation procedure $O(\ln(1/\eta))$ times and outputting the median of these runs. This can clearly be done using a polylogarithmic overhead.
The upper bound on the error probability follows from applying a Chernoff bound in order to upper bound the probability that the median has multiplicative error~$>\delta$.
\end{proof}

\begin{algorithm}[ht!]
   \caption{Quantum implementation of ApproxScalingFactor$(\vec a, r, \vec y, \delta, b_1, b_2, \eta, \mu)$}\label{alg:ApproxMarginalQuantumImpl}
   \ApproxScalingFactorIO
   replace the oracle for $\vec a$ by an oracle which maps $j$ to $a_j$, encoded in $(0, \lceil \log_2(1 / (\delta \mu)) + 2 \rceil)$ fixed-point format\;
   replace $r$ by the encoding of $r$ in $(0, \lceil \log_2(1/(r \delta)) + 2 \rceil)$ fixed-point format\;
   set $\delta' = \delta / 2$\;

   set $c = \lceil \log_2(n/\delta') \rceil + 6$\;
   find with quantum maximum finding a $j^*$ such that $e^{y_j - y_{j^*}} a_j / a_{j^*} \leq 3/2$ for all $j \in [n]$, using \GreaterOrEqual$(\vec a, \vec y, i, j, 1)$ for comparison, with failure probability $\eta / 2$\;
   let $U_{\vec v}$ be an oracle for the map $j \mapsto \frac{1}{2} \cdot $\RelativeEntryAdditiveApprox$(a_j, a_{j^*}, y_j, y_{j^*}, c, 1)$\;
    \label{algline:quantum REAA}
   use \QuantumApproximateSum with $U_{\vec v}$ to compute $S' = \sum_{j = 1}^n v_j$ in $(\lceil \log_2(n) \rceil + 1, c + 1)$ fixed-point format, with failure probability $\eta / 2$ and multiplicative error $\delta'/32$\;
   compute estimate $\alpha \geq 0$ of $\ln(2 S')$ in $(\lceil \log_2 (\log_2(n) + 2) \rceil, \lceil \log_2(1/\delta') + 3 \rceil)$ fixed-point format\;
   compute estimate $\beta$ of $\ln(r)$ in $(2 \lceil \abs{\log_2 r} \rceil, \lceil \log_2(1/\delta') + 3 \rceil)$ fixed-point format\;
   compute estimate $\gamma$ of $\ln(a_{j^*})$ in $(b, \lceil \log_2(1/\delta') + 3 \rceil)$ fixed-point format\;
   \Return{$\beta - (y_{j^*} + \gamma + \alpha)$} in $(b_1, b_2)$ fixed-point format\;
\end{algorithm}

With an implementation of \QuantumApproximateSum in hand, we prove the following theorem.
\begin{thm}
  \label{thm:quantum ApproxScalingFactor guarantee}
  With probability at least $1 - \eta$, \cref{alg:ApproxMarginalQuantumImpl} yields a $\delta$-additive approximation of $\ln(r / \sum_{j=1}^n a_j e^{y_j})$ encoded in $(b_1, b_2)$ fixed-point format, in time $\widetilde O (\sqrt n /\delta)$, where $\widetilde O$ here hides polynomial factors in $b$, $b_1$, $b_2$ and the encoding length of $\vec a$, and polylogarithmic factors in $n$, $1/\delta$, $1/\eta$, $1/\mu$ and $1/r$.
\end{thm}

\begin{proof}
  The details for rounding the input to fixed-point format are dealt with exactly as in the classical case (\cref{thm:classical ApproxScalingFactor guarantee}).
  Assume an index $j^*$ as stated in the algorithm is indeed found.
  For each $j \in [n]$, let $\xi_j$ be what would be the classical result of the $j$-th call to \RelativeEntryAdditiveApprox on \cref{algline:quantum REAA} in \cref{alg:ApproxMarginalQuantumImpl}.
  Note that in the algorithm, such calls are not made individually but in superposition (by the \QuantumApproximateSum subroutine); however, the $\xi_j$ is well-defined, as \RelativeEntryAdditiveApprox is a deterministic subroutine.
  Then for all $j \in [n]$, $\xi_j$ satisfies
  \begin{equation}
    \label{eq:xi_j bound}
    \abs{\xi_j - e^{y_j - y_{j^*}} a_j / a_{j^*}} \leq \frac{\delta'}{64 n},
  \end{equation}
  as $c \geq \log_2(n/\delta') + 6$ and $e^{y_j - y_{j^*}} a_j / a_{j^*} \leq 2^1 - 2^{-c}$.
  For convenience, we write $\chi_j = e^{y_j - y_{j^*}} a_j / a_{j^*}$.
  The number $S'$ returned by \QuantumApproximateSum satisfies
  \[
    S' \in \left[1 - \frac{\delta'}{32}, 1 + \frac{\delta'}{32}\right] \cdot \sum_{j=1}^n \frac{1}{2} \xi_j,
  \]
  so in particular, $2 S'$ satisfies
  \[
    2S' \in \left[\left(1 - \frac{\delta'}{32}\right) \sum_{j=1}^n \left(\chi_j - \frac{\delta'}{64 n}\right), \left(1 + \frac{\delta'}{32}\right) \sum_{j=1}^n \left(\chi_j + \frac{\delta'}{64n}\right)\right]
  \]
  by \cref{eq:xi_j bound}.
  Note that $\sum_{j=1}^n (\chi_j - \delta' / (64n))$ is non-negative, since every $\chi_j$ is non-negative, $\chi_{j^*} \geq 1$ and $\delta' \leq 1$.
  This implies that
  \begin{align*}
    \ln(2S') & \geq \ln\left(1 - \frac{\delta'}{32}\right) + \ln\left(\sum_{j=1}^n \left(\chi_j - \frac{\delta'}{64 n}\right)\right) \\
    & \geq - \frac{\delta'}{16} + \ln\left(\sum_{j=1}^n \chi_j\right) + \ln\left(1 - \frac{\delta'}{64 \sum_{j=1}^n \chi_j}\right) \\
    & \geq - \frac{\delta'}{16} + \ln\left(\sum_{j=1}^n \chi_j\right) + \ln\left(1 - \frac{\delta'}{64}\right) \\
    & \geq - \frac{\delta'}{16} + \ln\left(\sum_{j=1}^n \chi_j\right) - \frac{\delta'}{32}
  \end{align*}
  where we have used $\sum_{j=1}^n \chi_j \geq \chi_{j^*} = 1$ in the second inequality, and $\ln(1 - z) \geq -2z$ for $z \in [0, 1/2]$ in the first and third inequality.
  A similar computation shows that
  \[
    \ln(2 S') \leq \ln\left(\sum_{j=1}^n \chi_j\right) + \frac{\delta'}{16} + \frac{\delta'}{32}.
  \]
  To summarize, this shows that
  \[
    \abs*{\ln(2 S') - \ln\left(\sum_{j=1}^n \chi_j\right)} \leq \frac{\delta'}{16} + \frac{\delta'}{32}.
  \]
  For the rest of the argument, we proceed as in the classical case:
  since $\alpha \geq 0$ is an estimate of $\ln(2S')$ with $\lceil \log_2(1/\delta') + 3 \rceil$ bits of precision, we get
  \[
    \abs*{\alpha - \ln\left(\sum_{j=1}^n \chi_j\right)} \leq \abs{\alpha - \ln(2S')} + \frac{\delta'}{8} \leq \frac{\delta'}{4}.
  \]
  As we also have
  \[
    \abs*{\beta - \ln(r)} \leq \frac{\delta'}{8}, \quad \abs{\gamma - \ln(a_{j^*})} \leq \frac{\delta'}{8},
  \]
  we get
  \[
    \abs*{\beta - (y_{j^*} + \gamma + \alpha) - \ln\left(r / \sum_{j=1}^n a_j e^{y_j}\right)} \leq \frac{\delta'}{2}.
  \]
  Truncating to $b_2 \geq \lceil \log_2(1/\delta') \rceil$ bits introduces an additional error of at most $\frac{\delta'}{2}$, so the returned result is a $\delta'$-additive approximation of $\ln(r / \sum_{j=1}^n a_j e^{y_j})$.

  For the time complexity statement, note that the expensive operations are finding the maximum of the $\xi_j$ and approximating the sum of the $\xi_j$.
  The maximum finding subroutine returns a correct $j^*$ with error probability at most $\eta$ in time $\widetilde O(\sqrt n \ln(1/\eta))$.
  Approximating the sum can be done in time $\widetilde O(\sqrt{n} \cdot \frac{1}{\delta} \cdot \ln(\frac{1}{\eta}))$ by \cref{thm:QAC}.
  The other (arithmetic) operations can be implemented in time polynomial in $b_1, b_2, b$, the encoding length of $\vec a$, and polylogarithmic in $n$, $1/\eta$, $1/\mu$, $1/\delta$ and $1/r$, yielding the desired time complexity.
\end{proof}

\subsection{TestScaling implementations}
In this subsection, we describe how to efficiently implement \TestScaling; our implementation is shown in~\cref{alg:TestScaling impl}.
It relies on \ApproxScalingFactor and its subroutines, and its behavior is analyzed in~\cref{lem:TestScaling guarantee}.
Again, we can obtain either a classical or a quantum implementation, by using the respective version of \ApproxScalingFactor.
\begin{algorithm}
    \caption{Implementation of TestScaling$(\mat A, \vec r, \vec c, \vec x, \vec y, \delta, b_1, b_2, \eta, \mu)$}\label{alg:TestScaling impl}

    Compute $\approxA  \in [(1-\delta/80) \min\{\norm{\A x y}_1, 20\},(1+\delta/80) \min\{\norm{\A x y}_1, 20\}]$ with probability $\geq 1 - \eta / 2$\;\label{line:TestScaling ac}
    \If{$\approxA \geq 10$}{\Return{\False}\;}
    \For{$\ell \in [n]$}{
        Compute $a_\ell = -x_\ell + \ApproxScalingFactor(\mat A_{\ell \bullet}, r_\ell, \vec y, \delta / 4, b_1, b_2 + 2, \eta / 4n, \mu)$\;
        \tcp{on success, $a_\ell$ is a $\tfrac{\delta}{4}$-additive approx.~of $\ln(r_\ell/r_\ell(\A x y))$}
        Compute $b_\ell = -y_\ell + \ApproxScalingFactor(\mat A_{\bullet \ell}, c_\ell, \vec x, \delta / 4, b_1, b_2 + 2, \eta / 4n, \mu)$\;
        \tcp{on success, $b_\ell$ is a $\tfrac{\delta}{4}$-additive approx.~of $\ln(c_\ell/c_\ell(\A x y))$}
    }
    \Return{\True if $\approxA - 1 + \sum_{\ell=1}^n r_\ell a_\ell \leq 3\delta/2$ and $\approxA - 1 + \sum_{\ell=1}^n c_\ell b_\ell \leq 3\delta/2$}, otherwise \False\;
\end{algorithm}

We first show how to perform the first step of \TestScaling (\cref{line:TestScaling ac}), which is to check whether $\norm{\A x y}_1$ is at most a constant (here chosen to be $20$).

\begin{lem} \label{lem:multiplicative ell1}
  Let $\mat A \in [0, 1]^{n\times n}$ with $\norm{\mat A}_1 \leq 1$
and $m$ non-zero entries, each at least $\mu > 0$. Assume access to $\vec x, \vec y \in \R^n$.
 Then we can compute a multiplicative $(1\pm \delta)$-approximation of $\min\{\|\A x y\|_1,20\}$ in time $\widetilde O(m \polylog(1/\delta))$ using a classical algorithm, and (with probability $\geq 2/3$) in  time $\widetilde O(\sqrt{m}/\delta)$ using a quantum algorithm.
\end{lem}
\begin{proof}
We first compute the location of the largest entry $\nu>0$ of the matrix $\A x y$ using \GreaterOrEqual. Classically this can be done in time $\widetilde O(m \polylog(1/\delta))$, quantumly in time $\widetilde O(\sqrt{m} \polylog(1/\delta))$ using maximum-finding (see also the proof of \cref{thm:QAC}).

We then proceed as in the classical or quantum implementation of \ApproxScalingFactor and use \RelativeEntryAdditiveApprox and either summing (classically) or \QuantumApproximateSum (quantumly) to compute a multiplicative $(1\pm O(\delta))$-approximation of $\|\A x y\|_1/\nu$. This gives us an $O(\delta)$-additive approximation of $\ln(\|\A x y\|_1/\nu) + \ln(\nu) = \ln(\|\A x y\|_1)$. We can use this to determine whether $\norm{\A x y}_1$ is at most $20$ or at least $15$. In the latter case we can use the $O(\delta)$-additive approximation of $\ln(\|\A x y\|_1)$ to give a multiplicative $(1\pm \delta)$-approximation of $\min\{\|\A x y\|_1,20\}$. In the former case, we can efficiently exponentiate an $O(\delta)$-additive approximation of $\ln(\norm{\A x y}_1)$ to obtain a multiplicative approximation of $\norm{\A x y}_1$ since we have an upper bound on its value.

Note that the above has avoided computing the largest entry $\nu$ in $\A x y$ explicitly. This is important since $\nu$ may be exponentially large in $n$. We can, however, compute $\ln(\nu)$ efficiently since it is the logarithm of an entry of $\mat A$ plus the corresponding coordinates of $\vec x$ and $\vec y$.

Finally, note that (by \cref{lem:RelativeEntryAdditiveApprox guarantee} and \cref{thm:QAC}) the multiplicative $(1\pm O(\delta))$-approximation of $\|\A x y\|_1/\nu$ can be computed in time $\widetilde O(m \polylog(1/\delta))$ classically or $\widetilde O(\sqrt{m}/\delta)$ quantumly.
\end{proof}

We now analyze our implementation of \TestScaling, with the formal guarantee as follows.
\begin{thm}
\label{lem:TestScaling guarantee}
Let $\mat A \in [0, 1]^{n\times n}$ with $\norm{\mat A}_1 \leq 1$
and non-zero entries at least $\mu > 0$, and let $\vec r, \vec c \in (0,1]^n$ with $\norm{\vec r}_1 = \norm{\vec c}_1 = 1$.
Assume furthermore that $b_1 \geq \log_2(\abs{\ln(r_\ell / \sum_{j=1}^n A_{\ell j} e^{y_j})})$ and $b_1 \geq \log_2(\abs{\ln(c_\ell / \sum_{i=1}^n A_{i \ell} e^{x_i})})$ for every $\ell \in [n]$,
and that $b_2 \geq \lceil \log_2(1/\delta) \rceil$.
Then \TestScaling with inputs $\vec x, \vec y \in \R^n$ and $\delta\in(0,1]$ outputs, with success probability $\geq 1-\eta$,
\begin{itemize}
    \item \True if $(\vec x, \vec y)$ forms a $\delta$-relative-entropy-scaling of $\mat A$ to $(\vec r, \vec c)$.
    \item \False if either $D(\vec r \Vert \vec r(\A x y)) \geq 2 \delta$ or $D(\vec c \Vert \vec c(\A x y)) \geq 2 \delta$.
\end{itemize}
It does so using $1$ query to a subroutine for obtaining a multiplicative estimate of the sum of the matrix entries, and $2n$ queries to \ApproxScalingFactor, for a total classical time complexity of $\widetilde O(m\polylog(1/\delta))$, or a total quantum time complexity of $\widetilde O(\sqrt{mn} / \delta)$.
\end{thm}
\begin{proof}
First observe that the choices of $b_1$ and $b_2$ are assumed to be such that the assumptions for every call to \ApproxScalingFactor are satisfied. We use \cref{lem:multiplicative ell1} to implement \cref{line:TestScaling ac} with success probability $\geq 1-\eta/2$ so that $\approxA$ is a $(1 \pm \delta/80)$-multiplicative approximation of $\min\{ \norm{\A x y}_1, 20\}$.
Next, note that each call to \ApproxScalingFactor succeeds with probability at least $1 - \frac{\eta}{4n}$, so the probability of everything succeeding is at least $1 - \eta$ by a union bound.
Recall that
\begin{align*}
    D(\vec r\Vert\vec r(\A x y)) & = \sum_{i=1}^n \rho(r_i\Vert r_i(\A x y)) = \sum_{i=1}^n  (r_i(\A x y) - r_i + r_i \ln(\frac{r_i}{ r_i(\A x y)})) \\
    & = \norm{\A x y}_1 - 1 + \sum_{i=1}^n r_i \ln\left(\frac{r_i}{ r_i(\A x y)}\right).
\end{align*}
Therefore, we may estimate $D(\vec r \Vert \vec r(\A x y))$ to additive precision $\delta/2$ by estimating $\norm{\A x y}_1$ and $\sum_{i=1}^n r_i \ln(\frac{r_i}{r_i(\A x y)})$ to additive precision $\delta/4$.
Since $\norm{\A x y}_1 \leq 20$, obtaining a $(1\pm \delta/80)$-multiplicative approximation of  $\norm{\A x y}_1$ suffices to estimate it up to an additive error $\delta/4$.
We can now distinguish two cases:
\begin{enumerate}
\item If $\approxA > 10 \geq 5/(1-\delta/80)$, then we can conclude that $(\vec x, \vec y)$ does not form a $\delta$-relative-entropy-scaling of $\mat A$ to $(\vec r, \vec c)$. Indeed, a generalized version of Pinsker's inequality provided in~\cref{lem:generalized pinsker} shows that if $\norm{\vec r(\A x y)}_1 \geq 5$, then $D(\vec r\Vert \vec r(\A x y)) \geq (1-\ln 2)\cdot 4 > 1 \geq \delta$.\footnote{Note that naively applying Pinsker's inequality implies that whenever $D(\vec r\Vert\vec r(\A x y)) \leq \delta$ then also $\norm{\vec r - \vec r(\A x y)}_1 = O(\sqrt{\delta})$, which would indeed imply that $\norm{\vec r(\A x y)}_1 = O(1)$ for all $\delta \leq 1$. However, we may only apply Pinsker's inequality to probability distributions, which $\vec r(\A x y)$ need not  be. In~\cref{lem:generalized pinsker} we show a generalized version of Pinsker's inequality that says that $D(\vec r\Vert\vec r(\A x y))\geq \norm{\vec r - \vec r(\A x y)}_1 - \ln(1+\norm{\vec r - \vec r(\A x y)}_1)$ (the latter is lower bounded by $\norm{\cdot}_1^2/4$ on $[0,1]$ and by $(1-\ln 2)\norm{\cdot}_1$ on $[1,\infty)$).}
\item If $\approxA \leq 10$, then $\norm{\A x y}_1 \leq 10(1+\delta/80) \leq 15$ and a multiplicative $(1\pm \delta/80)$-approximation of $\norm{\A x y}_1$ thus forms an additive $15\delta/80 \leq \delta/4$-approximation of $\norm{\A x y}_1$.
\end{enumerate}
Finally, to estimate the last term, note that an additive $\delta/4$-approximation of $\ln(r_i/r_i(\A x y)$ (the output of \ApproxScalingFactor) for each $i \in [n]$ leads to an approximation of $\sum_{i=1}^n r_i \ln(\frac{r_i}{r_i(\A x y)})$ with additive error at most $\sum_{i=1}^n r_i \delta/4 = \delta/4$.
Therefore, the quantity $\gamma - 1 + \sum_{\ell=1}^n r_\ell a_\ell$ computed in \TestScaling is a $\delta/2$-additive approximation of $D(\vec r \Vert \vec r (\A x y))$. If $D(\vec r \Vert \vec r (\A x y))\leq \delta$, then the approximation is at most $3\delta/2$, and the condition evaluates to \True. Similarly, if $D(\vec r \Vert \vec r (\A x y)) \geq 2\delta$, then the approximation is at least $3\delta/2$, and the condition evaluates to \False. Note that if the approximation is exactly $3\delta/2$, then either return value is acceptable.
We can compute a $\delta/2$-additive approximation of $D(\vec c \Vert \vec c(\A x y))$ in the same manner, showing that our implementation satisfies the guarantees of the oracle.

Finally, for the time complexity of \TestScaling in the quantum setting, note that each call to \ApproxScalingFactor takes time $\widetilde O(\sqrt{s}/\delta)$ where $s$ is the number of potentially non-zero entries of the $\ell$-th row or column (and we suppress a polylogarithmic dependence on $n$). Since we call \ApproxScalingFactor once for each row and column and the square-root is a concave function, we thus obtain a time complexity of $\widetilde O(\sqrt{mn}/\delta)$ for \TestScaling.
\end{proof}

\section{Randomized Sinkhorn algorithm}\label{sec:overview random}

In this section we discuss a randomized version of the Sinkhorn algorithm, and sketch its analysis. Classically, this randomized version is of interest due to its good performance in practice (asymptotically the complexity is not better than the usual Sinkhorn algorithm). It is therefore natural to also provide a quantum analog.

Instead of updating all rows or all columns within a single iteration, the randomized Sinkhorn algorithm (\cref{alg:RSFP}) selects a uniformly random row or column, and only updates this row or column.
Because rows and columns have $m/n$ non-zero entries on average, this results in an \emph{expected} quantum time complexity of $\widetilde O(\sqrt{m/n} / \eps)$ per iteration, where $\eps > 0$ is the desired precision (measured in relative entropy).
Furthermore, we show that $\widetilde O(n / \eps)$ iterations suffice to obtain an $\eps$-scaling with probability $\geq 2/3$.
Note that we do not necessarily alternate between choosing rows or columns.
While the resulting runtime is comparable to that of the full Sinkhorn algorithm as in \cref{alg:FSFP testing}, the analysis is somewhat more difficult.
The reason is that we are no longer able to test whether the matrix is $\eps$-scaled at every iteration, since this has a quantum time complexity of roughly $\sqrt{mn} / \eps$.
Therefore, while the algorithm is running, we do not know whether the potential is still decreasing every iteration, and may lose progress during subsequent iterations.
However, we can show that for any probability $p > 0$ and subsequent appropriate choices of parameters (in particular for large enough $T$), a $1 - p$ fraction of the produced iterates $(x^{(t)},y^{(t)})$ yield an $\eps$-relative-entropy-scaling. This implies that a uniformly random choice of stopping iteration yields an $\eps$-relative-entropy-scaling with probability $1 - p$.
\begin{algorithm}[ht]
  \caption{Randomized Sinkhorn with finite precision and failure probability}
  \label{alg:RSFP}
  \Input{Oracle access to~$\mat A \in [0,1]^{n \times n}$ with~$\norm{\mat A}_1 \leq 1$ and non-zero entries at least $\mu > 0$, target marginals~$\vec r, \vec c \in [0,1]^n$ with $\norm{\vec r}_1 = \norm{\vec c}_1 = 1$, iteration count~$T \geq 0$, bit counts~$b_1, b_2 \geq 0$, precision~$\delta \in (0, 1)$ and subroutine failure probability~$\eta \in [0,1]$}
  \Output{Vectors $\vec x, \vec y \in \R^n$ with entries encoded in~$(b_1, b_2)$ fixed-point format}
  \Guarantee{For $p \in [0,1]$ and $\eps \in (0, 1]$ and parameters as in~\cref{thm: Full testing}, with probability at least $1 - p$, $(\vec x, \vec y)$ form an $\eps$-$(\vec r, \vec c)$-relative-entropy-scaling of $\mat A$}

  $\vec x^{(0)}, \vec y^{(0)} \leftarrow \vec 0$\;
  \smallskip
  \For{$t\leftarrow 1,2,\dotsc,T$}{
    Pick $\beta \in \{0,1\}$ uniformly at random\;\label{algline:RSFP row col choice}
    Pick $\ell \in [n]$ uniformly at random\;\label{algline:RSFP index choice}
    \uIf{$\beta =0$}{\label{algline:RSFP if row}
      $\vec x^{(t)} \leftarrow \vec x^{(t-1)}$\;
      $x_\ell^{(t)} \leftarrow \ApproxScalingFactor(\mat A_{\ell \bullet}, r_\ell, \vec y^{(t-1)}, \delta, b_1, b_2, \eta, \mu)$\;
      \label{algline:RSFP row estimate}
      \smallskip
      $\vec y^{(t)} \leftarrow \vec y^{(t-1)}$\;
    }
    \Else{
      $\vec y^{(t)} \leftarrow \vec y^{(t-1)}$\;
      $y_\ell^{(t)} \leftarrow \ApproxScalingFactor(\mat A_{\bullet \ell}, c_\ell, \vec x^{(t-1)}, \delta, b_1, b_2, \eta, \mu)$\;
      \label{algline:RSFP col estimate}
      \smallskip
      $\vec x^{(t)} \leftarrow \vec x^{(t-1)}$\;
    }
  }

  \smallskip

  Pick $\tau \in [T]$ uniformly at random\;
  \Return{$(\vec{x}^{(\tau-1)}, \vec{y}^{(\tau-1)})$}\;
\end{algorithm}
\begin{rem}
Note that \cref{alg:RSFP} returns $\vec x^{(\tau-1)}$ and $\vec y^{(\tau-1)}$ and could therefore stop after $\tau$ iterations.
However, for the sake of simplifying the analysis, we always continue for $T$ iterations.
\end{rem}
We use the same potential as in the analysis of the full Sinkhorn algorithm, namely $f\colon \R^n \times \R^n \to \R$ given by
\begin{align*}
  f(\vec x, \vec y) = \sum_{i,j=1}^n A_{ij} e^{x_i + y_j} - \sum_{i=1}^n r_i x_i - \sum_{j=1}^n c_j y_j.
\end{align*}
Recall that the potential gap $f(\vec 0, \vec 0) - f^*$ is upper bounded in \cref{lem:potential gap}.

We now show that on expectation we make progress in each iteration.
Recall that $\rho\colon \R_+ \times \R_+ \to [0,\infty]$ is given by $\rho(a \Vert b) = b - a + a \ln \frac{a}{b}$ (with the usual conventions).
\Cref{lem:progress no failure} shows that we can bound the progress of a row or column update in terms of this quantity if we ignore the effects of the truncation.

\begin{lem}
\label{lem:progress no failure}
 Let $\mat A \in \R_+^{n \times n}$ and $\vec x, \vec y \in \R^n$, let $\delta \in [0, 1]$ and $\ell \in [n]$, and let $\vec {\hat x} \in \R^n$ be a vector such that $\abs{\hat x_\ell -\ln(r_\ell / \sum_{j=1}^n A_{\ell j} e^{y_j})}\leq\delta$, and for every $k \neq \ell$, $\hat x_k = x_k$.
  Then
  \[
    f(\vec x, \vec y) - f(\vec {\hat x}, \vec y) \geq \rho\bigl(r_\ell \big\Vert r_\ell(\A x y)\bigr) - 2 \delta r_\ell.
  \]
  A similar statement holds for an update of $y_\ell$ (using $c_\ell$ rather than $r_\ell$).
\end{lem}
The proof is completely analogous to~\cref{lem:progress no failure full}, and left as an exercise for the reader.
We can use \cref{lem:progress no failure} to lower bound the expected progress in each iteration in terms of the \emph{relative entropy} or \emph{Kullback-Leibler divergence} $D\colon \R_+^n \times \R_+^n \to [0,\infty]$, which we recall is defined as $D(\vec a \Vert \vec b) = \sum_{i=1}^n \rho(a_i \Vert b_i)$.
The relative entropy is zero if and only if $\vec a = \vec b$.
To formally analyze the expected progress it will be useful to define the following events.

\begin{definition}
\label{def: events Gt and St}
For $t = 1,\ldots, T$, we define the following events:
\begin{itemize}
    \item Let $G_t$ denote the event that $D(\vec r \Vert \vec r(\mat A^{(t-1)})) \leq \eps$ and $D(\vec c \Vert \vec c(\mat A^{(t-1)})) \leq \eps$.
    \item Let $S_t$ denote the event that the call to {\ApproxScalingFactor} on line~\ref{algline:RSFP row estimate} succeeds (if $\beta=0$ in this iteration) or that the call to {\ApproxScalingFactor} on line~\ref{algline:RSFP col estimate} succeeds (if $\beta=1$ in this iteration).
    \item Define $S$ to be the intersection of the events $S_t$, i.e., $S =\cap_{t=1}^{T} S_t$.
\end{itemize}
\end{definition}
\noindent Let us also abbreviate $f_t = f(\vec x^{(t)}, \vec y^{(t)})$.
\begin{lem}
\label{lem:expected progress when S holds}
Assume $\Pr[S_t] \geq 1 - \eta$ for $t\in[T]$.
Then, for any $t\in[T]$, we have
\begin{align*}
  \Exp[(f_{t-1} - f_{t}) \indic_S] \geq \frac{\eps}{2n} \Pr[\overline{G_t}] - \eps \eta T - \frac{2 \delta}{n}.
\end{align*}
\end{lem}

\begin{proof}
In this proof we use that for independent random variables $X$, $Y$ and any function $h(x,y)$ we have\
\begin{align}
\label{eq:indep}
  \Exp[h(X,Y)] = \sum_{y\in Y} \Pr[Y=y] \, \Exp[h(X,y)].
\end{align}

  For $t\in[T]$, let $\beta^{(t)}$ be the random choice of row versus column scaling made on \cref{algline:RSFP row col choice}, and let $\ell^{(t)} \in [n]$ be the random index chosen on \cref{algline:RSFP index choice}.
  Then, \cref{lem:progress no failure} shows that if $S$ holds then we have
  \begin{align}\label{eq:f vs rho}
    f_{t-1} - f_{t} \geq
    \rho(M_t \Vert m_t) - 2\delta M_t,
    \quad\text{ where }\quad
    (M_t, m_t) =
    \begin{cases}
      (r_{\ell^{(t)}}, r_{\ell^{(t)}}(\mat A^{(t-1)})) & \text{ if $\beta^{(t)} = 0$}, \\
      (c_{\ell^{(t)}}, c_{\ell^{(t)}}(\mat A^{(t-1)})) & \text{ if $\beta^{(t)} = 1$}.
    \end{cases}
  \end{align}
  We define $R_t = \min\{\rho(M_t \Vert m_t),\eps\} - 2 \delta M_t$. Then the above implies
  \begin{align*}
    \Exp[(f_{t-1} - f_t) \indic_S] \geq \Exp[R_t \indic_S].
  \end{align*}
  We may expand this lower bound as
  \begin{equation}
    \label{eq:random sinkhorn exp decomposition}
    \Exp[R_t \indic_S] = \Exp[R_t] - \Exp[R_t \indic_{\overline S}] =  \Exp[R_t \indic_{\overline{G_t}}] + \Exp[R_t \indic_{G_t}] - \Exp[R_t \indic_{\overline S}].
  \end{equation}
  To lower bound the first term, recall that~$(\beta^{(t)}, \ell^{(t)})$ are drawn independently from $\mat A^{(t-1)}$, while the event $G_t$ only depends on $\mat A^{(t-1)}$ (and hence is independent from $(\beta^{(t)}, \ell^{(t)})$).
  Therefore, using \cref{eq:indep} we can lower bound
  \begin{align*}
    \Exp[ R_t \indic_{\overline{G_t}} ]
  &= \Exp[ (\min\{\rho(M_t \Vert m_t), \eps\} - 2 \delta M_t) \indic_{\overline{G_t}} ] \\
  &= \frac1{2n} \Exp\left[ \sum_{\ell=1}^n \left( \min\{\rho(r_\ell \Vert r_\ell(\mat A^{(t-1)})), \eps\} + \min\{\rho(c_\ell \Vert c_\ell(\mat A^{(t-1)})), \eps\} - 2 \delta (r_\ell + c_\ell) \right) \indic_{\overline{G_t}} \right] \\
  &= \frac1{2n} \Exp\left[ \sum_{\ell=1}^n \left( \min\{\rho(r_\ell \Vert r_\ell(\mat A^{(t-1)})), \eps\} + \min\{\rho(c_\ell \Vert c_\ell(\mat A^{(t-1)})), \eps\} \right) \indic_{\overline{G_t}} \right] - \frac1{2n} \Exp[4 \delta \indic_{\overline{G_t}}] \\
  &\geq \frac1{2n} \Exp\left[ \min\left\{ \sum_{\ell=1}^n \rho(r_\ell \Vert r_\ell(\mat A^{(t-1)})) + \rho(c_\ell \Vert c_\ell(\mat A^{(t-1)})), \, \eps \right\} \indic_{\overline{G_t}} \right] - \frac{2 \delta}{n} \Pr[\overline{G_t}] \\
  &= \frac1{2n} \Exp\left[ \min\left\{ D(\vec r \Vert \vec r(\mat A^{(t-1)})) + D(\vec c \Vert \vec c(\mat A^{(t-1)})), \, \eps \right\} \indic_{\overline{G_t}} \right]  - \frac{2 \delta}{n} \Pr[\overline{G_t}] \\
  &= \frac1{2n} \Exp\left[ \eps \indic_{\overline{G_t}} \right] - \frac{2 \delta}{n} \Pr[\overline{G_t}]
  = \left(\frac\eps{2n} - \frac{2 \delta}{n}\right) \Pr[\overline{G_t}]
  \end{align*}
  where we first used the inequality $\sum_{\ell=1}^n \min\{a_i,b\} \geq \min\{\sum_{\ell=1}^n a_i, b\}$, which holds for any real numbers $a_1,\dots,a_n\in\R$ and $b \geq 0$, and we then used that $D(\vec r \Vert \vec r(\mat A^{(t-1)})) + D(\vec c \Vert \vec c(\mat A^{(t-1)})) \geq \eps$ whenever $G_t$ does not hold.

  To lower bound $\Exp[R_t \indic_{G_t}]$, note that we also have the bound $R_t \geq - 2 \delta M_t$ as $\min\{\rho(M_t \Vert m_t), \eps\}$ is non-negative, so again using independence of $(\beta^{(t)}, \ell^{(t)})$ from $G_t$, we obtain
  \begin{align*}
    \Exp[R_t \indic_{G_t}] \geq -\frac{2 \delta}{n} \Pr[G_t]=\frac{2 \delta}{n}\Pr[\overline{G_t}]-\frac{2 \delta}{n}.
  \end{align*}
  Lastly, to upper bound $\Exp[R_t \indic_{\overline{S}}]$, note that $R_t \leq \eps$, so
  \begin{align*}
    \Exp[R_t \indic_{\overline{S}}] \leq \eps \Pr[\overline{S}] \leq \eps \eta T
  \end{align*}
  where the last step follows from the union bound.
  Combining the bounds in~\cref{eq:random sinkhorn exp decomposition} then yields the desired bound.
\end{proof}

\Cref{lem:scaling norm bound} shows how large we need to take $b_1$ to ensure that all components of $(\vec x^{(t)}, \vec y^{(t)})$ are in the interval $[-2^{b_1}, 2^{b_1}]$.\footnote{Technically, \cref{lem:scaling norm bound} is about updating all entries of either $\vec x$ or $\vec y$ in each iteration, however its proof shows that the same bound also applies if we update only a single coordinate per iteration.}
Finally, we can combine \cref{lem:potential gap,lem:expected progress when S holds} to show that \cref{alg:RSFP} returns an $\eps$-relative-entropy-scaling (with probability $\geq 2/3$).
This is also an $O(\sqrt{\eps})$-scaling in $\ell_1$-norm by Pinsker's inequality.

\begin{thm}
\label{thm:random sinkhorn guarantee}
Let $\mat A \in [0, 1]^{n \times n}$ be a matrix whose non-zero entries are at least $\mu > 0$ and $\norm{\mat A}_1 \leq 1$.
Assume that $\mat A$ is asymptotically scalable to $(\vec r, \vec c)$.
Let $p\in(0,1]$ and $\eps \in (0, 1]$.
Choose
\[
T =  \left\lceil\frac{6 n}{\eps p} \ln(\tfrac1\mu)\right\rceil,
\]
$\delta = \frac{\eps p}{12}$,
$\eta = \frac p{6 n T}$,
$b_2 = \lceil \log_2(1/\delta) \rceil$,
and $b_1 = \lceil \log_2(T) + \log_2(\ln(\frac1{\mu}) + 1 + \sigma) \rceil$, where $\sigma = \max(\abs{\ln r_{\min}}, \abs{\ln c_{\min}})$.
Then, \cref{alg:RSFP} with these parameters returns a pair $(\vec x, \vec y)$ such that
$D(\vec r \Vert \vec r(\mat A(\vec x, \vec y))) \leq \eps$ and $D(\vec c \Vert \vec c(\mat A(\vec x, \vec y))) \leq \eps$
with probability at least $1 - p$.
The total cost, measured in the number of calls to {\ApproxScalingFactor}, is $T \in \widetilde O(n / (\eps p))$, leading to a total classical time complexity of $\widetilde O(m / (\eps p))$ on expectation, or a total quantum time complexity of $\widetilde O(\sqrt{mn} / (\eps^2 p^2))$ on expectation.
\end{thm}

\begin{proof}
  By \cref{lem:potential gap}, we have
  \begin{align}\label{eq:upper}
    \Exp[(f_0 - f_T) \indic_S]
  \leq \Exp[\ln(\tfrac1\mu) \indic_S]
  = \ln(\tfrac1\mu) \Pr[S]
  \leq \ln(\tfrac1\mu).
  \end{align}
  using $\mu \leq 1$.
  We now lower bound the left-hand side by a telescoping sum, using \cref{lem:expected progress when S holds},
  \begin{align*}
    \Exp[(f_0 - f_T) \indic_S]
  &= \sum_{t=1}^T \Exp[(f_{t-1} - f_t) \indic_S]  \\
  &\geq \sum_{t=1}^T \left( \frac{\eps}{2n} \Pr[\overline{G_t}] - \eps \eta T - \frac{2 \delta}{n} \right) \\
  &= T \left( \frac{\eps}{2n} - \eps \eta T - \frac{2 \delta}{n} \right) - \frac{\eps}{2n} \sum_{t=1}^T \Pr[G_t].
  \end{align*}
  Together with \cref{eq:upper}, we obtain the bound
  \begin{align}\label{eq:succbound}
     \Pr[G_\tau]
  &= \frac1T \sum_{t=1}^T \Pr[G_t]
  \geq 1 - 2n \eta T - \frac{4\delta}{\eps} - \frac{\frac{2n}{\eps} \ln(\tfrac1\mu)}{T} \\
  \nonumber
 &= 1- \left(2n \cdot \frac{p}{6n (\frac{6 n}{\eps p} \ln(\tfrac1\mu))} \cdot  \frac{6 n}{\eps p} \ln(\tfrac1\mu)\right) - \left(\frac{4 \eps p}{12 \eps}\right) - \left( \frac{\frac{2n}{\eps} \ln(\tfrac1\mu)}{ \frac{6 n}{\eps p} \ln(\tfrac1\mu)}\right) = 1 - p.
  \end{align}
  Since we return $(\vec{x}^{(\tau-1)},\vec{y}^{(\tau-1)})$ for $\tau\in[T]$ independently and uniformly at random, $\Pr[G_\tau]$ is the success probability of our algorithm and the first equality in \cref{eq:succbound} holds.
  This shows the output of~\cref{alg:RSFP} satisfies the guarantees of the theorem.

  Finally, the time complexity of~\cref{alg:RSFP} follows from~\cref{thm:classical ApproxScalingFactor guarantee},~\cref{thm:quantum ApproxScalingFactor guarantee} and an application of Jensen's inequality (using concavity of the square root function). Indeed, the complexity of applying \ApproxScalingFactor to a row or column with $s$ possibly non-zero entries is $\widetilde O(s)$ classically (hiding also $\polylog(1/\delta)$ factors) and $\widetilde O(\sqrt{s}/\delta)$ quantumly, and hence the expected time complexity per application of \ApproxScalingFactor is $\widetilde O(m/n)$ classically (hiding also $\polylog(1/\delta)$ factors) and $\widetilde O(\sqrt{m/n}/\delta)$ quantumly. The total expected time complexity then follows from linearity of expectation and the bound on the number of iterations $T$.
\end{proof}

Note that the number of iterations~$T$ (and hence the runtime) of this algorithm scales inverse-polynomially with the failure probability~$p$. This is not optimal. Instead, to achieve a probability of success $1-p$, we can use  \cref{alg:RSFP} $O(\ln(1/p))$ many times with a constant success probability, say $2/3$, and use \TestScaling to test if we have obtained a good scaling after each run. Note that the resulting (quantum) time complexity is not better than that of the full Sinkhorn algorithm (\cref{sec: full testing}).

\section{Randomized Osborne algorithm for matrix balancing}\label{sec:balancing}
In this section we present an algorithm for the matrix-balancing problem (\cref{prob: balancing}).
The algorithm that we analyze is a quantum implementation of a randomized variant of Osborne's algorithm~\cite{10.1145/321043.321048} that was very recently analyzed by Altschuler and Parrilo~\cite{altschuler2020random}.
The analysis in the quantum setting is more complicated than in the classical case.
The basic argument follows similar lines to the one given in~\cite{altschuler2020random} (a potential argument like for matrix scaling, but with the relative entropy replaced by the Hellinger distance).
However, in the classical setting, one can increase the precision of individual updates at a very small cost, and one does not have to deal with the possibility of making backwards progress.
In the quantum setting, in contrast, we do not have this luxury: we cannot test whether the matrix is $\eps$-$\ell_1$-balanced each iteration, and the relatively high imprecision of the updates can cause subsequent iterations to destroy this property.
This situation is similar to the one discussed in \cref{sec:overview random} for the randomized Sinkhorn algorithm, and we adapt the ideas developed in that section in the analysis here.

Recall that an entrywise non-negative matrix $\mat A \in \R_+^{n \times n}$ is $\eps$-$\ell_1$-balanceable if there exists $\vec x \in \R^n$ such that
\[
    \frac{\norm{\vec r(\mat A(\vec x)) - \vec c(\mat A(\vec x))}_1}{\norm{\mat A(\vec x)}_1} \leq \eps,
\]
where $\mat A(\vec x) =\mat A(\vec x, - \vec x)$ is the matrix whose $(i,j)$-th entry is $A_{ij} e^{x_i - x_j}$.
For convenience, we assume that the diagonal entries of $\mat A$ are zero; and we also assume that every row and every column of $\mat A$ contains at least one non-zero entry.

We first describe Osborne's algorithm for finding an $\eps$-$\ell_1$-balancing. Similarly to Sinkhorn's algorithm for matrix scaling, the idea is to fix the requirement of being $\eps$-$\ell_1$-balanced for individual coordinates, one at a time.
More precisely, given an index $\ell \in [n]$, the update is given by $\vec x' = \vec x + \Step_\ell \vec e_\ell$, where $\Step_\ell$ is chosen such that
\[
    r_\ell(\mat A(\vec x')) = c_\ell(\mat A(\vec x')).
\]
Expanding the above equation and using $A_{\ell \ell} = 0$ yields
\[
    e^{\Step_\ell} \cdot r_\ell(\mat A(\vec x)) = e^{- \Step_\ell} \cdot c_\ell(\mat A(\vec x)).
\]
Since we assume every row and column contains at least one non-zero entry, the above equation has a unique solution, given by
\begin{equation}
    \label{eq:osborne update}
    \Step_\ell = \ln \left(\sqrt{\frac{c_\ell(\mat A(\vec x))}{r_\ell(\mat A(\vec x))}}\right).
\end{equation}

Note that the updates of multiple coordinates \emph{cannot} be done simultaneously, since each coordinate can potentially affect all row and column marginals.
This is in contrast with the Sinkhorn algorithm for matrix scaling, where all rows or all columns can be updated at the same time.
Therefore a choice must be made as to which index to update at each iteration.
Altschuler and Parrilo~\cite{altschuler2020random} analyze several such choices, including greedy, random and cyclic variants.
The greedy and random variants have better guaranteed performance than the cyclic variant.
However, to implement the greedy version, one has to maintain an auxiliary data structure which contains all current row and column marginals, which can be updated after each iteration with a cost of order $O(n)$.
Therefore, even though we can accelerate each iteration with quantum approximate counting to become sublinear in $n$, the greedy iterations would still incur at least a cost of order $n$ for the updates. So instead, we focus on a randomized variant of Osborne's algorithm.

\begin{algorithm}[ht]
  \caption{Random Osborne with finite precision and failure probability}\label{alg:Random Osborne}
  \Input{Oracle access to~$\mat A \in [0,1]^{n \times n}$ and non-zero entries at least $\mu > 0$, iteration count~$T \geq 0$, bit counts~$b_1, b_2 \geq 0$, update precision~$\delta \in (0, 1)$ and subroutine failure probability~$\eta \in [0,1]$}
  \Output{Vector $\vec x \in \R^n$ with entries encoded in~$(b_1, b_2)$ fixed-point format}
  \Guarantee{For $p \in (0,1]$ and $\eps \in (0, 1]$ and parameters as in~\cref{thm:random osborne guarantee}, with probability at least $1 - p$, $\mat A(\vec x)$ is $\eps$-$\ell_1$-balanced}

  $\vec x^{(0)} \leftarrow \vec 0$\tcp*{entries in~$(b_1, b_2)$ fixed-point format}

  \smallskip

  \For{$t\leftarrow 1,2,\dotsc,T$}{
    Pick $\ell \in [n]$ uniformly at random\;
    $\vec x^{(t)} \leftarrow \vec x^{(t-1)}$\;
    $x_\ell^{(t)} \leftarrow \frac{1}{2}$ \ApproxScalingFactor$(\mat A_{\ell \bullet}, 1, -\vec x^{(t-1)}, \delta, b_1, b_2, \eta / 2, \mu) -$ \\ \hspace{1.05cm} $\frac{1}{2}$ \ApproxScalingFactor$(\mat A_{\bullet \ell}, 1, \vec x^{(t-1)}, \delta, b_1, b_2, \eta / 2, \mu)$\;
  }

  Pick $\tau \in [T]$ uniformly at random\;
  \Return{$(\vec{x}^{(\tau)}, \vec{y}^{(\tau)})$}\;
\end{algorithm}

Our randomized version of Osborne's algorithm is given in \cref{alg:Random Osborne}.
It allows for an additive error~$\delta$ in computing the update as compared to \cref{eq:osborne update}.
To see that for $\delta=0$ the update in \cref{alg:Random Osborne} is exactly the same as in \cref{eq:osborne update}, one can rewrite
\begin{align*}
    x_\ell + \Step_\ell & = x_\ell + \ln\left(\sqrt{\frac{c_\ell(\mat A(\vec x))}{r_\ell(\mat A(\vec x))}}\right) \\
    & = x_\ell + \ln\left(\sqrt{\frac{\sum_{i=1}^n A_{i \ell} e^{x_i - x_\ell}}{\sum_{j=1}^n A_{\ell j} e^{x_\ell - x_j}}}\right) \\
    & = \frac{1}{2} \left(\ln\left(\frac1{\sum_{j=1}^n A_{\ell j} e^{- x_j}}\right) - \ln\left(\frac1{\sum_{i=1}^n A_{i \ell} e^{x_i}}\right) \right).
\end{align*}
In each iteration of \cref{alg:Random Osborne}, the two calls to the \ApproxScalingFactor subroutine compute the two logarithms to additive precision~$\delta$.
Hence each iteration computes an approximation to the ideal Osborne update with additive precision~$\delta$ (assuming no errors).

To analyze \cref{alg:Random Osborne}, we use the (convex) potential $f\colon\R^n \to \R$ given by
\[ f(\vec x) = \sum_{i,j=1}^n A_{ij} e^{x_i - x_j} = \norm{\mat A(\vec x)}_1, \]
in analogy with the potential for the analysis of matrix scaling.
Let $f^*$ be the infimum of~$f(\vec x)$.
We first prove a lower bound on the potential.
\begin{lem} \label{lem: balancing potential}
  If $\mat A \in \R_+^{n \times n}$ is asymptotically balanceable and its non-zero entries are at least $\mu > 0$, then $f^* \geq \mu$.
\end{lem}
\begin{proof}
  Suppose $\mat A$ is asymptotically balanceable.
  Then for every $\eps > 0$, there exists $\vec x \in \R^n$ such that $\mat A(\vec x)$ is $\eps$-$\ell_1$-balanced, i.e.,
  \[ \frac{\norm{\vec r(\mat A(\vec x)) - \vec c(\mat A(\vec x))}_1}{\norm{\mat A(\vec x)}_1} \leq \eps. \]
  Now observe that
  \begin{align*}
      \frac{\vec r(\mat A(\vec x)) - \vec c(\mat A(\vec x))}{\norm{\mat A(\vec x)}_1} = \frac{\sum_{i,j=1}^n A_{ij} e^{x_i - x_j} (\vec e_i - \vec e_j)}{\sum_{i,j=1}^n A_{ij} e^{x_i - x_j}},
  \end{align*}
  which is a convex combination of vectors in  $W = \{\vec e_i - \vec e_j : A_{ij} > 0\}$.
  By taking the limit $\eps \to 0$, this implies that $\vec 0 \in \R^n$ is in the convex hull of $W$, since the convex hull of finitely many points is closed.
  By Farkas' lemma, this implies that for every $\vec x \in \R^n$, there exists some $\vec e_k - \vec e_\ell \in W$ (so~$A_{k\ell}>0$) such that $x_k-x_\ell\geq 0$, and in particular, we see that
  \begin{align*}
      f(\vec x) = \norm{\mat A(\vec x)}_1 = \sum_{i,j=1}^n A_{ij} e^{x_i - x_j} \geq A_{k\ell} e^{x_k - x_\ell} \geq A_{k\ell}.
  \end{align*}
  Since $A_{k\ell} > 0$, we have $A_{k\ell} \geq \mu$, so $f(\vec x) \geq \mu$ for all $\vec x \in \R^n$, and we obtain the bound $f^* \geq \mu$.
\end{proof}

The next lemma gives a lower bound on the progress made by an approximate Osborne update.

\begin{lem}
  \label{lem:osborne progress lower bound}
  Let $\ell \in [n]$ be an index, $\vec x\in\R^n$ be a vector, $\delta \in [0, 1]$, and let $\vec x'$ be the vector with $x_k' = x_k$ for $k \neq \ell$ and
  \begin{equation}\label{eq:delta approx osborne condition}
    \left|x_\ell' - \left(x_\ell + \ln\left(\sqrt{\frac{c_\ell(\mat A(\vec x))}{r_\ell(\mat A(\vec x))}}\right)\right)\right|\leq \delta,
  \end{equation}
  i.e., $\vec x'$ is a $\delta$-additive approximation of the Osborne update of $\vec x$ for the $\ell$-th index.
  Then
  \begin{align*}
    f(\vec x) - f(\vec x') \geq \left( \sqrt{r_\ell(\mat A(\vec x))} - \sqrt{c_\ell(\mat A(\vec x))} \right)^2 - 2 \delta \sqrt{r_\ell(\mat A(\vec x)) c_\ell(\mat A(\vec x))}.
  \end{align*}
\end{lem}
\begin{proof}
  Note that $\mat A(\vec x')$ and $\mat A(\vec x)$ have the same entries outside of the $\ell$-th row and column.
  Expanding the definition and recalling that $A_{\ell\ell}=0$ gives
  \[
    f(\vec x) - f(\vec x') = r_\ell(\mat A(\vec x)) + c_\ell(\mat A(\vec x)) - r_\ell(\mat A(\vec x')) - c_\ell(\mat A(\vec x')).
  \]
  For convenience, write $z = x_\ell' - x_\ell - \ln(\sqrt{c_\ell(\mat A (\vec x)) / r_\ell(\mat A(\vec x))})\in[-\delta,\delta]$.
  Then
  \begin{align*}
    r_\ell(\mat A(\vec x')) &
    = e^{x'_\ell - x_\ell} \cdot r_\ell(\mat A(\vec x))
    = e^z \cdot \sqrt{r_\ell(\mat A(\vec x)) c_\ell(\mat A(\vec x))}, \\
    c_\ell(\mat A(\vec x')) &
    = e^{x_\ell - x'_\ell} \cdot c_\ell(\mat A(\vec x))
    = e^{-z} \cdot \sqrt{r_\ell(\mat A(\vec x)) c_\ell(\mat A(\vec x))}.
  \end{align*}
  Since $\abs z \leq \abs \delta \leq 1$, we have the estimate $e^z + e^{-z} \leq 2 + 2 \abs{z} \leq 2 + 2 \delta$, which yields
  \begin{align*}
    f(\vec x) - f(\vec x') & \geq r_\ell(\mat A(\vec x)) + c_\ell(\mat A(\vec x)) - (2 + 2 \delta) \sqrt{r_\ell(\mat A(\vec x)) c_\ell(\mat A(\vec x))}
  \end{align*}
  as desired.
\end{proof}

The following lemma gives a sufficient choice for the parameter $b_1$ to ensure correct functioning of the algorithm, i.e., that the requirements of \ApproxScalingFactor are always satisfied.
\begin{lem}
  Let $\vec x^{(0)} = \vec 0$ and $T\in\N$.
  Suppose that for all $t\in[T]$, $\vec x^{(t)}$ is a $\delta$-additive approximation of an Osborne update for $\vec x^{(t-1)}$ as in \cref{eq:delta approx osborne condition}.
  Then we have the following bound for all~$t\leq T$,
  \begin{equation}\label{eq:norm bound osborne}
    \norm{\vec x^{(t)}}_\infty \leq t \cdot \left( \frac12 \ln(\norm{\mat A}_1 / \mu) + \delta \right).
  \end{equation}
  Thus the choice $b_1 = \lceil \log_2(\sigma + T \cdot (\ln(\norm{\mat A}_1 / \mu)/2 + 1)) \rceil$, where $\sigma = \max_\ell \{ \abs{\ln(r_\ell(\mat A))}, \abs{\ln(c_\ell(\mat A))} \}$, guarantees that $2^{b_1} \geq \abs{\ln(\sum_{j=1}^n A_{\ell j} e^{-x_j^{(t)}})}$ and $2^{b_1} \geq \abs{\ln(\sum_{i=1}^n A_{i\ell} e^{x_i^{(t)}})}$ for any $\ell \in [n]$ and $t \leq T$.
\end{lem}
\begin{proof}
  We show that if $\vec x$ is any vector, then any vector $\vec x'$ obtained by a $\delta$-additive approximate Osborne update of~$\vec x$ satisfies
  \begin{equation}\label{eq:norm bound single osborne}
  \norm{\vec x'}_\infty \leq \norm{\vec x}_\infty + \frac12 \ln(\norm{\mat A}_1 / \mu) + \delta.
  \end{equation}
  Suppose $\vec x'$ is obtained by updating the $\ell$-th index of $\vec x$, i.e.,
  \[ \left|x_\ell' - \left(x_\ell + \ln \sqrt{\frac{c_\ell(\mat A(\vec x))}{r_\ell(\mat A(\vec x))}}\right)\right| \leq \delta. \]
  Then observe that
  \begin{align*}
    x_\ell + \ln \sqrt{\frac{c_\ell(\mat A(\vec x))}{r_\ell(\mat A(\vec x))}}
    & = x_\ell + \ln \sqrt{\frac{\sum_{i=1}^n A_{i \ell} e^{x_i - x_\ell}}{\sum_{j=1}^n A_{\ell j} e^{x_\ell - x_j}}}
    = \ln \sqrt{\frac{\sum_{i=1}^n A_{i \ell} e^{x_i}}{\sum_{j=1}^n A_{\ell j} e^{- x_j}}}.
  \end{align*}
  Since we have
  \[
    \sqrt{\frac{\mu}{\norm{\mat A}_1}} \cdot e^{-\norm{\vec x}_\infty} \leq \sqrt{\frac{\sum_{i=1}^n A_{i \ell} e^{x_i}}{\sum_{j=1}^n A_{\ell j} e^{- x_j}}} \leq \sqrt{\frac{\norm{\mat A}_1}{\mu}} \cdot e^{\norm{\vec x}_\infty},
  \]
  the updated coordinate $x_\ell'$ satisfies
  \[
    \abs{x_\ell'} \leq \norm{\vec x}_\infty + \frac12 \ln(\norm{\mat A}_1 / \mu) + \delta.
  \]
  Since all other coordinates of $\vec x'$ and $\vec x$ agree and $\norm{\mat A}_1 \geq \mu$, the same upper bound holds for~$\|\vec x'\|_\infty$.
  Thus we have proved \cref{eq:norm bound single osborne}, and \cref{eq:norm bound osborne} now follows by induction.

  As a consequence of \cref{eq:norm bound osborne}, for every $t\leq T$ and $\ell \in [n]$, we have that
  \begin{align*}
    \left|\ln \left(\sum_{j=1}^n A_{\ell j} e^{-x_j^{(t)}}\right) - \ln(r_\ell(\mat A))\right| &\leq \norm{\vec x^{(t)}}_\infty \leq t \cdot \left( \frac12 \ln(\norm{\mat A}_1 / \mu) + \delta \right), \\
    \left|\ln \left(\sum_{i=1}^n A_{i \ell} e^{x_i^{(t)}}\right) - \ln(c_\ell(\mat A))\right| &\leq \norm{\vec x^{(t)}}_\infty \leq t \cdot \left( \frac12 \ln(\norm{\mat A}_1 / \mu) + \delta \right).
  \end{align*}
  This implies the second statement.
\end{proof}

We now show that \cref{alg:Random Osborne} finds approximate balancings in a certain number of iterations.

\begin{prp}\label{thm:random osborne guarantee}
Let $\mat A \in [0, 1]^{n \times n}$ be a rational matrix with all zeroes on the diagonal, each row and column containing at least one non-zero element and all non-zero entries at least $\mu > 0$.
Assume~$\mat A$ is asymptotically balanceable, and let $\eps \in (0, 1]$ and $p \in (0, 1]$.
Choose
\[ T = \left\lceil\frac{12 n \ln(\norm{\mat A}_1 / \mu)}{p \eps^2}\right\rceil,
\]
as well as $\delta = p \eps^2 / 24$,
$\eta = p \eps^2 / (12 n T)$,
$b_1 = \lceil \log_2(\sigma + T \cdot (\ln(\norm{\mat A}_1 / \mu) / 2 + 1)) \rceil$,
where $\sigma = \max_\ell \{ \abs{\ln(r_\ell(\mat A))}, \abs{\ln(c_\ell(\mat A))} \}$, and $b_2 = \lceil\log_2(1/\delta)\rceil$.
Then \cref{alg:Random Osborne} with these parameters returns a vector~$\vec x$ such that $\mat A(\vec x)$ is $\eps$-$\ell_1$-balanced with probability at least $1 - p$.
\end{prp}
\begin{proof}
  We proceed in analogy with the analysis of randomized Sinkhorn.
  Let $S_t$ denote the event that both calls to \ApproxScalingFactor in the $t$-th iteration of~\cref{alg:Random Osborne} succeed, and let $S$ be the intersection of all these events for $t\in[T]$.
  Then $\Pr[\overline{S_t}] \leq \eta$, since two calls are made and each call fails with probability at most~$\eta/2$, and also $\Pr[\overline{S}] \leq \eta T$ by the union bound.
  Let $G_t$ be the event that $\mat A(\vec x^{(t-1)})$ is $\eps$-$\ell_1$-balanced, and let $\ell^{(t)}$ be the choice of index in the $t$-th iteration.
  For convenience, we abbreviate
  $F(\vec x) = \ln f(\vec x)$,
  $F_t = \ln f_t = \ln f(\vec x^{(t)})$,
  $\mat A^{(t)} = \mat A(\vec x^{(t)})$,
  $\prevrowmarginal = r_{\ell^{(t)}}(\mat A^{(t-1)})$, and
  $\prevcolmarginal = c_{\ell^{(t)}}(\mat A^{(t-1)})$.
  \Cref{lem:osborne progress lower bound} implies that if $S$ holds, then
  \begin{align*}
      f_{t-1} - f_t & \geq \Big(\sqrt{\prevrowmarginal} - \sqrt{\prevcolmarginal}\Big)^2 - 2 \delta \sqrt{\prevrowmarginal \prevcolmarginal}.
  \end{align*}
  Dividing by $f_{t-1}$ and rearranging yields
  \begin{align*}
      \frac{f_t}{f_{t-1}} & \leq 1 - \frac{1}{f_{t-1}}\Big(\sqrt{\prevrowmarginal} - \sqrt{\prevcolmarginal}\Big)^2 + \frac{2 \delta}{f_{t-1}} \sqrt{\prevrowmarginal \prevcolmarginal}.
  \end{align*}
  The quantity on the right-hand side is positive since $f_{t-1} \geq \phi_t + \psi_t > (\sqrt{\phi_t} - \sqrt{\psi_t})^2$ and $\phi_t, \psi_t$ are both positive, so taking logarithms and using the estimate $\ln(1+z) \leq z$ gives
  \begin{align*}
      F_t - F_{t-1} & \leq \frac{1}{f_{t-1}} \left(2 \delta \sqrt{\prevrowmarginal \prevcolmarginal} - \Big(\sqrt{\prevrowmarginal} - \sqrt{\prevcolmarginal}\Big)^2 \right).
  \end{align*}
  Define the random variable $R_t$ by
  \begin{align*}
      R_t = \frac{1}{f_{t-1}} \left(\Big(\sqrt{\prevrowmarginal} - \sqrt{\prevcolmarginal}\Big)^2 - 2 \delta \sqrt{\prevrowmarginal \prevcolmarginal} \right).
  \end{align*}
  Then the above estimates yield
  \begin{equation}\label{eq:osborne progress via R_t}
    \Exp[(F_{t-1} - F_t) \indic_S] \geq \Exp[R_t \indic_S].
  \end{equation}
  We now expand the right-hand side as
  \begin{align*}
      \Exp[R_t \indic_S] = \Exp[R_t \indic_{G_t}] + \Exp[R_t \indic_{\overline{G_t}}] - \Exp[R_t \indic_{\overline{S}}],
  \end{align*}
  and bound the terms individually.
  The first term can be bounded as
  \begin{align*}
      \Exp[R_t \indic_{G_t}]
      & \geq \Exp\left[- \frac{2 \delta}{f_{t-1}} \sqrt{\prevrowmarginal \prevcolmarginal} \indic_{G_t}\right] \\
      & \geq -\Exp\left[ \frac{\delta}{f_{t-1}} (\prevrowmarginal + \prevcolmarginal) \indic_{G_t}\right] \\
      & = -\delta \, \Exp\left[ \frac{r_{\ell^{(t)}}(\mat A^{(t-1)}) + c_{\ell^{(t)}}(\mat A^{(t-1)})}{\norm{\mat A^{(t-1)}}_1} \indic_{G_t}\right] \\
      & = -\delta \, \Exp\left[ \frac1n \sum_{\ell=1}^n \frac{r_{\ell}(\mat A^{(t-1)}) + c_{\ell}(\mat A^{(t-1)})}{\norm{\mat A^{(t-1)}}_1} \indic_{G_t}\right] \\
      & = - \frac{2\delta}{n} \Pr[G_t],
  \end{align*}
  where the first inequality is obtained by discarding a positive term, the second follows from the arithmetic-geometric mean inequality, and the second to last equality follows from the independence of $\ell^{(t)}$ and $\mat A^{(t-1)}$, which allows us to first average over $\ell^{(t)}$ and then over $\mat A^{(t-1)}$, using that $\sum_{\ell=1}^n r_{\ell}(\mat A^{(t-1)})=\sum_{\ell=1}^n c_{\ell}(\mat A^{(t-1)})=\norm{\mat A^{(t-1)}}_1$.

  For the second term, we bound
  \begin{align*}
      \Exp[R_t \indic_{\overline{G_t}}]
      & \geq \Exp\left[\frac{1}{f_{t-1}} \left( \sqrt{\prevrowmarginal} - \sqrt{\prevcolmarginal} \right)^2 \indic_{\overline{G_t}}\right] - \frac{2\delta}{n} \Pr[\overline{G_t}] \\
      & = \Exp\left[\frac{\left(\sqrt{r_{\ell^{(t)}}(\mat A^{(t-1)})} - \sqrt{c_{\ell^{(t)}}(\mat A^{(t-1)})} \right)^2}{\norm{\mat A^{(t-1)}}_1} \indic_{\overline{G_t}}\right] - \frac{2\delta}{n} \Pr[\overline{G_t}] \\
      & = \Exp\left[\frac1n \sum_{\ell=1}^n \frac{\left(\sqrt{r_{\ell}(\mat A^{(t-1)})} - \sqrt{c_{\ell}(\mat A^{(t-1)})} \right)^2}{\norm{\mat A^{(t-1)}}_1} \indic_{\overline{G_t}}\right] - \frac{2\delta}{n} \Pr[\overline{G_t}] \\
      & = \Exp\left[\frac1n \frac{\left\lVert \sqrt{\vec r(\mat A^{(t-1)})} - \sqrt{\vec c(\mat A^{(t-1)})} \right\rVert_2^2}{\norm{\mat A^{(t-1)}}_1} \indic_{\overline{G_t}}\right] - \frac{2\delta}{n} \Pr[\overline{G_t}] \\
      & \geq \Exp\left[\frac1{2n} \frac{\left\lVert \vec r(\mat A^{(t-1)}) - \vec c(\mat A^{(t-1)}) \right\rVert_1^2}{(\|\vec r(\mat A^{(t-1)})\|_1 + \|\vec c(\mat A^{(t-1)})\|_1)\norm{\mat A^{(t-1)}}_1} \indic_{\overline{G_t}}\right] - \frac{2\delta}{n} \Pr[\overline{G_t}] \\
      & = \Exp\left[\frac1{4n} \frac{\left\lVert \vec r(\mat A^{(t-1)}) - \vec c(\mat A^{(t-1)}) \right\rVert_1^2}{\norm{\mat A^{(t-1)}}_1^2} \indic_{\overline{G_t}}\right] - \frac{2\delta}{n} \Pr[\overline{G_t}] \\
      & \geq \frac{\eps^2}{4n} \Pr[\overline{G_t}] - \frac{2\delta}{n} \Pr[\overline{G_t}],
  \end{align*}
  where the first inequality is derived as in the lower bound for $\Exp[R_t \indic_{G_t}]$,
  the second equality holds by independence of $\ell^{(t)}$ from $\mat A^{(t-1)}$, the second inequality is \cref{lem:unnormalized hellinger} (the lower bound on the Hellinger distance), and the last inequality holds since $\mat A^{(t-1)}$ is not $\eps$-$\ell^1$-balanced in the event $\overline{G_t}$.

  Lastly, to estimate $\Exp[R_t \indic_{\overline{S}}]$, note that
  \begin{align*}
    R_t
  \leq \frac{\left( \sqrt{\prevrowmarginal} - \sqrt{\prevcolmarginal} \right)^2}{f_{t-1}}
  \leq \frac{\prevrowmarginal + \prevcolmarginal}{f_{t-1}}
  =  \frac{r_{\ell^{(t)}}(\mat A^{(t-1)}) + c_{\ell^{(t)}}(\mat A^{(t-1)})}{\norm{\mat A^{(t-1)}}_1}
  \leq 1,
  \end{align*}
  where in the last inequality we use that $A_{\ell^{(t)}\ell^{(t)}}=0$.
  Therefore, we obtain
  \begin{align*}
      \Exp[R_t \indic_{\overline{S}}] \leq \Pr[\overline S] \leq \eta T.
  \end{align*}
  Using \cref{eq:osborne progress via R_t} and the three bounds just derived, we have
  \begin{align*}
    \Exp[(F_{t-1} - F_t) \indic_S] \geq \Exp[R_t \indic_S] & = \Exp[R_t \indic_{G_t}] + \Exp[R_t \indic_{\overline{G_t}}] - \Exp[R_t \indic_{\overline{S}}] \\
      & \geq - \frac{2\delta}{n} \Pr[\overline{G_t}] + \frac{\eps^2}{4n} \Pr[\overline{G_t}] - \frac{2\delta}{n} \Pr[G_t] - \eta T \\
      & = \frac{\eps^2}{4n} \Pr[\overline{G_t}] - \frac{2\delta}{n} -  \eta T.
  \end{align*}
  Since we also have $F_0 - F_T \leq \ln(\norm{\mat A}_1)-\ln(\mu)=\ln(\norm{\mat A}_1 / \mu)$ by \cref{lem: balancing potential}, we can finish up the argument in exactly the same manner as for \cref{thm:random sinkhorn guarantee}.
  Indeed, we have a telescoping sum
  \begin{align*}
    \ln(\norm{\mat A}_1 / \mu) & \geq \Exp[(F_0 - F_T) \indic_S] = \sum_{t=1}^T \Exp[(F_{t-1} - F_t) \indic_S] \\
    & \geq \frac{\eps^2}{4n} \sum_{t=1}^T \Pr[\overline{G_t}] - T \left(\frac{2\delta}{n} +  \eta T\right).
  \end{align*}
  In other words, we have
  \[
    \sum_{t=1}^T \Pr[\overline{G_t}] \leq \frac{4}{\eps^2}\left(n \ln(\norm{\mat A}_1 / \mu) + 2\delta T + \eta n T^2\right).
  \]
  Since the choice of stopping time $\tau \in [T]$ is uniformly random, we have
  \begin{align*}
    \Pr[\overline{G_{\tau}}] = \frac1T \sum_{t=1}^T \Pr[\overline{G_t}]
    & \leq \frac{4}{T \eps^2}\left(n \ln(\norm{\mat A}_1 / \mu) + 2 \delta T +  \eta n T^2\right) \\
    & = \frac{4 n \ln(\norm{\mat A}_1 / \mu)}{T\eps^2} + \frac{8 \delta}{\eps^2} + \frac{4 \eta n T}{\eps^2} \leq p,
  \end{align*}
  where the last inequality follows from the choice of parameters.
\end{proof}

\noindent
Note that the number of iterations~$T$ (and hence the runtime) of this algorithm scales inverse-polynomially with the failure probability~$p$. As in the matrix scaling case, this is not optimal. One can implement a procedure similar to \TestScaling to test whether $\mat A(\vec x)$ is $\eps$-$\ell_1$-balanced, with success probability $1 - \eta$, for a quantum cost of $\widetilde O(\ln(1/\eta) \sqrt{mn} / \eps)$. To boost the success probability of \cref{alg:Random Osborne} to $1-p$, we can use  \cref{alg:Random Osborne} $O(\ln(1/p))$ many times with a constant success probability in each run, say $2/3$, and test if we have obtained a good balancing after each run.

Using the quantum implementation of the \ApproxScalingFactor subroutine from~\cref{sec:arithmetic}, we obtain the following performance guarantee for~\cref{alg:Random Osborne} (in a similar fashion as the cost computation in the proof of~\cref{thm:random sinkhorn guarantee}).
\begin{thm}
  \label{cor:random osborne performance}
  Let $\mat A \in [0, 1]^{n \times n}$ be a rational matrix whose non-zero entries are at least $\mu > 0$, with only zeroes on the diagonal, each row and column containing at least one non-zero element, and let $\eps \in (0, 1]$.
  Assume $\mat A$ is asymptotically balanceable.
  Then there exists a quantum algorithm that, given sparse oracle access to $\mat A$, returns with probability $\geq 2/3$ a vector $\vec x \in \R^n$ such that $\mat A(\vec x)$ is $\eps$-$\ell_1$-balanced, for a total time complexity of $\widetilde O(\sqrt{mn} / \eps^4)$ on expectation.
\end{thm}

The ``on expectation'' time complexity here can be converted to worst-case time complexity using Markov's inequality.


\medskip

The Osborne and Sinkhorn algorithms are special cases of a more general algorithm for a more general problem (see, e.g.,~\cite{cmtv17,burgisser2020interior} for details and motivations).
Suppose one is given a matrix~$\mat A \in \R_+^{n \times n}$ and a vector~$\vec d \in \R^n$, and one wishes to find $\vec x$ such that
\[
    r_\ell(\mat A(\vec x)) - c_\ell(\mat A(\vec x)) = d_\ell
\]
for each $\ell\in[n]$.
That is, one prescribes the differences between the row and column sums.
One can again solve this for individual $\ell \in [n]$, by expanding the above equation and solving for $x_\ell$.
It is clear that the case $\vec d=\vec 0$ amounts to the matrix balancing problem and the procedure just described is the Osborne algorithm.
On the other hand, the matrix scaling problem for $\mat B \in \R_+^{n \times n}$ and targets $\vec r, \vec c \in \R^n$ can be modeled by the above problem for the choices
\[
    \mat A = \begin{bmatrix} \mat 0 & \mat B \\ \mat 0 & \mat 0 \end{bmatrix}, \quad \vec d = (\vec r, - \vec c) \in \R^{2n},
\]
so that the first $n$ constraints yield the desired constraints on the row marginals, and the last $n$ yield the desired constraints on the column marginals.
Note that the support of this matrix $\mat A$ is such that we may simultaneously update the first $n$ coordinates (or the last $n$ coordinates) since their updates are independent.
This leads to the Sinkhorn algorithm.
More generally, if $G$ is the directed graph with vertex set $[n]$ and adjacency defined by the support of $\mat A$, then any subset of vertices which form an independent set in $G$ can be updated simultaneously (cf.~\cite[Sec.~2.5 \& App.~B]{altschuler2020random}).

In general, one can give an explicit expression for the updates, and analyze the progress  via the potential $\vec x \mapsto \norm{\mat A(\vec x)}_1 - \ip{\vec d}{\vec x}$.
This potential generalizes the ones we used for matrix scaling and matrix balancing (up to a change of sign for the last $n$ variables in the former case), and it admits similar potential bounds and lower bounds on the progress as derived above.
However, the details of carrying out such an analysis are less clear and we leave this to future work.

\section{A quantum query lower bound for matrix scaling}\label{sec:lower bound}
In this section we will prove an ${\Omega}(\sqrt{mn})$ lower bound on the query complexity of quantum algorithms for $\Theta(1)$-$\ell_1$-scaling to the uniform marginals $(\vec 1/n, \vec 1/n)$. We will do so by showing an $\Omega(n\sqrt{s})$ lower bound on instances with $s$ potentially non-zero entries per row and column (note that with $m=ns$, we have $n\sqrt{s}=\sqrt{mn}$).

\subsection{Partially learning permutations}
We first prove a query lower bound for the problem of learning a permutation `modulo two', given dense or sparse matrix oracle access to the corresponding permutation matrix.
In the setting where in each row and column of an $n \times n$ permutation matrix, the location of the $1$-entry in that row/column is restricted to $s$ possible locations, we will show that solving this problem requires~$\Omega(n \sqrt{s})$ quantum queries, and that this lower bound even holds for only recovering a constant fraction of this information.
This lower bound is also tight, since one can use Grover search on each column to fully recover the permutation matrix using $O(n \sqrt{s})$ quantum queries.

We will use the following definition:
\begin{definition}[Single-bit descriptor]
Let $\sigma \in S_n$ be a permutation.
The \emph{single-bit descriptor} of~$\sigma$ is the bit string $z \in \{0,1\}^n$ with entries $z_i \equiv \sigma(i) \bmod 2$.
\end{definition}

We record here the version of the adversary lower bound that we will use.

\begin{lem}[{\cite[Theorem~6.1]{ambainis:lowerboundsj}}]\label{lem:adv}
Let $f\colon A\subseteq\Sigma^N\rightarrow B$ be a function of $N$ variables, which takes values in some finite set $B$.
Let $X,Y \subseteq A$ be two sets of inputs such that $f(x)\neq f(y)$ if $x\in X$ and $y\in Y$.
Let $R\subseteq X\times Y$ be nonempty, and satisfy:
\begin{itemize}
\item For every $x\in X$, there exist at least $m_X$ different $y\in Y$ such that $(x,y)\in R$.
\item For every $y\in Y$, there exist at least $m_Y$ different $x\in X$ such that $(x,y)\in R$.
\end{itemize}
Let $\ell_{x,i}$ be the number of $y\in Y$ such that $(x,y)\in R$ and $x_i\neq y_i$, and similarly for $\ell_{y,i}$.
Let $\ell_{max} = \max_{i\in [N]} \max_{(x,y)\in R,x_i\neq y_i}\ell_{x,i}\ell_{y,i}$.
Then any algorithm that computes $f$ with success probability $\geq 2/3$ uses $\Omega\left(\sqrt{\frac{m_X m_Y}{\ell_{max}}}\right)$ quantum queries to the input function.
\end{lem}

We will start by proving a lower bound on the query complexity of recovering the single-bit descriptor of a permutation matrix from a dense matrix oracle. A dense matrix oracle is an oracle $O_{\mat A}$, without the oracles $O_I^{\mathrm{row}}$ and $O_I^{\mathrm{col}}$ for the non-zero entries.

\begin{lem}\label{lem:upperblockham}
  Let $n$ be a positive multiple of $4$.
  Given an $n\times n$ permutation matrix $\mat P$ corresponding to a permutation $\sigma$ (i.e., $P_{ij} = \delta_{i,\sigma(j)}$ for $i,j\in[n]$),
  recovering the single-bit descriptor~$z$ of $\sigma$, with success probability at least $2/3$, requires $\Omega(n\sqrt{n})$ quantum queries to a dense matrix oracle for $\mat P$.
\end{lem}
\begin{proof}
  We consider the problem of determining the number of $1$s in $\mat P$ that are both in an odd row and an odd column (we call such an entry an `odd-odd entry').
  Since this problem can clearly be solved using $z$, a lower bound on this problem will prove the lemma.

  We will use the adversary bound (\cref{lem:adv}), with $\Sigma = \{0,1\}$, $N = n^2$, $A$ the set of $n \times n$ permutation matrices, $B = \{0,\dots,n\}$, and $f$ the function that assigns to a permutation matrix~$\mat P$ the number of odd-odd entries.
  We distinguish two kinds of inputs:
  the set $X$ of those permutation matrices for which there are $n/4$ odd-odd entries in the matrix, and the set $Y$ of those permutation matrices for which there are $n/4+1$ odd-odd entries.
  Finally, define the relation $R \subseteq X\times Y$ by letting $(x,y)\in R$ if and only if~$y$ can be obtained from $x$ by swapping two columns.

  For a swap to change an input $x\in X$ to an input $y\in Y$ (or vice versa), one of the swapped columns needs to be odd, and the other needs to be even.
  For a swap that brings an input $x\in X$ to some input in $Y$, we can pick $n/4$ odd columns (those with only $0$s in odd rows), and $n/4$ even columns (those with a $1$ in an odd row).
  Hence $m_X = n^2/16$.
  Similarly, for the swaps that brings an input $y\in Y$ to an input in $X$, we have $m_Y = (n/4+1)^2$ possibilities.

  Now we compute $\ell_{max}=\max_{i,j \in [n]} \max_{(x,y)\in R, x_{ij} \neq y_{ij}} \ell_{x,(i,j)}\ell_{y,(i,j)}$.
  For any choice of $(i,j)$ and $(x,y)\in R$ there are two possibilities: either $x_{ij} = 1$ and $y_{ij} = 0$, or the other way around.
  We will only focus on the case where $i$ is odd, as the even case works analogously.
  If $x_{ij} = 1$ (and $y_{ij} = 0$) then $j$ has to be even, since $y$ has more odd-odd entries.
  There will be $n/4$ different $y'\in Y$ in relation with $x$ for which the $j$-th column is swapped (one for each odd column in $x$ with a $1$ in an even row), therefore $\ell_{x,(i,j)} = n/4$.
  There is, however, only a single $x$ in relation to $y$ that is distinguished by an $(i,j)$ query: a swap would have to be between the $j$-th column and the one containing a $1$ in the $i$-th row of $y$, so $\ell_{y,(i,j)} = 1$.
  Using a similar argument we get that, if $x_{ij} = 0$ (and $y_{ij} = 1$) then $j$ has to be odd, $\ell_{x,(i,j)} = 1$, and $\ell_{y,(i,j)} = n/4+1$.
  So $\ell_{max} = n/4+1$.
  In conclusion, we get a lower bound of $\Omega\left(\sqrt{\frac{m_X m_Y}{\ell_{max}}}\right) = \Omega(n\sqrt{n})$ for finding the number of $1$s that are both in odd rows and odd columns, and hence the lemma follows.
\end{proof}

The following lemma is well known and easy to prove.

\begin{lem}
Recovering an $n$-bit string (with success probability $\geq 2/3$) takes $\Omega(n)$ quantum queries t the entries of that string.
\end{lem}

We now combine the above two lemmas to show that learning the single-bit descriptor in the sparse access model requires $\Omega(n\sqrt{s})$ quantum queries.
Below we use a special case of the sparse access model defined in~\cref{sec:preliminaries} where each row and column has the same number~$s$ of potentially non-zero elements; we will refer to this as the \emph{$s$-sparse oracle model}.
Note that for this model we have $m=ns$.

\begin{cor} \label{cor:sparseonebit}
    Let $s$ be a positive multiple of $4$, and let $n$ be a multiple of $s$.
  Given an $n\times n$ permutation matrix $\mat P$ corresponding to a permutation $\sigma$,
  recovering the single-bit descriptor~$z$ of $\sigma$, with success probability at least $2/3$, requires $\Omega(n\sqrt{s})$ quantum queries to an $s$-sparse matrix oracle for $\mat P$.
\end{cor}
\begin{proof}
  We use the multiplicativity of quantum query complexity under composition~\cite{reichardt:tight}.
  Consider $n/s$ permutation matrices $\mat P^{(1)},\ldots, \mat P^{({n/s)}}$, each of size $s \times s$ with the promise that each permutation is in one of the two sets $X$ or $Y$ from the proof of \cref{lem:upperblockham}.
  Consider the bit string~$x$ of length $n/s$ whose $j$-th bit indicates which of the two sets $\mat P^{(j)}$ belongs to.
  Finally consider the $n \times n$ permutation matrix $\mat P = \oplus_{j=1}^{n/s} \mat P^{(j)}$, which has the different $\mat P^{(j)}$'s as its diagonal blocks.
  This matrix~$\mat P$ naturally has an $s$-sparse oracle corresponding to dense access to the $\mat P^{(j)}$'s.
  Finding the single-bit descriptor $z$ of $\mat P$ allows us to learn the full string $x$, and thus (due to the multiplicativity of quantum query complexity) takes at least $\Omega(\frac{n}{s}\cdot s\sqrt{s}) = \Omega(n\sqrt{s})$ queries to the $s$-sparse oracle for~$\mat P$.
\end{proof}

We now show that learning only a constant fraction of the entries of the single-bit descriptor of a permutation correctly still requires $\Omega(n\sqrt{s})$ quantum queries.

\begin{lem}\label{lem:partonebit}
  Let $s$ be a positive multiple of $4$, and let $n$ be a positive multiple of $s$.
  There exists a constant $\lambda \in (0,1/3]$ such that, given an $n\times n$ permutation matrix~$\mat P$ corresponding to a permutation $\sigma$,
  finding a $\tilde{z}$ that agrees with the single-bit descriptor $z$ of $\sigma$ on a $1-\lambda$ fraction of the elements (i.e., $|\tilde{z} -z |\leq \lambda n$), with success probability at least $2/3$, requires $\Omega(n\sqrt{s})$ quantum queries to an $s$-sparse matrix oracle for $\mat P$.
\end{lem}
\begin{proof}
Let $N, S \in \mathbb{N}$ and $c_1\in \mathbb{R}_{>0}$ be the constants from the big-$\Omega$ notation in \cref{cor:sparseonebit}, i.e., for  any $n>N$ and $s>S$, at least $c_1 n \sqrt{s}$ quantum queries to an $s$-sparse oracle are required to fully recover $z$ with success probability at least $2/3$.

Let $c_2$ be a constant such that repeated Grover search over a search space of size $a$, with at most~$b$ marked elements, uses at most $c_2 \sqrt{ab}$ queries to find all marked elements with probability~$1$.\footnote{This bound can be obtained for instance by running exact versions of Grover~\cite{grover1996QSearch} assuming there are $t$ solutions, for $t$ going down from $b$ to 1, and removing solutions when they are found. This has overall complexity $\sum_{t=1}^b \bigO{\sqrt{a/t}}=\bigO{\sqrt{ab}}$. It is not hard to see that if the number of solutions is at most $b$ then this algorithm will find all solutions with certainty (ignoring the negligible errors that come from the fact that with a finite gate set we cannot implement every rotation perfectly).}

Let $\lambda = \min\{\tfrac{1}{3},\tfrac{c_1^2}{4 c_2^2}\}$.
Let $\mathcal{A}$ be an algorithm that uses at most $T$ queries to an $s$-sparse oracle for the input $\mat P$ and outputs an $n$-bit string $\tilde{z}$ that (with probability $\geq 2/3$) agrees with the single-bit descriptor $z$ for $\mat P$ on all but $\leq\lambda n$ elements.
Searching for the elements where these strings do not agree can be done by Grover searching through all $ns$ possible non-zero positions in the matrix, and marking those $(i,j)$ where $P_{ij} = 1$ and $\tilde{z}_j \not\equiv i \bmod 2$. There are at most $\lambda n$ such elements, so we can find all of them using $c_2 n\sqrt{\lambda s}$ queries to $\mat P$. We can now flip the erroneous bits in $\tilde{z}$ to recover $z$ fully.

Hence $z$ can be identified exactly with error probability at most $1/3$ using $T+c_2n\sqrt{\lambda s}$ queries.
It follows from \cref{cor:sparseonebit} that (if $n>N$ and $s>S$)
\[
T+c_2 n \sqrt{\lambda s} \geq c_1 n\sqrt{s}.
\]
So
\begin{align*}
T &\geq c_1 n\sqrt{s} -c_2 n \sqrt{\lambda s}\\
& \geq c_1 n\sqrt{s} - c_2 \sqrt{\frac{c_1^2}{4c_2^2}}n  \sqrt{s}\\
&= c_1 n \sqrt{s}/2. \qedhere
\end{align*}
\end{proof}

\subsection{Matrix scaling}
To obtain a query lower bound for matrix scaling, we will reduce the problem of learning the single-bit descriptor of a permutation to the scaling problem, by replacing each $1$-entry of the permutation matrix by one of two $2 \times 2$ gadget matrices (and each $0$-entry by the $2\times 2$ all-$0$ matrix).
These gadget matrices are chosen such that we can determine the single-bit descriptor from the column-scaling vectors $\vec y$ of an $\Theta(1)$-$\ell_1$-scaling to uniform marginals.

We will use the following gadget matrices:
\begin{align*}
  \mat B_0 =
  \begin{bmatrix}
    \tfrac29 & \tfrac49 \\
    \tfrac19 & \tfrac29
  \end{bmatrix}, \qquad
  \mat B_1 =
  \begin{bmatrix}
    \tfrac49 & \tfrac29 \\
    \tfrac29 & \tfrac19
  \end{bmatrix},
\end{align*}
Note that the two matrices have the same columns, but in reverse order.

\begin{lem}\label{lem:gadgets}
  The matrices $\mat B_0, \mat B_1 \in \{\frac19,\tfrac29,\tfrac49\}^{2\times 2}$ are entrywise positive, with entries summing to one, and they are exactly scalable to uniform marginals.
  For $i\in\{0,1\}$, let $(\vec x,\vec y)$ be $\frac18$-$\ell_1$-scaling vectors for $\mat B_i$ to uniform marginals.
  If~$i=0$ then $y_1 - y_2 > 0.18$, while if $i=1$ then $y_1 - y_2 <-0.18$.
  Moreover, $(\vec x,(y_2,y_1))$ are $\frac18$-$\ell_1$-scaling vectors for $\mat B_{1-i}$ to uniform marginals.

  In other words, the matrices can be distinguished just by learning the column-scaling vectors, but they have the same set of possible row-scaling vectors.
\end{lem}
\begin{proof}
Since one matrix is obtained by swapping the columns of the other, the last claim is immediately clear, and it suffices to prove the remaining claims for~$\mat B_0$.

First, we note that $\mat B_0$ is exactly scalable, since
\begin{align*}
  \begin{bmatrix}
    \frac34 & 0 \\
    0 & \frac32
  \end{bmatrix}
  \begin{bmatrix}
    \tfrac29 & \tfrac49 \\
    \tfrac19 & \tfrac29
  \end{bmatrix}
  \begin{bmatrix}
    \frac32 & 0 \\
    0 & \frac34
  \end{bmatrix}
= \begin{bmatrix}
  \frac14 & \frac14 \\
  \frac14 & \frac14
\end{bmatrix}
\end{align*}
has uniform marginals.
Now suppose that $(\vec x,\vec y)$ is an $\frac18$-$\ell_1$-scaling of $\mat B_0$ to uniform marginals.
By the requirement on the column marginals, we have
\begin{align*}
  \left( \frac29 e^{x_1} + \frac19 e^{x_2} \right) e^{y_1} \geq \frac12 - \frac18 \quad\text{and}\quad
  \left( \frac49 e^{x_1} + \frac29 e^{x_2} \right) e^{y_2} \leq \frac12 + \frac18.
\end{align*}
By dividing the first inequality by the second one we get
\begin{align*}
  \frac12 \cdot \frac{e^{y_1}}{e^{y_2}} \geq \frac35,
\end{align*}
and so
\begin{equation*}
  y_1 - y_2 \geq \ln\frac65 > 0.18.
  \qedhere
\end{equation*}
\end{proof}
We now prove our query lower bound for matrix scaling. This lower bound also holds if we allow the algorithm to know the set of vectors $\vec x$ that can occur in an $\eps$-$\ell_1$-scaling $(\vec x, \vec y)$ of the matrix.
\begin{thm}\label{thm:scaling lower bound}
Let $s$ be a positive multiple of $8$ and let $n$ be a multiple of $s$.
There exists a constant~$\eps\in (0,1)$ such that any quantum algorithm which, given an $s$-sparse oracle for an $n\times n$-matrix with entries summing to one that is exactly scalable to uniform marginals, returns an $\eps$-$\ell_1$-scaling with probability $\geq 2/3$, requires $\Omega(n \sqrt{s})$ quantum queries to the oracle.
\end{thm}

\begin{proof}
Let $\eps = \frac\lambda{16} \in \Theta(1)$, where $\lambda$ is the constant from~\cref{lem:partonebit}.
Assume there is a $T$-query quantum algorithm~$\mathcal{A}$ that solves the $\eps$-$\ell_1$-scaling problem with success probability $\geq 2/3$.
We will construct an algorithm for recovering a $1-\lambda$ fraction of the $n/2$-bit single-bit descriptor~$z$ of an unknown permutation in $S_{n/2}$, given $s/2$-sparse oracle access to the corresponding $n/2\times n/2$ permutation matrix $\mat P$.

Let $\mat Q \in \mathbb{R}^{n\times n}$ be a $n/2 \times n/2$ block-matrix with the following $2 \times 2$ blocks:
 \begin{itemize}
     \item If $P_{ij} = 0$ then the $(i,j)$-block in $\mat Q$ is the $2\times 2$ all-$0$ matrix.
     \item If $P_{ij} = 1$ then the $(i,j)$-block in $\mat Q$ is the matrix $\mat{B}_{i \bmod 2}$ from \cref{lem:gadgets}.
 \end{itemize}

Finally, let $\mat R = \frac2n \mat Q$.
Then $\mat R$ is an $n\times n$ matrix that has entries summing to one, and it is exactly scalable to uniform marginals. Note that an $s$-sparse matrix oracle for $\mat R$ can be constructed using $\bigO{1}$ queries to an $s/2$-sparse oracle for $\mat P$.

By applying $\mathcal{A}$ to $\mat R$ we obtain $\eps$-$\ell_1$-scaling vectors $(\vec x,\vec y)$.
Let~$\vec x^{(i)} \in \R^2$ (resp.~$\vec y^{(j)}$) denote the restriction of $\vec x$ (resp.~$\vec y$) to the two coordinates corresponding to the $i$-th row (resp.~the $j$-th column) of $\mat P$.
We can compute the $\ell_1$-error of the row and column marginals in terms of the $2\times 2$-blocks $\frac{2}{n}\mat B_{i\bmod 2}$ corresponding to $P_{ij} = 1$.
Accordingly, we obtain
\begin{align*}
  \sum_{(i,j) : P_{ij} = 1}  \norm{\vec r(e^{\vec x^{(i)}} \frac{2}{n}\mat B_{i \bmod 2} e^{\vec y^{(j)}}) - \frac{\vec 1} n}_1 \leq \eps
  \quad\text{and}\quad
  \sum_{(i,j) : P_{ij} = 1}  \norm{\vec c(e^{\vec x^{(i)}} \frac2n \mat B_{i \bmod 2} e^{\vec y^{(j)}}) - \frac{\vec 1} n}_1 \leq \eps,
\end{align*}
and hence
\begin{align*}
  \sum_{(i,j) : P_{ij} = 1}  \norm{\vec r(e^{\vec x^{(i)}} \mat B_{i \bmod 2} e^{\vec y^{(j)}}) - \frac{\vec 1} 2}_1 + \norm{\vec c(e^{\vec x^{(i)}} \mat B_{i \bmod 2} e^{\vec y^{(j)}}) - \frac{\vec 1} 2}_1 \leq \eps n.
\end{align*}
This implies that for all but $\frac{\eps n}{1/8} = \lambda n/2$ of these blocks, the corresponding~$(\vec x^{(i)}, \vec y^{(j)})$ are $\frac18$-$\ell_1$-scaling vectors.
By \cref{lem:gadgets}, we can therefore correctly identify $i\bmod 2$ from $\vec y^{(j)}$ for  $(1-\lambda)n/2$ columns, and hence we learn that $z_j = i \bmod 2$ for those $j$.
Since this identifies $(1-\lambda)n/2$ of the $n/2$ elements of~$z$, and because the algorithm has error probability at most $1/3$, it follows from \cref{lem:partonebit} that $T \in \Omega(n\sqrt{s})$.
\end{proof}

\begin{cor}
  There exists a constant $\eps\in (0,1)$ such that any quantum algorithm which, given a sparse oracle for an $n\times n$-matrix that is exactly scalable to uniform marginals and has $m$ potentially non-zero entries which sum to~$1$, returns an $\eps$-$\ell^1$-scaling with probability $\geq 2/3$, requires $\Omega(\sqrt{mn})$ quantum queries to the oracle.
\end{cor}

\bibliographystyle{alpha}
\bibliography{references}

\appendix

\section{Potential bound for matrix scaling}\label{sec:potential bound}
The purpose of this section is to prove the following result.
Recall that we write $\mat A(\vec x, \vec y)$ for the matrix whose $(i,j)$-th entry is $A_{ij} e^{x_i + y_j}$.
For vectors $\vec x,\vec y\in\R^n$, denote by $\ip{\vec x}{\vec y} =\sum_{i=1}^n x_i y_i$ their inner product.

\begin{prp}
\label{thm:generalpotentialbound}
  Let $\mat A \in \R_+^{n \times n}$ be a non-zero matrix with non-negative entries, and let $\vec r, \vec c \in \R_+^n$ such that $\norm{\vec r}_1 = \norm{\vec c}_1 = 1$.
  Then the following statements are equivalent:
  \begin{enumerate}[label=(\roman*)]
  \item\label{item:f} The function $f\colon \R^n \times \R^n \to \R$ given by
    \[
      f(\vec x, \vec y) = \sum_{i,j = 1}^n A_{ij} e^{x_i + y_j} - \ip{\vec r}{\vec x} - \ip{\vec c}{\vec y}
    \]
    is bounded from below.
  \item\label{item:F} The function $F\colon \R^n \times \R^n \to \R$ given by
    \[
      F(\vec x, \vec y) = \ln\left(\sum_{i,j = 1}^n A_{ij} e^{x_i + y_j}\right) - \ip{\vec r}{\vec x} - \ip{\vec c}{\vec y}
    \]
    is bounded from below.
  \item\label{item:scalable} The matrix $\mat A$ is asymptotically $(\vec r, \vec c)$-scalable.
  \item\label{item:conv} The point $(\vec r, \vec c)$ is in the convex hull of the set
  \begin{equation*}
    \Omega = \{ \vec\omega_{ij} = (\vec e_i, \vec e_j) : A_{ij} > 0 \} \subseteq \R^n \times \R^n.
  \end{equation*}
  \end{enumerate}
  Furthermore, if any of these conditions hold, then
  \begin{equation}\label{eq:generalpotentialbound}
    f(\vec 0, \vec 0) - \inf_{\vec x, \vec y \in \R^n} \, f(\vec x, \vec y) \leq \norm{\mat A}_1 - 1 + \ln(1/\mu)
  \end{equation}
  where $\mu$ is the smallest non-zero entry of $\mat A$.
\end{prp}

\noindent
The equivalence of the conditions is well-known for general $\mat A$, and the same goes for the bound in \cref{eq:generalpotentialbound} when $\mat A$ has \emph{strictly positive} entries.
For general $\mat A$, this bound appears to be new.

The equivalence will be proved by a series of lemmas.
We first establish the equivalence between \ref{item:f} and \ref{item:F}.

\begin{lem}
  \label{lem:equality f F inf}
  For any $\vec x, \vec y \in \R^n$, we have the equality
  \[
    \min_{t \in \R} \, f(\vec x + t \vec 1, \vec y) = 1 + F(\vec x, \vec y).
  \]
\end{lem}
\begin{proof}
  Consider the function $g: \R \to \R$ given by $g(t) = f(\vec x + t \vec 1, \vec y)$.
  Then
  \begin{align*}
    g(t) & = \sum_{i,j=1}^n A_{ij} e^{x_i + y_j + t} - \ip{\vec r}{\vec x} - t \ip{\vec r}{\vec 1} - \ip{\vec c}{\vec y} \\
    & = \left( \sum_{i,j=1}^n A_{ij} e^{x_i + y_j} \right) e^t - \ip{\vec r}{\vec x} - \ip{\vec c}{\vec y} - t
  \end{align*}
  since $\ip{\vec r}{\vec 1} = \norm{\vec r}_1 = 1$.
  From this expression, it is clear that $g(t)$ is strictly convex, and attains its minimum at $t^* \in \R$ such that
  \[
    g'(t^*) = \left( \sum_{i,j=1}^n A_{ij} e^{x_i + y_j} \right) e^{t^*} - 1 = 0,
  \]
  i.e.,
  \[
    t^* = - \ln\left( \sum_{i,j=1}^n A_{ij} e^{x_i + y_j} \right).
  \]
  Consequently, we see that
  \[
    \min_{t \in \R} g(t) = g(t^*) = 1 - \ip{\vec r}{\vec x} - \ip{\vec c}{\vec y} + \ln\left(\sum_{i,j=1}^n A_{ij} e^{x_i + y_j}\right) = 1 + F(\vec x, \vec y)
  \]
  as desired.
\end{proof}

\begin{cor}
  \label{cor:equality f F inf}
  We have
  \[
    \inf_{\vec x, \vec y \in \R^n} \, f(\vec x, \vec y) = 1 + \inf_{\vec x, \vec y \in \R^n} \, F(\vec x, \vec y).
  \]
  In particular, $f$ is bounded from below if and only if $F$ is.
\end{cor}

The rest of the equivalences we prove by showing that~\ref{item:scalable} implies~\ref{item:conv}, which implies~\ref{item:F}, which in turn implies \ref{item:scalable}.
We first show that~\ref{item:scalable} implies~\ref{item:conv}.
\begin{lem}
  Let $\mat A$ be asymptotically $(\vec r, \vec c)$-scalable for $\norm{\vec r}_1 = \norm{\vec c}_1 = 1$.
  Then, $(\vec r, \vec c) \in \conv \Omega$.
\end{lem}
\begin{proof}
  For any $(\vec x, \vec y)$, note that
  \[
    \bigl(\vec r(\mat A(\vec x, \vec y)), \vec c(\mat A(\vec x, \vec y))\bigr) = \sum_{i,j=1}^n A_{ij} e^{x_i + y_i} \vec \omega_{ij} = \sum_{i,j=1, \, A_{ij} > 0}^n A_{ij} e^{x_i + y_j}.
  \]
  Now choose for every $k > 0$ a pair $(\vec x^k, \vec y^k)$ such that
  \[
     \norm{\vec r(\mat A(\vec x^k, \vec y^k)) - \vec r}_1 + \norm{\vec c(\mat A(\vec x^k, \vec y^k)) - \vec c}_1 \leq \frac{1}{k}.
  \]
  Then as $k \to \infty$, we have
  \[
    \norm{\mat A(\vec x^k, \vec y^k)}_1 = \norm{\vec r(\mat A(\vec x^k, \vec y^k))}_1 \to \norm{\vec r}_1 = 1,
  \]
  so we also have
  \[
    \frac{\sum_{i,j = 1, A_{ij} > 0}^n A_{ij} e^{x_i^k + y_j^k} \vec\omega_{ij}}{\sum_{i,j = 1, A_{ij} > 0}^n A_{ij} e^{x_i^k + y_j^k}} = \frac{(\vec r(\mat A(\vec x^k, \vec y^k)), \vec c(\mat A(\vec x^k, \vec y^k)))}{\norm{\mat A(\vec x^k, \vec y^k)}_1} \to (\vec r, \vec c)
  \]
  as $k \to \infty$. As the left-hand side is in $\conv \Omega$, which is closed, we conclude that $(\vec r, \vec c) \in \conv \Omega$.
\end{proof}
Now, we show that~\ref{item:conv} implies~\ref{item:F}.
\begin{lem}
  \label{lem:F lower bound}
  Assume $(\vec r, \vec c) \in \R^n \times \R^n$ is in the convex hull of $\Omega$. Then $F$ is bounded from below by $\ln(\mu)$, where $\mu$ is the smallest non-zero $A_{ij}$.
\end{lem}
\begin{proof}
  Note that we have
  \begin{align*}
    F(\vec x, \vec y) & = \ln\left(\sum_{i,j=1}^n A_{ij} e^{x_i + y_j}\right) - \ip{\vec r}{\vec x} - \ip{\vec c}{\vec y} \\
    & = \ln\left(\sum_{i,j=1, A_{ij} > 0}^n A_{ij} e^{\ip{\vec \omega_{ij} - (\vec r, \vec c)}{(\vec x, \vec y)}}\right).
  \end{align*}
  Since $(\vec r, \vec c)$ is in the convex hull of $\Omega$, for every $(\vec x, \vec y)$ at least one of the inner products in the exponential is non-negative by Farkas' lemma (since it is impossible to separate $(\vec r, \vec c)$ from $\conv \Omega$).
  Therefore, $\sum_{i,j=1}^n A_{ij} e^{\ip{\vec \omega_{ij} - (\vec r, \vec c)}{(\vec x, \vec y)}} \geq \mu$.
\end{proof}
Lastly, we show that~\ref{item:F} implies~\ref{item:scalable}.
\begin{lem}
  Assume $F$ is bounded from below. Then $\mat A$ is asymptotically $(\vec r, \vec c)$-scalable.
\end{lem}
\begin{proof}
  The first observation is that the gradient of $F$ at $(\vec x, \vec y)$ is
  \[
    \grad F(\vec x, \vec y) = \frac{\sum_{i,j=1}^n A_{ij} e^{x_i + y_j} \vec \omega_{ij}}{\sum_{i,j=1}^n A_{ij} e^{x_i + y_j}} - (\vec r, \vec c).
  \]
  From this, one can show that $F$ is $L$-smooth with $L = 2$ (see, e.g., \cite[Lem.~3.10]{burgisser2020interior}).
  Now assume that $(\vec x, \vec y)$ are such that $F(\vec x, \vec y) \leq F^* + \delta$ for some $\delta > 0$, where $F^*$ is the infimum of~$F$.
  Then after a small gradient step $(\vec x', \vec y') = (\vec x, \vec y) - \frac{1}{L} \grad F(\vec x, \vec y)$, with a second-order Taylor expansion one obtains
  \begin{align*}
    F(\vec x', \vec y') - F(\vec x, \vec y) \leq - \frac{1}{L} \norm{\grad F(\vec x, \vec y)}_2^2 + \frac{1}{2L} \norm{\grad F(\vec x, \vec y)}_2^2 = -\frac{1}{2L} \norm{\grad F(\vec x, \vec y)}_2^2.
  \end{align*}
  Since also $F(\vec x', \vec y') - F(\vec x, \vec y) \geq F^* - F(\vec x, \vec y) \geq -\delta$, this implies that $\norm{\grad F(\vec x, \vec y)}_2^2/(2L) \leq \delta$.
  Therefore, if $(\vec x, \vec y)$ is such that $F(\vec x, \vec y) \leq F^* + \eps^2/2L$, then $\norm{\grad F(\vec x, \vec y)}_2^2\leq\eps^2$ and hence $\mat A(\vec x, \vec y)/\norm{\mat A(\vec x, \vec y)}_1$ is an $\eps$-$\ell_2$-scaling.
  Here $\eps>0$ is arbitrary, so it follows that $\mat A$ is asymptotically $(\vec r,\vec c)$-scalable.
\end{proof}
\begin{proof}[Proof of~\cref{thm:generalpotentialbound}]
The equivalence of the conditions follows from the previous lemmas.
To upper bound $f(\vec 0, \vec 0) - \inf_{(\vec x, \vec y)} f(\vec x, \vec y)$, note that we now have
\begin{align*}
    f(\vec 0, \vec 0) - \inf_{(\vec x, \vec y)} f(\vec x, \vec y) & = \norm{\mat A}_1 - \inf_{(\vec x, \vec y)} (1 + F(\vec x, \vec y)) \\
    & = \norm{\mat A}_1 - 1 -\inf_{(\vec x, \vec y)} F(\vec x, \vec y) \\
    & \leq \norm{\mat A}_1 - 1 -\ln(\mu) \\
    & = \norm{\mat A}_1 - 1 + \ln(1/\mu),
\end{align*}
where the first equality follows from \cref{cor:equality f F inf} and the inequality follows from \cref{lem:F lower bound}.
\end{proof}

\section{Proof of generalized Pinsker's inequality and Hellinger bound}\label{sec:generalized pinsker}
In this appendix we prove the following generalization of Pinsker's inequality (used in the analysis of the Sinkhorn algorithm) and a lower bound on the Hellinger distance (used in the analysis of the Osborne algorithm).

\lemGenPinsker*
\noindent
Note that we do not require $\vec b$ to be a probability distribution.
The proof we give is heavily inspired by arguments made in~\cite{klrs08}, which gave a lower bound on a relative entropy in terms of an $\ell_2$-distance.
We start with a lemma that verifies the properties of the function $w(\beta)$.

\begin{lem}
  \label{lem:generalized pinsker ingredient two}
  Let $w\colon(-1, \infty) \to \R$ be the function defined in~\cref{lem:generalized pinsker}.
  Then for $\beta \in [0, 1]$, we have $w(\beta) \geq \beta^2 / 4$.
  Furthermore, for $\beta \geq 1$, we have $w(\beta) \geq (1 - \ln 2) \beta$.
\end{lem}
\begin{proof}
  Set $g(\beta) = \beta^2 / 4$.
  We have $w'(\beta) = 1 - 1/(1 + \beta)$ and $g'(\beta) = 2 \beta/4$.
  On $[0, 1]$, we have
  \begin{align*}
    (1 + \beta) w'(\beta) = 1 + \beta - 1 \geq \beta (1 + \beta) / 2 = (1 + \beta)g'(\beta)
  \end{align*}
  and $w(0) = g(0)$, so we see that $w(\beta) \geq g(\beta) = \beta^2/4$ on $[0,1]$.
  For the last claim, note that $w(1) = 1 + \ln 2$ and $w'(1) = \tfrac12 > (1 - \ln 2)$, so by convexity of $w$ we have $w(\beta) \geq (1 - \ln 2) \beta$ for any $\beta \in (-1, \infty)$.
\end{proof}

Next we prove the main inequality that will imply the generalized Pinsker inequality.

\begin{lem}
  \label{lem:generalized pinsker ingredient one}
  Let $\vec a, \vec d \in \R^n$ be vectors such that $\norm{\vec a}_1 = 1$, $\vec a$ has positive entries, and $d_\ell > - a_\ell$ for every $\ell \in [n]$.
  Then
  \begin{align*}
    \sum_{\ell=1}^n d_\ell - a_\ell \ln\left(1 + \frac{d_\ell}{a_\ell}\right) \geq \norm{\vec d}_1 - \ln(1 + \norm{\vec d}_1).
  \end{align*}
\end{lem}
\begin{proof}
  The proof consists of two parts. First, we show that for any $a,d \in \R$ with $a > 0$ and $d > -a$, we have
  \begin{equation}\label{eq:genpinintermed}
    d - a \ln\left(1 + \frac{d}{a}\right) \geq \abs{d} - a \ln\left(1 + \frac{\abs{d}}{a}\right).
  \end{equation}
  Clearly this holds with equality if $d \geq 0$, so assume $d < 0$.
  The function
  \[
    g(\beta) = -2\beta - a \ln\left(1 - \frac{\beta}{a}\right) + a \ln\left(1 + \frac{\beta}{a}\right)
  \]
  satisfies $g(0) = 0$ and $g'(\beta) \geq 0$ on $[0, a)$, so $g(\beta) \geq 0$ on $[0, a)$. Setting $\beta = - d$ yields
  \[
    2 d - a \ln\left(1 + \frac{d}{a}\right) + a \ln\left(1 - \frac{d}{a}\right) \geq 0
  \]
  for $d < 0$, as desired.

  To finish the proof, we use \cref{eq:genpinintermed} and see that
  \begin{align*}
    \sum_{\ell = 1}^n d_\ell - a_\ell \ln\left(1 + \frac{d_\ell}{a_\ell}\right) & \geq \sum_{\ell = 1}^n \abs{d_\ell} - a_\ell \ln\left(1 + \frac{\abs{d_\ell}}{a_\ell}\right) \\
    & = \norm{\vec d}_1 - \sum_{\ell=1}^n a_\ell \ln\left(1 + \frac{\abs{d_\ell}}{a_\ell}\right).
  \end{align*}
  Since the function $t \mapsto \ln(1 + t)$ is concave and $\norm{\vec a}_1 = 1$ and $a_\ell > 0$, we see that
  \begin{align*}
    \sum_{\ell=1}^n a_\ell \ln\left(1 + \frac{\abs{d_\ell}}{a_\ell}\right) \leq \ln\left(1 + \sum_{\ell=1}^n a_\ell \cdot \frac{\abs{d_\ell}}{a_\ell}\right) = \ln(1 + \norm{\vec d}_1)
  \end{align*}
  and the desired result follows.
\end{proof}

\begin{proof}[Proof of~\cref{lem:generalized pinsker}.]
  By continuity, we may assume without loss of generality that $\vec a$ has positive entries.
  Recall that
  \begin{align*}
    D(\vec a \Vert \vec b) & = \sum_{\ell = 1}^n b_\ell - a_\ell + a_\ell \ln(\frac{a_\ell}{b_\ell}) \\
                           & = \sum_{\ell = 1}^n (b_\ell - a_\ell) - a_\ell \ln\left(\frac{b_\ell - a_\ell}{a_\ell} + 1\right).
  \end{align*}
  Set $\vec d = \vec b - \vec a$, so that $d_\ell = b_\ell - a_\ell > -a_\ell$ for every $\ell \in [n]$.
  Therefore, we may apply \cref{lem:generalized pinsker ingredient one} to $\vec a$ and $\vec d$ to get
  \begin{align*}
    D(\vec a \Vert \vec b) \geq w(\norm{\vec b - \vec a}_1).
  \end{align*}
  The claimed bounds on the function $w(\beta)$ follow from \cref{lem:generalized pinsker ingredient one}.
\end{proof}

We also prove the following lower bound on the Hellinger distance between two non-negative vectors that are not necessarily normalized.
\lemPinskerHell*
\begin{proof}
  Note that we have
  \begin{align*}
      \norm{\vec a - \vec b}_1^2 & = \Big(\sum_{\ell=1}^n \abs{a_\ell - b_\ell}\Big)^2 \\
       & = \Big(\sum_{\ell=1}^n \abs*{\sqrt{a_\ell} - \sqrt{b_\ell}} \cdot \abs*{\sqrt{a_\ell} + \sqrt{b_\ell}}\Big)^2 \\
       & \leq \norm*{\sqrt{\vec a} - \sqrt{\vec b}}_2^2 \cdot \norm*{\sqrt{\vec a} + \sqrt{\vec b}}_2^2
  \end{align*}
  where we used the Cauchy--Schwarz inequality in the last step.
  The bound then follows from
  \[
    \norm{\sqrt{\vec a} + \sqrt{\vec b}}_2^2 = \sum_{\ell = 1}^n (\sqrt{a_\ell} + \sqrt{b_\ell})^2 \leq 2 \sum_{\ell=1}^n (a_\ell + b_\ell) = 2 (\norm{\vec a}_1 + \norm{\vec b}_1).
  \]
  where the inequality follows from the arithmetic-geometric mean inequality.
\end{proof}

\section{Improved analysis of Sinkhorn algorithm for entrywise-positive matrices}\label{sec:positive}
In this section we give an improved runtime analysis of~\cref{alg:FSFP testing} where the goal is to produce an $\eps$-relative-entropy-scaling of an \emph{entrywise-positive} matrix $\mat A$ to arbitrary marginals~$\vec r, \vec c$.
In \cref{sec: full testing}, we showed that for general matrices $\mat A$ and arbitrary marginals $\vec r, \vec c$, Sinkhorn's algorithm takes $\widetilde O(1/\eps)$ iterations to find an $\eps$-relative-entropy-scaling of $\mat A$ to $(\vec r, \vec c)$ (in this section, $\widetilde O$ always suppresses polylogarithmic factors in $n$ and $1/\eps$).
Here we show that for entrywise-positive matrices~$\mat A$, the required number of iterations is in fact $\widetilde O(1/\sqrt{\eps})$.
We do so by mimicking the analysis of~\cite{klrs08}, where a similar result is derived for $\ell_2$-scaling.
We show that their analysis is robust with respect to only using estimates of marginals, and that it extends to the relative-entropy setting.

The analysis we give here improves upon~\cref{sec: full testing} as follows.
For entrywise-positive matrices one can show the scaling vectors $\vec x, \vec y$ produced by the Sinkhorn algorithm each have \emph{variation norm} ($x_{\max} - x_{\min}$) bounded by a \emph{constant}, whereas for arbitrary matrices we can only show a bound that is linear in the number of iterations (compare~\cref{lem:full support bounded variation} with~\cref{lem:scaling norm bound}).
The convexity of the potential $f$ can then be used, along with the previous fact, to determine a potential bound which becomes better as the scaling error goes down.
Combining this `adaptive' potential bound with Pinsker's inequality then allows one to show that once~$\mat A$ is $\eps$-relative-entropy-scaled for $\eps \leq 1$, it takes $\widetilde O(1/\sqrt{\eps})$ full Sinkhorn iterations to obtain an $\eps/2$-relative-entropy-scaling.
Obtaining a relative-entropy-scaling with constant error takes a constant number of Sinkhorn iterations, and from there onwards it suffices to halve the scaling error at most $\log_2(1/\eps)$ times, where the number of iterations required to halve the scaling error increases by a factor $\sqrt{2}$ every time.
Carefully keeping track of the total number of iterations then gives a total iteration count of $\widetilde O(1/\sqrt{\eps})$.

As in the setting of \cref{sec: full testing}, let $\vec r, \vec c \in \R_+^n$ with $\norm{\vec r}_1 = \norm{\vec c}_1 = 1$.
Let $\ip{\cdot}{\cdot}$ denote the standard inner product in $\R^n$.

The following lemma is an $\ell_1$-analog of \cite[Lem.~6.1]{klrs08}.
\begin{lem}
  \label{lem:full support l1 potential gap}
  Let $\mat A \in \R_+^{n \times n}$ and let $(\vec x, \vec y), (\vec x^*, \vec y^*) \in \R^n \times \R^n$.
  If $(\vec x, \vec y)$ is such that $\vec c(\A x y) = \vec c$, then
  \[
    f(\vec x, \vec y) - f(\vec x^*, \vec y^*) \leq \norm{\vec r(\A x y) - \vec r}_1 (x_{\max} - x_{\min} + x_{\max}^* - x_{\min}^*).
  \]
  A similar statement holds if $\vec r(\A x y) = \vec r$.
\end{lem}
\begin{proof}
  We have
  \begin{align*}
    \grad_{\vec x} f(\vec x, \vec y) & = \vec r(\A x y) - \vec r, \\
    \grad_{\vec y} f(\vec x, \vec y) & = \vec c(\A x y) - \vec c,
  \end{align*}
  so in particular, if $\vec c(\A x y) = c$, then $\grad_{\vec y} f(\vec x, \vec y) = 0$.
  Now, by convexity of $f$, we have
  \[
    f(\vec x, \vec y) + \ip{\grad_{\vec x} f(\vec x, \vec y)}{\vec x^* - \vec x} \leq f(\vec x^*, \vec y^*),
  \]
  which we rearrange as
  \begin{equation}\label{eq:convex rearranged}
    f(\vec x, \vec y) - f(\vec x^*, \vec y^*)
  \leq \ip{\grad_{\vec x} f(\vec x, \vec y)}{\vec x - \vec x^*}
  = \ip{\vec r(\A x y) - \vec r}{\vec x - \vec x^*}.
  \end{equation}
  Since $\vec c(\A x y) = \vec c$ and $\|\vec c\|_1=1$, we have $\norm{\A x y}_1 = 1$ and $\norm{\vec r(\A x y)}_1 = 1$ as well. In particular,
  \[
    \ip{\vec r(\A x y) - \vec r}{\vec 1} = 0.
  \]
  Set $\vec z = \vec x - \frac{\ip{\vec x}{\vec 1}}{n} \vec 1$ and $\vec z^* = \vec x^* - \frac{\ip{\vec x^*}{\vec 1}}{n} \vec 1$.
  Then, using \cref{eq:convex rearranged},
  \begin{align*}
    f(\vec x, \vec y) - f(\vec x^*, \vec y^*)
    & \leq \ip{\vec r(\A x y) - \vec r}{\vec x - \vec x^*} \\
    & = \ip{\vec r(\A x y) - \vec r}{\vec z - \vec z^*} \\
    & \leq \norm{\vec r(\A x y) - \vec r}_1 \left( \norm{\vec z}_\infty + \norm{\vec z^*}_\infty \right) \\
    & \leq\norm{\vec r(\A x y) - \vec r}_1 (z_{\max} - z_{\min} + z_{\max}^* - z_{\min}^*) \\
    & =\norm{\vec r(\A x y) - \vec r}_1 (x_{\max} - x_{\min} + x_{\max}^* - x_{\min}^*)
  \end{align*}
  as desired.
  The last inequality holds because the entries of $\vec z$ and of $\vec z^*$ sum to zero.
\end{proof}

In the next lemma, we provide an analog of \cite[Lem.~6.2]{klrs08}, which shows that the vectors $\vec x$ and $\vec y$ produced by full Sinkhorn iterations for entrywise-positive $\mat A$ have bounded variation norm.
This is the only part of the analysis which requires entrywise positivity.
\begin{lem}\label{lem:full support bounded variation}
  Let $0 < \mu < \nu \leq 1$ and assume $\mat A \in [\mu, \nu]^{n \times n}$ and $\vec r \in \R_+^n$ strictly positive.
  Let $\vec y \in \R^n$, let~$\delta \geq 0$, and let $\vec x' \in \R^n$ be such that $\abs{x_\ell' - \ln(r_\ell / \sum_{j=1}^n A_{\ell j} e^{y_j})} \leq \delta$ for all $\ell \in [n]$.
  Then
  \[
    x_{\max}' - x_{\min}' \leq 2 \delta + \ln \frac{\nu}{\mu} + \ln \frac{r_{\max}}{r_{\min}}.
  \]
  The analogous statement holds for the column-scaling vectors after a column update.
\end{lem}
\noindent In other words, the variation in the row-scaling vectors after a $\delta$-approximate Sinkhorn update is bounded above by $2\delta$ plus a quantity depending only on $\mat A$ and $\vec r$.
\begin{proof}
  For any $\ell \in [n]$, we have $\abs{x_\ell' - \ln (r_\ell / \sum_{j=1}^n A_{\ell j} e^{y_j})} \leq \delta$.
  By using the upper and lower bound on the entries of $\mat A$, we obtain
  \begin{align*}
    &\ln\left(\frac{r_{\max}}{\mu \sum_{j = 1}^n e^{y_j}}\right) + \delta
    \geq \ln\left(\frac{r_\ell}{\sum_{j = 1}^n A_{\ell j} e^{y_j}}\right) + \delta \\
    \geq x_\ell' \geq &\ln\left(\frac{r_\ell}{\sum_{j = 1}^n A_{\ell j} e^{y_j}}\right) - \delta
    \geq \ln\left(\frac{r_{\min}}{ \nu \sum_{j = 1}^n e^{y_j}}\right) - \delta.
  \end{align*}
  Therefore, for any $k, \ell \in [n]$, we obtain
  \begin{align*}
    x'_k - x'_{\ell} \leq \ln\left(\frac {r_{\max}} {\mu \sum_{j=1}^n e^{y_j}} \right) + \delta - \ln\left(\frac {r_{\min}} {\nu \sum_{j=1}^n e^{y_j}} \right) + \delta = 2 \delta + \ln \frac{r_{\max}}{r_{\min}} + \ln \frac{\nu}{\mu}
  \end{align*}
  as desired.
\end{proof}

Note that the previous proof fails if $\mat A$ does not have full support; one can still attempt to use an upper and lower bound on the non-zero entries of $\mat A$, but the support in the $k$-th and $\ell$-th rows generally differ, so the corresponding (logarithms of) sums of column-scaling vectors do not necessarily cancel.
We now state a useful corollary of \cref{lem:full support bounded variation}.

\begin{cor}
  \label{cor:fs bounded variation fixed point}
  Let $\mat A \in [\mu, \nu]^{n \times n}$ and let $(\vec x^*, \vec y^*) \in \R^n \times \R^n$ be such that $\mat A(\vec x^*, \vec y^*)$ is exactly $(\vec r, \vec c)$-scaled.
  Then
  \[
    x_{\max}^* - x_{\min}^* \leq \ln \frac{\nu}{\mu} + \ln \frac{r_{\max}}{r_{\min}},
  \]
  and
  \[
    y_{\max}^* - y_{\min}^* \leq \ln \frac{\nu}{\mu} + \ln \frac{c_{\max}}{c_{\min}}.
  \]
\end{cor}
\begin{proof}
  Since $(\vec x^*, \vec y^*)$ exactly scale $\mat A$ to $(\vec r, \vec c)$, an exact full Sinkhorn update does not change the scaling vectors, so it holds that
  $x^*_\ell = \ln(r_\ell / \sum_{j=1}^n A_{\ell j} e^{y^*_j})$ for all $\ell\in[n]$.
  Thus, $\vec x^*$ satisfies the assumptions of~\cref{lem:full support bounded variation} with $\delta = 0$ and the first bound follows.
  The second bound is proved completely analogously.
\end{proof}

To deal with updates with finite precision, we adapt \cref{lem:full support l1 potential gap} to the case where $\vec c(\A x y)$ and $\vec c$ are only approximately equal.

\begin{lem}
  \label{lem:full support l1 potential gap fp}
  Let $(\vec x, \vec y) \in \RR^n \times \RR^n$ and $\delta \in [0, 1/2]$ be such that $\abs{y_\ell - \ln(c_\ell / \sum_{i=1}^n A_{i \ell} e^{x_i})} \leq \delta$.
  Then, for any $(\vec x^*, \vec y^*) \in \RR^n \times \RR^n$,
  \[
    f(\vec x, \vec y) - f(\vec x^*, \vec y^*) \leq \delta + \left( \norm{\vec r(\A x y) - \vec r}_1 + 2 \delta \right)(x_{\max} - x_{\min} + x_{\max}^* - x_{\min}^*).
  \]
\end{lem}
\begin{proof}
  Let $\vec y'$ be the vector defined by
  \[
    y_\ell' = \ln\left( \frac {c_\ell} {\sum_{i=1}^n A_{i \ell} e^{x_i}} \right),
  \]
  i.e., $\vec y'$ is the vector of column-scaling vectors after an exact Sinkhorn column update starting from~$(\vec x, \vec y)$.
  Then
  \begin{equation}\label{eq:intermed ful gap 2}
    f(\vec x, \vec y) - f(\vec x, \vec y') = D(\vec c \Vert \vec c(\A x y)) \leq \delta,
  \end{equation}
  where the first inequality is \cref{eq:exact row update} and the second inequality follows from the assumption on $\vec y$ (see proof of \cref{lem:other marginals full}; the event $S_t$ corresponds precisely to the assumption on $\vec y$).

  Furthermore, $\vec c(\mat A(\vec x, \vec y')) = \vec c$, so we may apply \cref{lem:full support l1 potential gap} with $(\vec x,\vec y')$ and $(\vec x^*,\vec y^*)$ to obtain
  \begin{equation}\label{eq:intermed ful gap}
    f(\vec x, \vec y') - f(\vec x^*, \vec y^*) \leq \norm{\vec r(\mat A(\vec x, \vec y')) - \vec r}_1 (x_{\max} - x_{\min} + x_{\max}^* - x_{\min}^*).
  \end{equation}
  Since $y_\ell' - y_{\ell} \in [-\delta, \delta]$ for every $\ell \in [n]$, for every $i, j \in [n]$ we have
  \[
    A_{ij} e^{x_i + y_j} \in [e^{-\delta} A_{ij}e^{x_i + y_j'}, e^{\delta} A_{ij}e^{x_i + y_j'}].
  \]
  Since $\delta \leq 1/2$, we can use the estimates $e^{-\delta} \geq 1 - 2 \delta$ and $e^{\delta} \leq 1 + 2 \delta$, which imply that
  \[ r_\ell(\A x y) \in \left[ (1 - 2 \delta) r_\ell(\mat A(\vec x, \vec y')),(1 + 2 \delta) r_\ell(\mat A(\vec x, \vec y')) \right] \]
  for every $\ell \in [n]$.
  By the triangle inequality we get
  \begin{align*}
    \norm{\vec r(\mat A(\vec x, \vec y')) - \vec r}_1 & \leq
                                                        \norm{\vec r(\mat A(\vec x, \vec y')) - \vec r(\A x y)}_1 +
                                                        \norm{\vec r(\A x y) - \vec r}_1 \\
                                                      & \leq 2 \delta \norm{\vec r(\mat A(\vec x, \vec y'))}_1 + \norm{\vec r(\A x y) - \vec r}_1 \\
                                                      & = 2 \delta + \norm{\vec r(\A x y) - \vec r}_1.
  \end{align*}
  where the last equality holds since $\norm{\vec r(\mat A(\vec x, \vec y'))}_1 = \norm{\vec c(\mat A(\vec x, \vec y'))}_1 = \norm{\vec c}_1 = 1$.
  If we plug this into \cref{eq:intermed ful gap} then together with \cref{eq:intermed ful gap 2} the proof is complete.
\end{proof}

We now combine these results to obtain an adaptive potential bound for iterations produced by the Sinkhorn algorithm.

\begin{cor}
  \label{cor:fs adaptive potential bound}
  Let $A \in [\mu, \nu]^{n \times n}$, let $t \geq 1$, and let $x^{(t)}$ and $y^{(t)}$ be as in \cref{alg:FSFP testing}, and assume no call to {\ApproxScalingFactor} has failed.
  If $t$ is even, then we have
  \[
    f(\vec x^{(t)}, \vec y^{(t)}) - f^* \leq \delta + 2 \left( \norm{\vec r(\mat A(\vec x^{(t)},\vec y^{(t)}) - \vec r}_1 + 2 \delta \right) \left( \delta + \ln \frac{r_{\max}}{r_{\min}} + \ln \frac{\nu}{\mu} \right)
  \]
  where $f^* = \inf_{\vec x, \vec y} f(\vec x, \vec y)$.
  A similar statement holds if $t$ is odd.
\end{cor}
\begin{proof}
  It is well-known that any entrywise-positive matrix is exactly scalable to arbitrary strictly positive $(\vec r,\vec c)$.
  Thus, there exist $(\vec x^*, \vec y^*)$ that exactly $(\vec r, \vec c)$-scale $\mat A$.
  By \cref{cor:fs bounded variation fixed point}, we have $x_{\max}^* - x_{\min}^* \leq \ln \frac{r_{\max}}{r_{\min}} + \ln \frac{\nu}{\mu}$.
  Furthermore, as we assume no call to \ApproxScalingFactor fails, we also have, by~\cref{lem:full support bounded variation},
  \[
    x_{\max}^{(t)} - x_{\min}^{(t)} \leq 2 \delta + \ln \frac{r_{\max}}{r_{\min}} + \ln \frac{\nu}{\mu}.
  \]
  The result now follows from applying \cref{lem:full support l1 potential gap fp} and a simple estimate.
\end{proof}

The main result of this section is the following theorem, which adapts~\cite[Thm.~6.1]{klrs08} from the $\ell_2$-distance to the relative-entropy setting, and assumes only $\delta$-additive precision in the updates.
The proof follows the same strategy.
Note that the adaptive potential bound (\cref{cor:fs adaptive potential bound}) is stated in terms of the $\ell_1$-distance, which we can further upper bound in terms of the relative entropy using Pinsker's inequality.

\begin{thm} \label{thm: faster for positive}
  Let $0 < \mu \leq \nu \leq 1$, let $\mat A \in [\mu, \nu]^{n \times n}$ with $\norm{\mat A}_1 \leq 1$, let $\vec r, \vec c \in (0, 1]^n$ with $\norm{\vec r}_1 = \norm{\vec c}_1 = 1$, and let $\eps \in (0, 1]$.
  Choose
  \[
    T = \left\lceil\frac{32\ln(1 / \mu) + \log_2(2/\eps) (1+34C)}{\sqrt{\eps}}\right\rceil,
  \]
  $\delta = \eps / 64$,
  $\delta' = \eps / 2$,
  $\eta = 1/(3(n+1)T)$,
  $b_1 = \lceil \log_2(T) + \log_2(\ln( \frac1{\mu}) + 1 + \sigma) \rceil$,
  $b_2 = \lceil\log_2(1/\delta)\rceil$,
  where $C = \delta + \ln(r_{\max}/r_{\min}) + \ln(c_{\max}/c_{\min}) + \ln(\nu / \mu)$
  and $\sigma = \max\{\abs{\ln r_{\min}}, \abs{\ln c_{\min}}\}$.
  Then \cref{alg:FSFP testing} with these parameters returns vectors $(\vec x, \vec y) \in \R^n \times \R^n$ such that, with probability $\geq 2/3$, $\A x y$ is $\eps$-relative-entropy-scaled to~$(\vec r, \vec c)$.

  The resulting classical algorithm has time complexity $\widetilde O(n^2 / \sqrt{\eps})$, while the corresponding quantum algorithm has time complexity $\widetilde O(n^{1.5} / \eps^{1.5})$.
\end{thm}

\begin{proof}
  Observe first that we have chosen $b_1, b_2$ such that the guarantees of \ApproxScalingFactor and \TestScaling are satisfied at any iteration (cf.~\cref{lem:scaling norm bound}).
  Throughout the proof, we assume that all calls to \ApproxScalingFactor and to \TestScaling made by \cref{alg:FSFP testing} succeed, which by the choice of $\eta$ happens with probability $\geq 2/3$.
  As always, we write $f^* = \inf_{\vec x, \vec y} f(\vec x, \vec y)$, and we abbreviate $f_t=f(\vec x^{(t)},\vec y^{(t)})$, $\vec r^{(t)} = \vec r(\mat A(\vec x^{(t)}, \vec y^{(t)}))$, and $\vec c^{(t)} = \vec c(\mat A(\vec x^{(t)}, \vec y^{(t)}))$ for the potential and the row and column marginals after the $t$-th iteration.

  The strategy is as follows.
  For $t \geq 0$, define $\eps_t = D(\vec r \Vert \vec r^{(t)})$ if $t$ is even, and $\eps_t = D(\vec c \Vert \vec c^{(t)})$ if $t$ is odd.
  By \cref{lem:other marginals full}, the other of the two relative entropies is at most $\delta$ for every $t \geq 1$, and so it suffices to bound the time until $\eps_t$ is sufficiently small.
  We will first bound the number of iterations until $\eps_t \leq 1-\ln 2$, and subsequently bound the number of iterations required for $\eps_t$ to halve.

  We first argue that there exists an $N \leq 32 \ln(1/\mu)$ such that $\eps_N \leq 1 - \ln 2$.
  Suppose for contradiction that this is not the case.
  Then for $t = 0, \dotsc, \lfloor 32 \ln(1 / \mu) \rfloor$, we have $\eps_t > 1 - \ln 2 \geq 1/4$, and by~\cref{lem:potential gap,lem:progress no failure full}, we see that
  \begin{align*}
      \ln(1/\mu)
    & \geq f_0 - f^*
      \geq f_0 - f_{\lfloor 32 \ln(1/\mu) \rfloor + 1}
    = \sum_{t=0}^{\lfloor 32 \ln(1/\mu) \rfloor } \left( f_t - f_{t+1} \right)
    \geq \sum_{t=0}^{\lfloor 32 \ln(1/\mu) \rfloor } \left(\eps_t - 2 \delta \right) \\
    & > \left( \lfloor 32 \ln(1/\mu)\rfloor + 1 \right) \left((1 - \ln 2) - 2 \delta\right)
    \geq 32 \ln(1/\mu) \cdot \frac{1}{16}
    = 2 \ln(1/\mu)
  \end{align*}
  where we used $2 \delta \leq 1/32$. This is the desired contradiction.

  We now bound the halving time.
  Let $t \geq 1$ be such that $\eps_t \leq 1 - \ln 2$, and define
  \begin{equation}\label{eq:def N_t}
    N_t = \inf \{ \tau \geq 0 \;:\; \eps_{t + \tau} \leq \eps_t / 2 \}.
  \end{equation}
  For $\tau = 0, \dotsc, N_t-1$, we have $\eps_{t+\tau} > \eps_t / 2$, and so
  \[
    f_t - f^*
  \geq f_t - f_{N_t}
  \geq \sum_{\tau=0}^{N_t-1} \left( f_{t+\tau} - f_{t+\tau+1} \right)
  \geq \sum_{\tau=0}^{N_t-1} \left( \eps_{t+\tau} - 2 \delta \right)
  \geq N_t \left( \frac {\eps_t} 2 - 2 \delta \right).
  \]
  again by \cref{lem:potential gap,lem:progress no failure full}.
  Therefore, so long as $\eps_t > 4\delta$, we obtain
  \begin{equation}\label{eq:N_t first bound}
    N_t \leq \frac{f_t - f^*}{\eps_t / 2 - 2 \delta}.
  \end{equation}
  We first prove a bound on $N_t$ assuming $8\delta\leq \eps_t\leq 1-\ln 2$.
  For $t$ even, \cref{lem:generalized pinsker} then implies that $\norm{\vec r - \vec r^{(t)}}_1 \leq 1$ (since $\eps_t = D(\vec r\Vert\vec r^{(t)}) \leq 1-\ln 2$, while the function $w(\alpha)$ is strictly larger than $1-\ln 2$ for $\alpha > 1$) and hence $D(\vec r\Vert\vec r^{(t)}) \geq \norm{\vec r - \vec r^{(t)}}_1^2 / 4$, so
  \begin{align*}
      N_t 
    \leq \frac{f_t - f^*}{\eps_t / 2 - 2 \delta}
    &\leq \frac{\delta + 2 \left( \norm{\vec r^{(t)} - \vec r}_1 + 2 \delta \right) C} {\eps_t / 2 - 2 \delta} \\
    &\leq \frac{\delta + 4 \left( \sqrt{\eps_t} + \delta \right) C} {\eps_t / 2 - 2 \delta} \\
    &= \frac{2 \delta (1 + 4C) + 8 \sqrt{\eps_t} C} {\eps_t - 4 \delta} \\
    &\leq \frac{(\eps_t/4) (1 + 4C) + 8 \sqrt{\eps_t} C} {\eps_t/2}
    = \frac {1 + 4C} 2 + \frac {16 C} {\sqrt{\eps_t}},
  \end{align*}
  where the first inequality is \cref{eq:N_t first bound}, the second follows from \cref{cor:fs adaptive potential bound}, and in the last inequality we assume that $\eps_t \geq 8\delta$.
  The same inequality holds for $t$ odd, with an analogous proof.
  Thus, we have proved that, for any $t$ such that $8\delta \leq \eps_t \leq 1-\ln 2$,
  \begin{equation}\label{eq:N_t bound}
      N_t \leq \frac {1 + 4C} 2 + \frac {16 C} {\sqrt{\eps_t}}.
  \end{equation}

  We now combine the preceding to verify that the desired number of iterations suffices for \cref{alg:FSFP testing} to return an $\eps$-relative-entropy-scaling.
  For $s\geq0$, define
  \[
    h_s = \min \left\{ t \geq 0 : \eps_t \leq \frac{1 - \ln 2}{2^{s}} \right\}.
  \]
  We proved above that $h_0 \leq 32 \ln(1/\mu)$.
  Clearly, the sequence $h_s$ is non-decreasing.
  Let
  \[
    S = \min \left\{ s\geq0 : \frac{1 - \ln 2}{2^{s}} \leq 32\delta = \delta' \right\}.
  \]
  Note that the algorithm will necessarily return within the first $h_S$ iterations with an $\eps$-relative-entropy-scaling.
  Indeed, either \TestScaling returns \True during one of the first~$h_S-1$ iterations, or it must return \True in the $h_S$-th iteration, since then $\eps_{h_S} \leq \delta'$ (and the other relative entropy is always at most~$\delta \leq \delta'$).
  Thus it suffices to bound $h_S$.
  For any $0 \leq s < S$, if $h_{s+1} > h_s$ we have $\eps_{h_s} > \frac{1 - \ln 2}{2^{s+1}} > 16\delta > 8\delta$, so
  \begin{align}\label{eq:h_s step}
    h_{s+1} - h_s
  \leq N_{h_s}
  \leq \frac{1+4C}2 + \frac{16C}{\sqrt{\eps_{h_s}}}
  < \frac{1+4C}2 + 2^{s/2} \frac{16\sqrt2C}{\sqrt{1 - \ln 2}}
  \end{align}
  where the first inequality holds by definition of $N_t$ in \cref{eq:def N_t} and the second inequality is \cref{eq:N_t bound};
  the latter is applicable since $8\delta \leq \eps_{h_s} \leq 1-\ln 2$.
  Clearly \cref{eq:h_s step} also holds if $h_{s+1} = h_s$.
  Thus we can upper bound the total number of iterations required by
  \begin{align*}
      h_S
    &= h_0 + \sum_{s=0}^{S-1} \left( h_{s+1} - h_s \right) \\
    &< 32 \ln(1/\mu) + \sum_{s=0}^{S-1} \left( \frac{1+4C}2 + 2^{s/2} \frac{16\sqrt2C}{\sqrt{1 - \ln 2}} \right) \\
    &\leq 32 \ln(1/\mu) + S \left( \frac{1+4C}2 + 2^{(S-1)/2} \frac{16\sqrt2C}{\sqrt{1 - \ln 2}} \right) \\
    &\leq 32 \ln(1/\mu) + S \left( \frac{1+4C}2 + \frac{\sqrt{2(1 - \ln 2)}}{\sqrt{\delta'} \cdot \sqrt{2}} \frac{16\sqrt2 C}{\sqrt{1 - \ln 2}} \right) \\
    &\leq 32 \ln(1/\mu) + \log_2\bigl(2(1-\ln 2) / \delta'\bigr) \left( \frac{1+4C}2 + \frac{32C}{\sqrt{2\delta'}} \right) \\
    &\leq 32 \ln(1/\mu) + \log_2\bigl(2 / \eps\bigr) \left( \frac12 + 2C + \frac{32C}{\sqrt{\eps}} \right) \\
    &\leq \frac{32 \ln(1/\mu) + \log_2\bigl(2 / \eps\bigr) \left( 1 + 34C \right)}{\sqrt\eps},
  \end{align*}
  where we used that $2^S < 2(1-\ln 2) / \delta'$ by definition of~$S$ (noting that $S\geq1$), as well as $\delta'=\eps/2$ and $\eps\leq1$.
\end{proof}

\end{document}